\documentclass[journal,a4,twocolumn]{IEEEtran}
\usepackage{amsmath,graphicx}
\usepackage{color}
\usepackage{graphicx}
\usepackage{epstopdf}
\usepackage{amsmath}
\usepackage{amssymb}
\usepackage{mathrsfs}
\usepackage{hyperref}
\usepackage{tikz}
\usepackage{bm}
\usepackage[english]{babel}
\usepackage{cite}
\usepackage{rotfloat}
\usepackage{mathtools}
\usepackage[font=normalsize,labelfont=bf]{caption}
\usepackage{amsmath}
\usepackage{makecell}
\usepackage[ruled, vlined, linesnumbered]{algorithm2e}
\usepackage{multirow}
\usepackage{subfigure}
\usepackage{booktabs}
\usepackage{colortbl}
\usepackage{multirow}
\usepackage{hhline}
\usepackage{stfloats}
\usepackage{multicol}
\usepackage{bbm}
\usepackage{cases}
\graphicspath{ {Figures/} }
\newcommand{\T}{{\scriptscriptstyle\mathsf{T}}}
\renewcommand{\H}{{\scriptscriptstyle\mathsf{H}}}

\newsavebox{\foobox}

\definecolor{kugray5}{RGB}{224,224,224}

\usepackage[normalem]{ulem}
\newcommand\rsout{\bgroup\markoverwith
	{\textcolor{red}{\rule[0.5ex]{2pt}{0.8pt}}}\ULon}



\makeatletter
\newcommand{\ALOOP}[1]{\ALC@it\algorithmicloop\ #1%
	\begin{ALC@loop}}
	\newcommand{\ENDALOOP}{\end{ALC@loop}\ALC@it\algorithmicendloop}

\makeatother

\usepackage{etoolbox}
\let\mybibitem\bibitem
\renewcommand{\bibitem}[1]{%
	\ifstrequal{#1}{nature}
	{\color{blue}\mybibitem{#1}}
	{\color{black}\mybibitem{#1}}%
}

\graphicspath{ {Figures/} }

\newtheorem{theorem}{\textbf{Theorem}}
\newtheorem{remark}{\textbf{Remark}}

\newtheorem{corollary}{\textbf{Corollary}}
\newtheorem{proof}{Proof}

\DeclareCaptionLabelSeparator{periodspace}{.\quad}

\renewcommand{\figurename}{Fig.}
\addto\captionsenglish{\renewcommand{\figurename}{Fig.}}
\newcommand\nbthis{\addtocounter{equation}{1}\tag{\theequation}}

\newcommand{\norm}[1]{\left\lVert#1\right\rVert} 
\newcommand{\abs}[1]{\left|#1\right|} 

\newcommand{\tr}[1]{\mathtt{tr}\left(#1\right)} 

\newcommand{\diag}[1]{\mathtt{diag}\left\{#1\right\}} 

\newcommand{\re}[1]{\mathtt{Re}\left(#1\right)}
\newcommand{\im}[1]{\mathtt{Im}\left(#1\right)}
\allowdisplaybreaks

\newcommand{\mean}[1]{\mathbb{E} \left\{#1\right\}}

\usepackage{setspace}

\newcommand{\mR}{{\mathbf{R}}}
\newcommand{\mH}{{\mathbf{H}}} 

\newcommand{\mA}{{\mathbf{A}}}
\newcommand{\mS}{{\mathbf{S}}}
\newcommand{\mW}{{\mathbf{W}}}

\newcommand{\mI}{\textbf{\textbf{I}}}

\newcommand{\mX}{{\mathbf{X}}}
\newcommand{\mY}{{\mathbf{Y}}}

\newcommand{\mF}{{\mathbf{F}}}

\newcommand{\mN}{{\mathbf{N}}}

\newcommand{\mJ}{{\mathbf{T}}}

\newcommand{\setC}{\mathbb{C}} 
\newcommand{\setR}{\mathbb{R}}

\newcommand{\ve}{{\mathbf{e}}} 
\newcommand{\vs}{{\mathbf{s}}}
\newcommand{\vx}{{\mathbf{x}}}
\newcommand{\vy}{{\mathbf{y}}}

\newcommand{\vn}{{\mathbf{n}}}
\newcommand{\vu}{{\mathbf{u}}}
\newcommand{\vz}{{\mathbf{z}}} 
\newcommand{\vh}{{\mathbf{h}}}

\newcommand{\vb}{{\mathbf{b}}}
\newcommand{\vw}{{\mathbf{w}}}
\newcommand{\va}{{\mathbf{a}}}

\newcommand{\vp}{{\mathbf{p}}}

\newcommand{\vj}{{\mathbf{t}}}

\newcommand{\vf}{\mathbf{f}}

\def\bGamma{{\pmb{\Gamma}}}

\def\b0{{\pmb{0}}}

\newcommand{\Nr}{N_\mathtt{r}}
\newcommand{\Nt}{N_\mathtt{t}}

\newcommand{\Nth}{N_\mathtt{ty}}
\newcommand{\Ntv}{N_\mathtt{tz}}

\newcommand{\ah}{\va_\mathtt{y}} 
\newcommand{\av}{\va_\mathtt{z}}

\newcommand{\sigmac}{\sigma_{\mathtt{c}}^2}
\newcommand{\sigmas}{\sigma_{\mathtt{s}}^2}

\newcommand{\noise}{\sigma^2}

\newcommand{\Jtt}{T_{\theta \theta}}
\newcommand{\Jtp}{T_{\theta \phi}}
\newcommand{\Jta}{\vj_{\theta \tilde{\bm{\beta_\mathtt{s}}}}}
\newcommand{\Jpp}{T_{\phi \phi}}
\newcommand{\Jpa}{\vj_{\phi \tilde{\bm{\beta_\mathtt{s}}}}}
\newcommand{\Jaa}{\mJ_{\tilde{\bm{\beta_\mathtt{s}}}\tilde{\bm{\beta_\mathtt{s}}}}}

\newcommand{\xibf}{\bm{\xi}_{\mathtt{bf}}}
\newcommand{\lambdabf}{{\lambda_{\mathtt{bf}}}} 
\newcommand{\zetabf}{{\bm{\zeta}_{\mathtt{bf}}}}
\newcommand{\adotphi}{\dot{\va}_{\phi}}
\newcommand{\adottheta}{\dot{\va}_{\theta}}
\newcommand{\bdotphi}{\dot{\vb}_{\phi}}
\newcommand{\bdottheta}{\dot{\vb}_{\theta}}
\newcommand{\vahat}{\va (\hat{\theta}, \hat{\phi} )}

\usepackage{enumitem}
\usepackage{marginnote}
\begin{document}
        \setlength{\abovedisplayskip}{4.5pt}
	\setlength{\belowdisplayskip}{4.5pt}
	\title{\huge{Massive MIMO ISAC Under Target-Angle Uncertainty:
CRLB Outage Analysis and Robust Resource Allocation}}
	\author{
		Smriti Uniyal, \IEEEmembership{Graduate Student Member, IEEE}, Tianyu Fang, \IEEEmembership{Graduate Student Member, IEEE}, \\Van-Dinh~Nguyen, \IEEEmembership{Senior Member, IEEE}, 
		Hien~Quoc~Ngo, \IEEEmembership{Fellow, IEEE},\\ Markku Juntti, \IEEEmembership{Fellow, IEEE}, and  Nhan~Thanh~Nguyen, \IEEEmembership{Senior Member, IEEE}\vspace{-0.5cm}
	}
	\maketitle
	
\begin{abstract} 
In integrated sensing and communications (ISAC), the same spectral and hardware resources are shared for two functionalities. Most ISAC designs assume perfect target-angle information neglecting angle estimation errors, which introduce steering-vector mismatches, degrade sensing accuracy, and may invalidate deterministic sensing guarantees. This paper investigates monostatic massive multiple-input-multiple-output (MIMO) ISAC systems under imperfect target-angle estimates. We derive closed-form expressions for the Cramér-Rao lower bounds (CRLBs) of target azimuth and elevation estimates in the presence of angle uncertainty. We  characterize the cumulative distribution functions and outage probabilities of the CRLBs under Gaussian, generalized uniform, and von Mises angle-error models. Our analysis reveals that, in the small-error regime, the CRLBs increase quadratically with the angle errors due to transmit steering-vector mismatch. To ensure reliable sensing, we propose a robust power allocation framework that jointly optimizes pilot training and communications/sensing transmission powers to maximize the communications sum rate while satisfying CRLB outage constraints. The resulting nonconvex problem is solved using an alternating-optimization algorithm based on successive convex approximation. Numerical results validate the developed analysis and show that the proposed robust design reduces azimuth and elevation CRLB outage probabilities by up to $60\%$ compared with conventional non-robust schemes. It attains up to $45\%$ higher sum rates than the non-robust design under strict CRLB thresholds.
\end{abstract}
	
\begin{IEEEkeywords}
	CRLB outage probability, integrated sensing and communications, massive MIMO, robust resource allocation, sensing reliability, target-angle uncertainty.
\end{IEEEkeywords}
\IEEEpeerreviewmaketitle
	
\section{Introduction}

\IEEEPARstart{I}{ntegrated} sensing and communications (ISAC) is a key technology for future wireless systems, enabling simultaneous communications and sensing using shared spectral and hardware resources~\cite{liu2021cramer}. Equipped with large antenna arrays, massive multiple-input multiple-output (mMIMO) systems are particularly attractive for ISAC due to their high spatial resolution and beamforming gains, which benefit both communications and target sensing~\cite{liao,nhan_isac_power_allocation}. In particular, angular information, such as angle-of-departure and angle-of-arrival, plays a fundamental role in sensing, localization, and positioning systems, and has been incorporated into standardized frameworks such as 3GPP NR~\cite{chen2026mismatch}.
Consequently, accurate angle information is essential for beamforming design and sensing-performance evaluation in mMIMO-ISAC systems.

Most existing ISAC designs assume perfect target-angle information at the base station (BS). In practice, however, target angles are estimated from noisy radar observations and are subject to array calibration errors, synchronization impairments, and model mismatches~\cite{li2008target,he2010crlb,nguyen2011modeling}. These imperfections induce steering-vector mismatches between the presumed and actual target directions, which can significantly degrade sensing accuracy and distort performance metrics such as the Cramér--Rao lower bound (CRLB)~\cite{jia2024imperfect,chen2026mismatch}. Since ISAC systems jointly allocate resources to communications and sensing tasks, such sensing degradation also affects the communications--sensing tradeoff. This motivates the development of robust resource-allocation strategies that explicitly account for target-angle uncertainty in mMIMO-ISAC systems.

\subsection{Related Works}
The design and optimization of ISAC systems have attracted significant research attention in recent years. A common approach is to characterize communications and sensing performances through the achievable rate and the CRLB, respectively~\cite{geng2024joint,bian2025fundamental,guo2025fundamental,guang2025communication,uniyal2025sum}. Based on these metrics, various communications--sensing tradeoff analyses and resource-allocation frameworks have been developed. For example, Gan \emph{et al.}~\cite{gan2024coverage} proposed a stochastic geometry framework to characterize sensing coverage and communications performance, while Soltani \emph{et al.}~\cite{soltani2025stochastic} studied sensing reliability through the CRLB outage probability and ergodic CRLB. CRLB-constrained beamforming and resource-allocation problems have subsequently been addressed using semidefinite relaxation~\cite{liu2021cramer}, alternating optimization~\cite{xia2025symbiotic}, successive convex approximation (SCA)~\cite{fang2026optimal}, and manifold optimization~\cite{zargari2024riemannian}.

More recently, ISAC has been integrated with mMIMO architectures, where large antenna arrays provide both high sensing resolution and substantial communications beamforming gains~\cite{liao,nhan_isac_power_allocation,nguyen_jsac_energy_efficiency,topal2024multi,meng_mmimo}. Existing studies have investigated power allocation~\cite{liao,nhan_isac_power_allocation,nguyen_jsac_energy_efficiency}, multi-target sensing~\cite{topal2024multi}, and cooperative mMIMO-ISAC deployments~\cite{meng_mmimo}. In particular, closed-form expressions for the achievable rate and sensing CRLB derived in~\cite{nhan_isac_power_allocation} enabled efficient communications-sensing resource-allocation designs under sensing constraints.

The above works generally assume perfect target-angle information at the BS. In practice, however, target-angle estimates are subject to location uncertainty~\cite{lyu2024dualrobust}, steering-vector errors~\cite{li2008target}, hardware impairments~\cite{chen2023modeling}, calibration errors~\cite{wang2021ota}, and beampattern mismatches~\cite{chen2026mismatch,jia2024imperfect}. These imperfections create mismatches between the presumed and actual target directions, thereby degrading sensing and beamforming performance. In radar systems, the effects of phase-synchronization and modeling errors have been quantified through CRLB analysis~\cite{he2010crlb,nguyen2011modeling}, while steering-vector mismatches were shown to cause substantial radar-performance degradation~\cite{jia2024imperfect}. For ISAC systems, robust beamforming designs have been proposed under imperfect angle and channel information using worst-case and outage-constrained formulations~\cite{hsieh2025robust,zhang2026mmse}. More recently, Chen \emph{et al.}~\cite{chen2026mismatch} investigated angle-estimation degradation caused by beampattern mismatch and proposed a calibration framework to mitigate the resulting errors.

Despite the advances, the impact of target-angle uncertainty on sensing performance and resource allocation in mMIMO-ISAC systems remains largely unexplored. In particular, the degradation of sensing CRLBs and their outage behavior under random angle errors have not been characterized. This  limits the development
of robust communications-sensing resource-allocation strategies and motivates
the present work.

\subsection{{Motivations and Contributions}}
  
Most existing CRLB-based ISAC frameworks assume perfect target-angle information and optimize communications-sensing resources using deterministic CRLB constraints \cite{liu2021cramer,bian2025fundamental,guang2025communication,soltani2025stochastic,xia2025symbiotic}. In practice, however, target angles are estimated from noisy radar observations and are inevitably subject to estimation errors. Such uncertainties induce steering-vector mismatches, causing the sensing CRLB to become a random quantity rather than a deterministic performance metric. Consequently, conventional CRLB-constrained designs may fail to provide reliable sensing guarantees, even when the prescribed CRLB requirements are satisfied. This challenge becomes particularly relevant in practical mMIMO systems employing uniform planar arrays (UPAs)\cite{nhan_isac_power_allocation,liao,uniyal2025sum}
, where sensing performance depends on both azimuth and elevation angles. These observations suggest that sensing reliability should be characterized via the statistical behavior and outage probability of the CRLB, and that resource-allocation strategies should explicitly account for target-angle uncertainty.


Motivated by the above considerations, we investigate a monostatic downlink mMIMO-ISAC system with a UPA-equipped BS under target-angle uncertainty. We first characterize the CRLBs and their outage probabilities for both azimuth and elevation angles. Building upon these results, we develop a robust power allocation framework that maximizes the communications rate while satisfying probabilistic sensing-reliability constraints. Leveraging the channel-hardening property of mMIMO systems, the proposed design depends only on large-scale fading coefficients, making it practically implementable. 
Our main contributions are summarized as follows:
\begin{itemize}
\item We derive closed-form expressions for the sensing CRLBs of target azimuth and elevation estimation in a UPA-equipped mMIMO-ISAC system under target-angle uncertainty. The proposed expressions explicitly account for angle errors, are applicable to both maximum-ratio transmission (MRT) and zero-forcing (ZF) precoding, and depend only on large-scale fading coefficients. 

\item We characterize the statistical behavior of the sensing CRLBs under target-angle uncertainty. In particular, closed-form approximations for the CRLB distributions and outage probabilities are derived for Gaussian, generalized-uniform, and von Mises angle-error models, providing a unified framework for sensing reliability analysis in mMIMO-ISAC systems.

\item We provide fundamental insights into the impact of target-angle errors on sensing performance. We show that the CRLBs increase quadratically with the angle errors in the small-error regime due to steering-vector mismatch, while the CRLB outage probabilities vanish as the number of transmit or receive antennas grows large.

\item We formulate a robust power allocation problem that jointly optimizes pilot-training and communications/ sensing transmission powers to maximize the communications sum rate subject to sensing outage and total power constraints. The resulting nonconvex problem is solved using an alternating optimization (AO)-SCA algorithm. 

\item Extensive simulations validate the developed analysis and demonstrate the effectiveness of the proposed design. In particular, under severe target-angle uncertainty, the proposed robust scheme reduces the azimuth and elevation CRLB outage probabilities by up to $60\%$ compared with conventional non-robust designs. {Furthermore, the proposed robust scheme attains up to
$45\%$ higher sum rates than the non-robust
design at strict CRLB thresholds and less stringent outage
thresholds.}

	\end{itemize}

Notably, the closely related works in \cite{nhan_isac_power_allocation,uniyal2026} provide closed-form expressions for the CRLBs under perfect target-angle information. In contrast, we present closed-form expressions for the CRLBs under target-angle uncertainty and obtain corresponding approximations in the small-error regime. We also provide unified expressions for the outage probability of the CRLBs under different angle error distributions. Unlike \cite{nhan_isac_power_allocation,uniyal2026}, we devise a robust power allocation scheme to jointly optimize the pilot and signal transmission/probing powers. Thus, the existing results for the CRLBs and the optimization framework in \cite{nhan_isac_power_allocation,uniyal2026} are not directly applicable to our design. 

	\subsection{{Paper Organization and Notation}}
	The rest of the paper is organized as follows. In Section \ref{sec_system_model}, we present the ISAC signal models and transmit precoder. Section \ref{sec_perf_analysis}, presents the closed-form expressions for ISAC performance metrics. Section \ref{sec_opt} details the proposed robust power allocation method. Simulation results are provided in Section \ref{section:numerical_results}, while Section \ref{sec:conclusion} concludes the paper. 

\emph{Notations:}	
	Vectors and matrices are represented by lowercase
and uppercase bold letters, respectively. We use $\mathbb{E}\{\cdot\}$ for expectation operator, while
$\mathtt{tr}(\cdot)$ and $\mathtt{vec}(\cdot)$ denote the trace and
vectorization operators, respectively; $(\cdot)^*$, $(\cdot)^{\T}$, and $(\cdot)^{\H}$ denote the conjugate,
transpose, and conjugate-transpose operators, respectively. The notation $\dot{y}_{x}$ denotes the partial derivative of $y$
with respect to (w.r.t.) $x$, i.e., $\dot{y}_{x}\triangleq \partial y/\partial x$,
whereas $\odot$ and $\otimes$ denote the Hadamard and Kronecker products, respectively. We use $|\cdot|$ and $\|\cdot\|$ for the complex modulus and Euclidean norm, respectively, and $\mathcal{CN}(0,\sigma^2)$ for the zero-mean complex Gaussian distribution with variance $\sigma^2$.

\begin{figure}[t]
	\centering
	\includegraphics[scale=0.25]{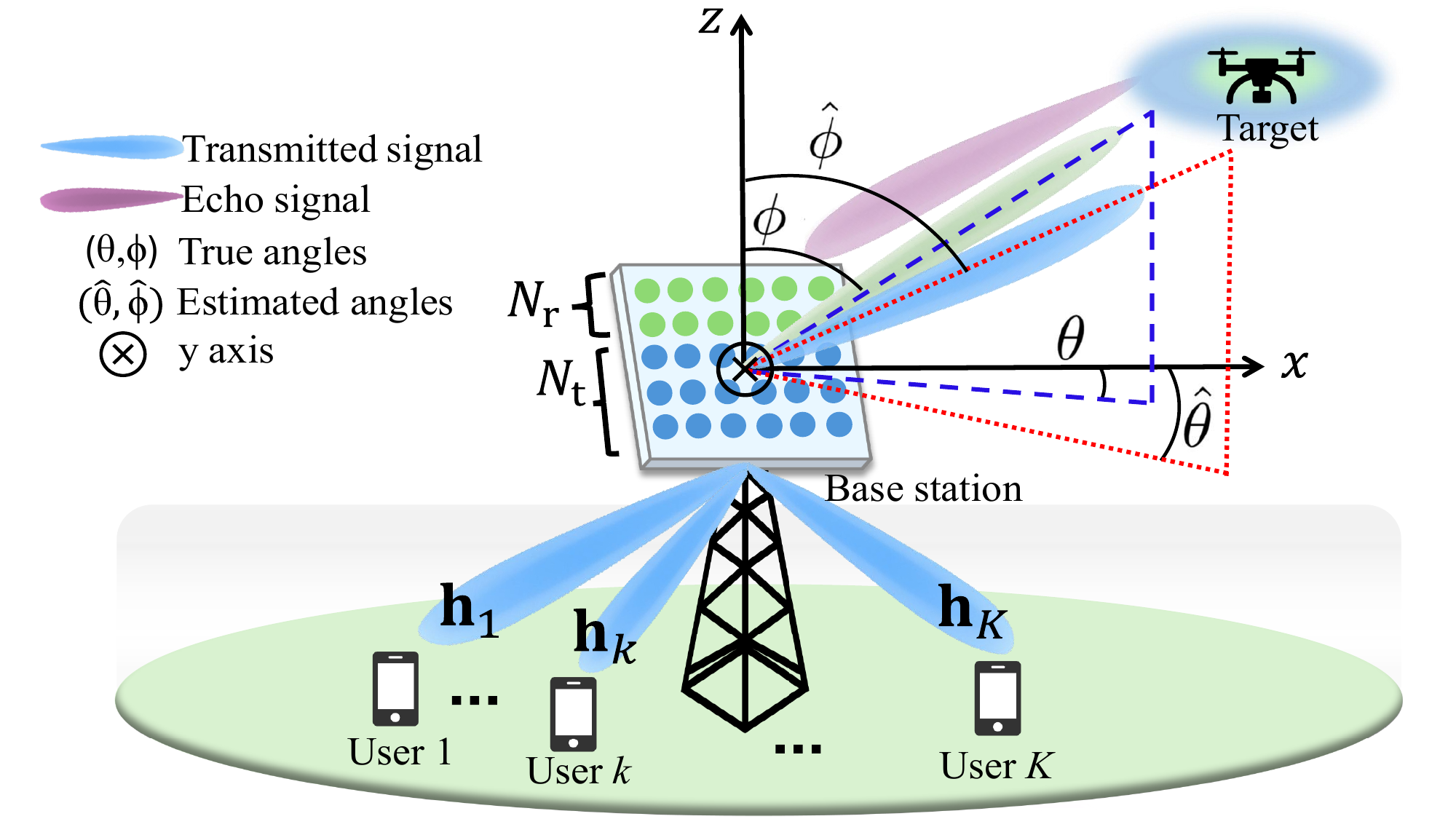}  
	\caption{\small  Monostatic mMIMO ISAC system under target-angle error.}
	\label{fig:schematic}
\end{figure}
    
\section{Signal Model}\label{sec_system_model}
We consider a monostatic downlink mMIMO-ISAC system consisting of a BS, $K$ single-antenna communications users, and a sensing target, as illustrated in Fig.~\ref{fig:schematic}. The BS is equipped with a large UPA comprising $N_{\mathtt t}$ transmit antennas and $N_{\mathtt r}$ receive antennas. 
The BS simultaneously supports downlink communications and target sensing through a shared transmit waveform. The signal reflected by the target is received by the $N_{\mathtt r}$ receive antennas for sensing processing.
	
\subsection{Communications Model}
\subsubsection{Uplink Pilot Training}
We consider a time-division duplex (TDD) protocol, where, for each coherence interval, uplink pilot training is first performed for channel estimation, followed by downlink communications and sensing transmission. Let $\tau_{\mathtt c}$ and $\tau_{\mathtt p}$ denote the lengths of the coherence interval and pilot sequence, respectively, with $\tau_{\mathtt p}<\tau_{\mathtt c}$. During the uplink training phase, all $K$ users simultaneously transmit pilot signals to the BS. Specifically, user $k$ transmits the pilot sequence $\sqrt{\tau_{\mathtt p}\vartheta_k}\vp_k\in\setC^{\tau_{\mathtt p}\times 1}$, where $\norm{\vp_k}^2=1$ and $\vartheta_k$ denotes the pilot transmit power. The total pilot-training power is therefore given by
        \begin{align*}
            P_{\mathtt{p}}(\bm{\vartheta}) = \sum_{k=1}^K \vartheta_k \norm{\vp_k}^2 = \bm{1}_K^\T \bm{\vartheta}, \nbthis \label{eq_pow_pilot}
        \end{align*}
        where $\bm{\vartheta} \triangleq [\vartheta_1, \ldots, \vartheta_K]^\T$, and the last equality follows from the fact that $\norm{\vp_k}^2=1, \forall k$.
        
Let $\vh_k\in\mathbb{C}^{N_{\mathtt t}\times 1}$ denote the channel between the BS and user $k$, modeled as
	\begin{align}
		\vh_k = \beta_k^{1/2} \bar{\vh}_k,\label{eq_channel_model}
	\end{align} 
where $\beta_k$ is the large-scale fading coefficient and $\bar{\vh}_k\sim\mathcal{CN}(\mathbf{0},\mI_{N_{\mathtt t}})$ represents the small-scale Rayleigh fading component. The pilot signal received at the BS is expressed as
	\begin{align*}
		\mY_{\mathtt{p}} = \sqrt{\tau_{\mathtt{p}}} \sum_{k=1}^{K} \sqrt{\vartheta_k} \vh_k \vp_k^\H + \mN_{\mathtt{p}},
	\end{align*}
    where $\mN_{\mathtt{p}}\in \mathbb{C}^{\Nt \times \tau_\mathtt{p}}$ is the additive white Gaussian noise (AWGN) matrix at the BS, with entries distributed as $\mathcal{CN}(0,\sigma^{2})$. 
	To estimate $\vh_k$, $\mY_{\mathtt{p}}$ is projected to $\vp_k$, yielding
	\begin{align*}
		\breve{\vy}_k = \mY_{\mathtt{p}} \vp_k = \sqrt{\tau_{\mathtt{p}} \vartheta_k} \vh_k + \sqrt{\tau_{\mathtt{p}} } \sum_{j\neq k} \sqrt{\vartheta_j} \vh_j \vp_j^\H \vp_k + \breve{\vn}_{\mathtt{p}},
	\end{align*}
	where $\breve{\vn}_{\mathtt{p}} = \mN_{\mathtt{p}} \vp_k$. 
    Let $\hat{\vh}_k$ denote the channel estimate of $\vh_k$ using the minimum mean square error (MMSE). By \cite{nhan_isac_power_allocation,Ngo2013}, $\hat{\vh}_k$ is 
	\begin{align*}
		\hat{\vh}_k &= \frac{\sqrt{\tau_{\mathtt{p}} \vartheta_k} \beta_k}{\tau_{\mathtt{p}} \sum_{j=1}^{K} \vartheta_j \beta_j  \abs{\vp_j^\H \vp_k}^2 + \noise} \breve{\vy}_k ,\nbthis \label{eq_est_channel}
	\end{align*}
	where $\hat{\vh}_k \sim \mathcal{CN}(0, \xi_k(\bm{\vartheta}) \mI_{\Nt})$,  with 
	\begin{align}
	\xi_k(\bm{\vartheta}) \triangleq \frac{\tau_{\mathtt{p}} \vartheta_k \beta_k^2}{\tau_{\mathtt{p}} \sum_{j=1}^{K} \vartheta_j \beta_j  \abs{\vp_j^\H \vp_k}^2 + \noise}. \label{eq_xi_k}
	\end{align}
From \eqref{eq_est_channel}, we can model $\vh_k$ as
$\vh_k=\hat{\vh}_k+\ve_k$, where $\ve_k$ is the channel estimation error. It follows that $\ve_k$ is independent of $\hat{\vh}_k$ and is distributed as
$\ve_k \sim \mathcal{CN}(\mathbf{0}, \epsilon_k(\bm{\vartheta}) \mI_{\Nt})$, where
	\begin{align*}
\!\!\epsilon_k(\bm{\vartheta}) &\!\triangleq\!\beta_k \!-\! \xi_k(\bm{\vartheta}) \!= \!\frac{\beta_k \left(\tau_{\mathtt{p}}  \sum_{j\neq k} \vartheta_j \beta_j  \abs{\vp_j^\H \vp_k}^2 + \noise\right)}{\tau_{\mathtt{p}} \sum_{j=1}^{K} \vartheta_j \beta_j  \abs{\vp_j^\H \vp_k}^2 + \noise}. \nbthis \label{eq_epsilon_k}
	\end{align*}
    
\subsubsection{Communications Signal Model}
Let $\vs_m\triangleq [s_{1m},\ldots,s_{Km}]^\T\in\mathbb{C}^{K\times 1}$ denote the transmit symbol vector at time slot $m$, with $\mean{\vs_m \vs_m^{\H}}=\mI_K$. Defining $\mS\triangleq[\vs_1,\ldots,\vs_M]\in\mathbb{C}^{K\times M}$ over a frame of length $M$, with $M < \tau_\mathtt{c}$, we have
$\mS\mS^\H \approx M\mI_K$
for sufficiently large $M$ by the law of large numbers~\cite{liu2021cramer}.

The BS adopts a linear dual-functional precoder
$\mW = \mF \bm{\Gamma} + \vu \bm{\eta}^\T
        \in \mathbb{C}^{\Nt \times K}$,
where $\mF \triangleq [\vf_1,\ldots,\vf_K] \in \mathbb{C}^{\Nt \times K}$
is the communications precoding matrix, and
$\bm{\Gamma} \triangleq \diag{\sqrt{\gamma_1}, \ldots, \sqrt{\gamma_K}}
\in \setC^{K \times K}$ is the corresponding power allocation matrix for the communications users.
In particular, $\vf_k$ and $\gamma_k$ denote the transmit beamformer and power allocated
to user $k$, respectively. Moreover, $\vu \in \mathbb{C}^{\Nt \times 1}$ and
$\bm{\eta} = [\sqrt{\eta_1},\ldots,\sqrt{\eta_K}]^\T
\in \setC^{K \times 1}$ represent the sensing beamformer and the sensing power allocation
coefficients associated with each data stream, respectively.
The dual-functional beamformer vector of user $k$ corresponds to the $k$-th column of $\mW$, \textit{i.e.}
$\vw_k = \sqrt{\gamma_k}\vf_k + \sqrt{\eta_k}\vu$.
The signal vector transmitted at the $m$-th time slot is
$\vx_m = \mW\vs_m = \sum_{k=1}^{K}\vw_k s_{km}$, and the transmit waveform 
is $\mX \triangleq [\vx_1,\ldots,\vx_M] \in \setC^{\Nt \times M}$,
or equivalently, $\mX = \mW\mS$. As a result, the signal matrix received by the communications users is written as
	\begin{align*}
		\mY_{\mathtt{c}} = \mH^\H \mX + \mN_{\mathtt{c}} = \mH^\H \mW \mS + \mN_{\mathtt{c}} \  \in \mathbb{C}^{K\times M},\nbthis \label{eq_comm_model}
	\end{align*}
where $\mH \triangleq [\vh_1, \ldots,\vh_K] \in \mathbb{C}^{\Nt \times K}$, and $\mN_{\mathtt{c}}$ denotes the AWGN
matrix with independent entries distributed as
$\mathcal{CN}(0,\sigma_{\mathtt{c}}^{2})$, and $\sigma_{\mathtt{c}}^{2}$
is the noise variance at the communications users.
The received signal at user $k$ during time slot $m$ is 
\begin{align*}
	{y_{\mathtt{c}}}_{km} = \vh_k^\H \vw_k s_{km}
	+ \vh_k^\H \sum\nolimits_{j \neq k} \vw_j s_{jm}
	+ {n_{\mathtt{c}}}_{km}. \nbthis \label{eq_y_kl}
\end{align*}

\subsection{Radar Model}
The BS simultaneously serves the communications users and probes the sensing target. The target direction is characterized by the azimuth angle $\theta\in[-\pi,\pi]$ and elevation angle $\phi\in[-\pi/2,\pi/2]$, while the corresponding echo signal is received by the BS. The UPA is centered at the origin and deployed on the $\mathtt{y}$--$\mathtt{z}$ plane, as shown in  Fig.~\ref{fig:schematic}. The transmit steering vector is given by
$\va(\theta,\phi)=\ah(\theta,\phi)\otimes\av(\phi)$,
where
\begin{align}
\ah &=
\left[
e^{-j\pi\frac{\Nth-1}{2}\upsilon_{\mathtt{y}}},
e^{-j\pi\frac{\Nth-3}{2}\upsilon_{\mathtt{y}}},
\ldots,
e^{j\pi\frac{\Nth-1}{2}\upsilon_{\mathtt{y}}}
\right]^{\T}, \label{eq_ath}\\
\av &=
\left[
e^{-j\pi\frac{\Ntv-1}{2}\upsilon_{\mathtt{z}}},
e^{-j\pi\frac{\Ntv-3}{2}\upsilon_{\mathtt{z}}},
\ldots,
e^{j\pi\frac{\Ntv-1}{2}\upsilon_{\mathtt{z}}}
\right]^{\T}. \label{eq_atv}
\end{align}
Herein, $\upsilon_{\mathtt{y}}=\sin(\theta)\sin(\phi)$ and
$\upsilon_{\mathtt{z}}=\cos(\phi)$, while $\Nth$ and $\Ntv$
are the numbers of transmit antennas along the $\mathtt{y}$-axis and
$\mathtt{z}$-axis of the UPA, respectively, with $\Nt=\Nth\Ntv$.
The receive steering vector $\vb(\theta,\phi)$ corresponding to the $\Nr$ receive antennas at the BS is modeled analogously. Based on this, the round-trip channel between the BS and the target is given as
\begin{align}
    \mA(\theta,\phi)=\vb(\theta,\phi)\va^{\H}(\theta,\phi).
    \label{eq_A}
\end{align}
Then, the echo signal received at the BS from the target is
\begin{align*}
    \mY_{\mathtt{s}}
    = \beta_\mathtt{s}\mA(\theta,\phi)\mX + \mN_{\mathtt{s}}
    \in \setC^{\Nr \times M},
    \nbthis \label{eq_radar_model}
\end{align*}
where $\beta_{\mathtt{s}}$ denotes the complex reflection coefficient of the
target, and $\mN_{\mathtt{s}}\in\setC^{\Nr\times M}$ is the sensing AWGN
matrix with independent entries distributed as $\mathcal{CN}(0,\sigmas)$. Here, $\sigmas$ is the noise variance at the sensing receiver.
In the subsequent analysis, we omit $(\theta,\phi)$ for notational convenience.

  
\subsection{Dual-Functional Transmit Precoder}
From $\mX = \mW \mS$, the echo signal in \eqref{eq_radar_model} is rewritten as
	\begin{align*}
		\!\!\mY_{\mathtt{s}} = \beta_\mathtt{s} \vb \va^\H  \left(\mF \bm{\Gamma} + \vu \bm{\eta}^\T\right) \!\mS \!+ \!\mN_{\mathtt{s}} = \beta_\mathtt{s} \vb \va^\H \vu \bm{\eta}^\T  \mS + \breve{\mN}_{\mathtt{s}}, \nbthis \label{eq_radar_2}
	\end{align*}
	where $\breve{\mN}_{\mathtt{s}} \triangleq \beta_\mathtt{s} \vb \va^\H \mF \bm{\Gamma} \mS + \mN_{\mathtt{s}}$. With perfect target-direction knowledge, the sensing beamformer is ideally choosen as $\vu=\va(\theta,\phi)$.
    However, since the target direction is precisely the quantity to be estimated, the BS does not have its perfect knowledge. We assume that, prior to the sensing and communications phase, coarse target-angle estimates are obtained during an initial radar search phase \cite{liu2021cramer,nguyen_jsac_energy_efficiency}. Based on these coarse estimates, the BS refines the target angles. To this end, the BS is assumed to know noisy estimates of the target angles. More precisely, the
    sensing beamformer is designed using the estimated angles, $\vu=\va(\hat{\theta},\hat{\phi})$, with
$\hat{\theta}\triangleq\theta+\varepsilon_{\theta}$ and
$\hat{\phi}\triangleq\phi+\varepsilon_{\phi}$. Here,
$\varepsilon_{\theta}$ and $\varepsilon_{\phi}$ denote the azimuth and
elevation angle errors, respectively. We further define the angle-error vector as
$\boldsymbol{\varepsilon}\triangleq[\varepsilon_{\theta},\varepsilon_{\phi}]^{\T}$,
whose statistical properties will be specified in
Section~\ref{section:angleerror_outage}.
For communications, widely-used linear precoders MRT and ZF are deployed, expressed as
	\begin{align*} 
		\mF=\mF_\mathtt{bf} = 
		\begin{cases}
			\hat{\mH}, &\text{for}\ \mathtt{bf=MRT}, \\
			\hat{\mH} (\hat{\mH}^\H \hat{\mH})^{-1},&\text{for}\ \mathtt{bf=ZF},
		\end{cases} \nbthis \label{eq_W_linear}
	\end{align*}
	where $\hat{\mH} = [\hat{\vh}_1, \ldots,\hat{\vh}_K] \in \mathbb{C}^{\Nt \times K}$. As a result, the linear dual-function transmit precoder is
		$\mW_\mathtt{bf} =  \mF_\mathtt{bf}  \bGamma + {\vahat} \bm{\eta}^\T$, with $\mathtt{bf} \in \{\mathtt{MRT}, \mathtt{ZF}\}$. The average transmit power associated with this precoder is given as \cite{nhan_isac_power_allocation}
		\begin{align*}
			P &= \Nt \xibf^\T \bm{\gamma} + \Nt \rho, \nbthis \label{eq_power}
		\end{align*}
		where $\rho = \norm{\bm{\eta}}^2$, and
        \begin{align}
				\!\!\!\!\xibf \!=\! 
				\begin{cases} \begin{array}{ll}
					\!\!\!\!\![\xi_1({\bm{\vartheta}}),\ldots,\xi_K({\bm{\vartheta}})]^\T, &\!\!\!\! \text{for}\ \mathtt{bf=MRT},\\
					\!\!\!\!\!\frac{1}{\Nt(\Nt-K)}\left[\frac{1}{\xi_1({\bm{\vartheta}})},\ldots,\frac{1}{\xi_K({\bm{\vartheta}})}\right]^\T\!\!,& \!\!\!\!\text{for}\ \mathtt{bf=ZF}. \end{array}
				\end{cases}  \nonumber
			\end{align}

\section{Communications and Sensing Performance Analysis}\label{sec_perf_analysis}
We next derive closed-form expressions for the communications rate, sensing CRLBs, and CRLB outage probabilities, which serve as the basis for the subsequent resource-allocation.

\subsection{Communications Sum Rate}
From \eqref{eq_y_kl} and following \cite{nhan_isac_power_allocation}, the achievable rate of user $k$ with MRT or ZF precoder can be expressed in closed form as 
\begin{align*}
			\!\!\!\!\mathcal{R}_k^\mathtt{bf}({\bm{\vartheta}}, \bm{\gamma},\rho) &\!= \!{\tau_\mathtt{0}} \!\log_2 \!\!\left(\!1 \!+\! \frac{ \lambdabf_k({\bm{\vartheta}}) \gamma_k }{\Nt \beta_k \rho  + \Nt {\zetabf({\bm{\vartheta}})^\T_k} \bm{\gamma} + \sigmac}\!\right)\!, \nbthis \label{eq_SE_theo}
		\end{align*}
		where  ${\tau}_\mathtt{0} \triangleq \left(\tau_{\mathtt{c}} - \tau_{\mathtt{p}}\right)/\tau_{\mathtt{c}}$, $\bm{\gamma} \triangleq [\gamma_1,\ldots,\gamma_K]^\T$, and
		\begin{align*}
			\lambdabf_k({\bm{\vartheta}}) &= 
			\begin{cases}
				\begin{array}{ll} \Nt^2 \xi_k^2({\bm{\vartheta}}), &  \text{for}\ \mathtt{bf=MRT}, \\
				1,& \text{for}\ \mathtt{bf=ZF},\end{array}
			\end{cases} 
            \\
			\!\!\zetabf_k({\bm{\vartheta}}) \!&=\! 
			\begin{cases}
				\begin{array}{ll} 
                    \!\!\!\!\!\beta_k \left[\xi_1({\bm{\vartheta}}), \ldots, \xi_K({\bm{\vartheta}})\right]^\T, & \hspace{-0.1cm}\text{for}\ \mathtt{bf=MRT}, \\
				\!\!\!\!\!\frac{\epsilon_k({\bm{\vartheta}})}{\Nt(\Nt-K)} \left[ \frac{1}{\xi_1({\bm{\vartheta}})}, \ldots, \frac{1}{\xi_K({\bm{\vartheta}})} \right]^\T\!\!\!, & \hspace{-0.1cm}\text{for}\ \mathtt{bf=ZF}. \end{array}
			\end{cases} 
		\end{align*}
Unlike~\cite{nhan_isac_power_allocation}, the achievable rate in~\eqref{eq_SE_theo} explicitly depends on the pilot-power allocation vector  $\boldsymbol{\vartheta}$ through $\xi_k(\boldsymbol{\vartheta})$, in addition to the communications and sensing power allocation variables $\boldsymbol{\gamma}$ and $\rho$, respectively.
   
\subsection{Sensing CRLBs Under Target-Angle Uncertainty}
To characterize the sensing performance, we derive CRLBs for the target azimuth and elevation angles. To this end, we first derive the Fisher information matrix (FIM) associated with the unknown parameter vector as
\begin{align*}
    \vy_{\mathtt{s}}
    = \mathtt{vec}(\mY_{\mathtt{s}})
    =  \mathbf{v}_{\mathtt{s}}+\vn_{\mathtt{s}},
    \nbthis \label{eq_radar_model_2}
\end{align*}
where
$\mathbf{v}_{\mathtt{s}}\triangleq\beta_\mathtt{s}\mathtt{vec}(\mA\mX)
\in\setC^{\Nr M\times 1}$ and
$\vn_{\mathtt{s}}\triangleq\mathtt{vec}(\mN_{\mathtt{s}})
\in\setC^{\Nr M\times 1}$, with
$\vn_{\mathtt{s}}\sim\mathcal{CN}(\bm{0},\sigmas\mI_{\Nr M})$.
Let 
$\bm{\omega}\triangleq[\theta,\phi,\tilde{\bm{\beta_\mathtt{s}}}^{\T}]^{\T}
\in\setR^{4\times 1}$ denote the unknown parameter vector, where
$\tilde{\bm{\beta_\mathtt{s}}}=[\re{\beta_\mathtt{s}},\im{\beta_\mathtt{s}}]^{\T}$.
Following \cite{nhan_isac_power_allocation}, the FIM associated with estimating $\bm{\omega}$ from
\eqref{eq_radar_model_2} is given by
	\begin{align}
		\mJ_{\bm{\omega}} = 
		\begin{bmatrix}
			\Jtt  & \Jtp & \Jta\\
			\Jtp & \Jpp & \Jpa \\
			\Jta^\T & \Jpa^\T &\Jaa
		\end{bmatrix},
        \label{full_fim}
	\end{align}
where the entries of $\mJ_{\bm{\omega}}$ are given by
	\begin{align*}
	T_{\psi\psi} &= \chi_\mathtt{s} \abs{\beta_\mathtt{s}}^2 \tr{\dot{\mA}_{\psi} \mR_x \dot{\mA}^\H_{\psi}} \in \setR^{1 \times 1}, \nbthis \label{eq_def_Jtt} \\
		\Jtp &= \chi_\mathtt{s} \abs{\beta_\mathtt{s}}^2 \tr{\dot{\mA}_{\phi} \mR_x \dot{\mA}^\H_{\theta}} \in \setR^{1 \times 1}, \nbthis \label{eq_def_Jtp}\\
		\mathbf{t}_{\psi\tilde{\boldsymbol{\beta}}_\mathtt{s}} &= \chi_\mathtt{s} \re{\beta_\mathtt{s}^* \tr{\mA \mR_x \dot{\mA}^\H_{\psi}}[1,j]} \in \setR^{1 \times 2}, \nbthis \label{eq_def_Jta}\\
		\Jaa &= \chi_\mathtt{s} \tr{\mA \mR_x \mA^\H} \mI_2  \in \setR^{2 \times 2}. \nbthis \label{eq_def_Jaa}
	\end{align*}
Here, $\psi\in\{\theta,\phi\}$, $\dot{\mA}_{\theta} = \bdottheta \va^\H + \vb \adottheta^\H,\ 
		\dot{\mA}_{\phi} = \bdotphi \va^\H + \vb \adotphi^\H$,
	\begin{align*}
            &\mR_x \triangleq \frac{1}{M} \mean{\mX \mX^\H}  =
\mean{\mW_\mathtt{bf} \mW_\mathtt{bf}^\H} \triangleq \mR_\mathtt{bf}, \nbthis\label{eq_approx_cov}
	\end{align*}
	and $\chi_\mathtt{s} \triangleq \frac{2M}{\sigmas}$. From~\eqref{full_fim}, the equivalent FIM for $(\theta,\phi)$ is obtained
via the Schur complement w.r.t. $\tilde{\boldsymbol{\beta}}_{\mathtt{s}}$ as 
 \begin{align}
	\!\!\mJ_{\theta,\phi}
&=
\begin{bmatrix}
\!\Jtt\!-\!\Jta\Jaa^{-1}\Jta^\T
\!
&
\Jtp\!-\!\Jta\Jaa^{-1}\Jpa^\T
\!
\\
\Jtp\!-\!\Jpa\Jaa^{-1}\Jta^\T
\!
&
\Jpp\!-\!\Jpa\Jaa^{-1}\Jpa^\T
\!
\end{bmatrix}\nonumber \\
&\triangleq\begin{bmatrix}
 \widetilde{T}_{\theta\theta}
&
\widetilde{T}_{\theta\phi}
\\
\widetilde{T}_{\theta\phi}
&
 \widetilde{T}_{\phi\phi}
\end{bmatrix}.
\label{eq_FIM}
\end{align} 
 Based on \eqref{eq_FIM}, the CRLBs for estimating the target azimuth and elevation angles are obtained as follows.
\begin{theorem}
\label{thm:deterministic_crlb}
The closed-form CRLBs for the target azimuth and elevation angles, $\psi \in\{\theta,\phi\}$, under target-angle uncertainty are given in \eqref{eq_closed_crlb_theta} and
\eqref{eq_closed_crlb_phi}, where
\begin{figure*}[!t]
\begin{align}
\label{eq_closed_crlb_theta}
{\mathtt{CRLB}}_{\theta}^{\mathtt{bf}}
(\boldsymbol{\vartheta},\boldsymbol{\gamma},\rho,\boldsymbol{\varepsilon})
&\!=\!\!
\left(
\boldsymbol{\xi}_{\mathtt{bf}}^{\T}(\boldsymbol{\vartheta})
\boldsymbol{\gamma}c_{\theta}
+\rho\hat{c}_{\theta}
\!-\!
\frac{
\rho^{2}
|\hat{c}_{\theta\tilde{\boldsymbol{\beta}}_{\mathtt{s}}}|^{2}
}
{
\boldsymbol{\xi}_{\mathtt{bf}}^{\T}(\boldsymbol{\vartheta})\boldsymbol{\gamma}
c_{\tilde{\boldsymbol{\beta}}_{\mathtt{s}}\tilde{\boldsymbol{\beta}}_{\mathtt{s}}}
\!+\!\rho
\hat{c}_{\tilde{\boldsymbol{\beta}}_{\mathtt{s}}\tilde{\boldsymbol{\beta}}_{\mathtt{s}}}
}
\!-\!
\frac{
\left(
\boldsymbol{\xi}_{\mathtt{bf}}^{\T}(\boldsymbol{\vartheta})
\boldsymbol{\gamma}c_{\theta\phi}
+\rho\hat{c}_{\theta\phi}
-
\frac{
\rho^{2}
\hat{c}_{\theta\tilde{\boldsymbol{\beta}}_{\mathtt{s}}}
\hat{c}_{\phi\tilde{\boldsymbol{\beta}}_{\mathtt{s}}}^{*}
}
{
\boldsymbol{\xi}_{\mathtt{bf}}^{\T}(\boldsymbol{\vartheta})\boldsymbol{\gamma}
c_{\tilde{\boldsymbol{\beta}}_{\mathtt{s}}\tilde{\boldsymbol{\beta}}_{\mathtt{s}}}
+\rho
\hat{c}_{\tilde{\boldsymbol{\beta}}_{\mathtt{s}}\tilde{\boldsymbol{\beta}}_{\mathtt{s}}}
}
\right)^{2}
}
{
\boldsymbol{\xi}_{\mathtt{bf}}^{\T}(\boldsymbol{\vartheta})
\boldsymbol{\gamma}c_{\phi}
+\rho\hat{c}_{\phi}
-
\frac{
\rho^{2}
|\hat{c}_{\phi\tilde{\boldsymbol{\beta}}_{\mathtt{s}}}|^{2}
}
{
\boldsymbol{\xi}_{\mathtt{bf}}^{\T}(\boldsymbol{\vartheta})\boldsymbol{\gamma}
c_{\tilde{\boldsymbol{\beta}}_{\mathtt{s}}\tilde{\boldsymbol{\beta}}_{\mathtt{s}}}
+\rho
\hat{c}_{\tilde{\boldsymbol{\beta}}_{\mathtt{s}}\tilde{\boldsymbol{\beta}}_{\mathtt{s}}}
}
}
\right)^{-1}
\\
\label{eq_closed_crlb_phi}
{\mathtt{CRLB}}_{\phi}^{\mathtt{bf}}
(\boldsymbol{\vartheta},\boldsymbol{\gamma},\rho,\boldsymbol{\varepsilon})
&\!=\!\!
\left(
\boldsymbol{\xi}_{\mathtt{bf}}^{\T}(\boldsymbol{\vartheta})
\boldsymbol{\gamma}c_{\phi}
\!+\!\rho\hat{c}_{\phi}
\!-\!
\frac{
\rho^{2}
|\hat{c}_{\phi\tilde{\boldsymbol{\beta}}_{\mathtt{s}}}|^{2}
}
{
\boldsymbol{\xi}_{\mathtt{bf}}^{\T}(\boldsymbol{\vartheta})\boldsymbol{\gamma}
c_{\tilde{\boldsymbol{\beta}}_{\mathtt{s}}\tilde{\boldsymbol{\beta}}_{\mathtt{s}}}
+\rho
\hat{c}_{\tilde{\boldsymbol{\beta}}_{\mathtt{s}}\tilde{\boldsymbol{\beta}}_{\mathtt{s}}}
}
-
\frac{
\left(
\boldsymbol{\xi}_{\mathtt{bf}}^{\T}(\boldsymbol{\vartheta})
\boldsymbol{\gamma}c_{\theta\phi}
+\rho\hat{c}_{\theta\phi}
-
\frac{
\rho^{2}
\hat{c}_{\theta\tilde{\boldsymbol{\beta}}_{\mathtt{s}}}
\hat{c}_{\phi\tilde{\boldsymbol{\beta}}_{\mathtt{s}}}^{*}
}
{
\boldsymbol{\xi}_{\mathtt{bf}}^{\T}(\boldsymbol{\vartheta})\boldsymbol{\gamma}
c_{\tilde{\boldsymbol{\beta}}_{\mathtt{s}}\tilde{\boldsymbol{\beta}}_{\mathtt{s}}}
\!+\!\rho
\hat{c}_{\tilde{\boldsymbol{\beta}}_{\mathtt{s}}\tilde{\boldsymbol{\beta}}_{\mathtt{s}}}
}
\right)^{2}
}
{
\boldsymbol{\xi}_{\mathtt{bf}}^{\T}(\boldsymbol{\vartheta})
\boldsymbol{\gamma}c_{\theta}
\!+\!\rho\hat{c}_{\theta}
\!-\!
\frac{
\rho^{2}
|\hat{c}_{\theta\tilde{\boldsymbol{\beta}}_{\mathtt{s}}}|^{2}
}
{
\boldsymbol{\xi}_{\mathtt{bf}}^{\T}(\boldsymbol{\vartheta})\boldsymbol{\gamma}
c_{\tilde{\boldsymbol{\beta}}_{\mathtt{s}}\tilde{\boldsymbol{\beta}}_{\mathtt{s}}}
\!+\!\rho
\hat{c}_{\tilde{\boldsymbol{\beta}}_{\mathtt{s}}\tilde{\boldsymbol{\beta}}_{\mathtt{s}}}
}
}
\right)^{-1}
\end{align}
\hrule
\end{figure*}
\begin{align}
\label{chat1_thm}
c_{\psi} &\triangleq \chi_\mathtt{s}|\beta_\mathtt{s}|^{2}\Big(\|\dot{\mathbf{a}}_{\psi}\|^{2}N_\mathtt{r}+\|\dot{\mathbf{b}}_{\psi}\|^{2}N_\mathtt{t}\Big), \\
c_{\theta\phi} &\triangleq \chi_\mathtt{s}|\beta_\mathtt{s}|^{2}\Big(\dot{\mathbf{a}}_{\theta}^{\H}\dot{\mathbf{a}}_{\phi}N_\mathtt{r}+\dot{\mathbf{b}}_{\theta}^{\H}\dot{\mathbf{b}}_{\phi}N_\mathtt{t}\Big), \\
c_{\tilde{\boldsymbol{\beta_\mathtt{s}}}\tilde{\boldsymbol{\beta_\mathtt{s}}}} &\triangleq \chi_\mathtt{s}N_\mathtt{t}N_\mathtt{r},
\label{cbeta}\\
\hat{c}_{\psi} &\triangleq \chi_\mathtt{s}|\beta_\mathtt{s}|^{2}\Big(|g_{0}|^{2}\|\dot{\mathbf{b}}_{\psi}\|^{2}+|g_{\psi}|^{2}N_\mathtt{r}\Big), 
\label{chattheta}\\
\label{chatthetaphi}
\hat{c}_{\theta\phi} &\triangleq \chi_\mathtt{s}|\beta_\mathtt{s}|^{2}\Big(|g_{0}|^{2}\dot{\mathbf{b}}_{\theta}^{\H}\dot{\mathbf{b}}_{\phi}+g_{\theta}g_{\phi}N_\mathtt{r}\Big), \\
\hat{c}_{\tilde{\boldsymbol{\beta_\mathtt{s}}}\tilde{\boldsymbol{\beta_\mathtt{s}}}} &\triangleq \chi_\mathtt{s} |g_{0}|^{2}N_\mathtt{r}, \\
\hat{c}_{\psi\tilde{\boldsymbol{\beta_\mathtt{s}}}} &\triangleq \chi_\mathtt{s} \beta_\mathtt{s}^{*}\,(g_{0}g_{\psi})N_\mathtt{r},\
\label{chat2_thm}
\end{align}
with 
\begin{align}
\|\dot{\mathbf{a}}_{\theta}\|^{2}
&= \frac{N_{\mathtt{t}}\left(N_{\mathtt{ty}}^{2}-1\right)}{12}
   \pi^{2}\cos^{2}(\theta)\sin^{2}(\phi),
\label{eq:norm_adottheta}
\\
\|\dot{\mathbf{a}}_{\phi}\|^{2}
&= \frac{N_{\mathtt{t}}}{12}\pi^{2}
   \big(
   \left(N_{\mathtt{ty}}^{2}-1\right)\sin^{2}(\theta)\cos^{2}(\phi)
   \nonumber \\ &\quad+\left(N_{\mathtt{tz}}^{2}-1\big)\sin^{2}(\phi)
   \right),
\label{eq:norm_adotphi}
\\
\|\dot{\mathbf{b}}_{\theta}\|^{2}
&= \frac{N_{\mathtt{r}}\left(N_{\mathtt{ry}}^{2}-1\right)}{12}
   \pi^{2}\cos^{2}(\theta)\sin^{2}(\phi),
\label{eq:norm_bdottheta}
\\
\|\dot{\mathbf{b}}_{\phi}\|^{2}
&= \frac{N_{\mathtt{r}}}{12}\pi^{2}
   \big(
   \left(N_{\mathtt{ry}}^{2}-1\right)\sin^{2}(\theta)\cos^{2}(\phi)
   \nonumber \\& \quad+\left(N_{\mathtt{rz}}^{2}-1\right)\sin^{2}(\phi)
   \big),
\label{eq:norm_bdotphi}
\\
\dot{\mathbf{a}}_{\theta}^{\H}\dot{\mathbf{a}}_{\phi}
&= \frac{N_{\mathtt{t}}\left(N_{\mathtt{ty}}^{2}-1\right)}{48}
   \pi^{2}\sin(2\theta)\sin(2\phi),
\label{eq:inner_adot}
\\
\dot{\mathbf{b}}_{\theta}^{\H}\dot{\mathbf{b}}_{\phi}
&= \frac{N_{\mathtt{r}}\left(N_{\mathtt{ry}}^{2}-1\right)}{48}
   \pi^{2}\sin(2\theta)\sin(2\phi),
\label{eq:inner_bdot}
\end{align}
 and 
\begin{align}
\label{g0}
g_{0} &\triangleq \mathbf{u}^\H \mathbf{a}
= D_{N_{\mathtt{t}\mathtt{y}}}(\Delta_{\mathtt{y}})
  D_{N_{\mathtt{t}\mathtt{z}}}(\Delta_{\mathtt{z}}), \\
  \label{gtheta}
g_{\theta} & \triangleq \mathbf{u}^\H \dot{\mathbf{a}}_\theta
= j\pi\cos(\theta)\sin(\phi)
  Q_{N_{\mathtt{t}\mathtt{y}}}(\Delta_{\mathtt{y}})
  D_{N_{\mathtt{t}\mathtt{z}}}(\Delta_{\mathtt{z}}), \\
g_{\phi} &\triangleq \mathbf{u}^\H \dot{\mathbf{a}}_\phi = j\pi\sin(\theta)\cos(\phi)
  Q_{N_{\mathtt{t}\mathtt{y}}}(\Delta_{\mathtt{y}})
  D_{N_{\mathtt{t}\mathtt{z}}}(\Delta_{\mathtt{z}})\nonumber\\
&\hspace{1.5cm}
 - j\pi\sin(\phi)
  D_{N_{\mathtt{t}\mathtt{y}}}(\Delta_{\mathtt{y}})
  Q_{N_{\mathtt{t}\mathtt{z}}}(\Delta_{\mathtt{z}}),
  \label{gphi}
\end{align}

\begin{align}
D_{N_{\mathtt{t}\mathtt{x}}}(\Delta_{\mathtt{x}})
&\triangleq
\frac{
\sin\left(\frac{N_{\mathtt{t}\mathtt{x}}\pi\Delta_{\mathtt{x}}}{2}\right)}
{
\sin\left(\frac{\pi\Delta_{\mathtt{x}}}{2}\right)
}, \label{Sn}\\
Q_{N_{\mathtt{t}\mathtt{x}}}(\Delta_{\mathtt{x}})
&\triangleq
\frac{-j}{2\sin^2\left(\frac{\pi\Delta_{\mathtt{x}}}{2}\right)}
\Bigg(
N_{\mathtt{t}\mathtt{x}}
\cos\left(\frac{N_{\mathtt{t}\mathtt{x}}\pi\Delta_{\mathtt{x}}}{2}\right)
\sin\left(\frac{\pi\Delta_{\mathtt{x}}}{2}\right) \nonumber\\
&\hspace{0.5cm}
-
\sin\left(\frac{N_{\mathtt{t}\mathtt{x}}\pi\Delta_{\mathtt{x}}}{2}\right)
\cos\left(\frac{\pi\Delta_{\mathtt{x}}}{2}\right)
\Bigg),
\label{Qn}
\end{align}
for $\mathtt{x}\in \{\mathtt{y},\mathtt{z}\}$ and
\begin{align}
\Delta_{\mathtt{x}} \triangleq
\begin{cases}
\sin(\theta)\sin(\phi)-\sin(\hat{\theta})\sin(\hat{\phi}),
& \mathtt{x}=\mathtt{y},\\[1mm]
\cos(\phi)-\cos(\hat{\phi}),
& \mathtt{x}=\mathtt{z}.
\label{deltax}
\end{cases}
\end{align}
\end{theorem}
\begin{proof}
 See Appendix~\ref{app:crlb_error}. \hfill$\blacksquare$ 
\end{proof}
Based on \eqref{eq_closed_crlb_theta} and \eqref{eq_closed_crlb_phi}, we now
investigate the scaling laws of the CRLBs w.r.t.
$N_{\mathtt{t}}$ and $N_{\mathtt{r}}$. For analytical tractability, we consider equal power allocation between communications and sensing in~\eqref{eq_power}, such that
$N_{\mathtt{t}}\boldsymbol{\xi}_{\mathtt{bf}}^{\T}(\boldsymbol{\vartheta})
\boldsymbol{\gamma}=N_{\mathtt{t}}\rho=P/2$. The scaling laws are summarized in the following corollaries.
\begin{corollary}
\label{rem:asymptotic_Nt}
For fixed nonzero target-angle errors and square UPAs, \textit{i.e.}
$N_{\mathtt{t}\mathtt{y}}=N_{\mathtt{t}\mathtt{z}}=\sqrt{N_{\mathtt{t}}}$
and
$N_{\mathtt{r}\mathtt{y}}=N_{\mathtt{r}\mathtt{z}}=\sqrt{N_{\mathtt{r}}}$,
the CRLBs scale as $1/N_{\mathtt{t}}$ for $N_{\mathtt{t}}\gg1$ with fixed
$N_{\mathtt{r}}$, and as $1/N_{\mathtt{r}}^{2}$ for
$N_{\mathtt{r}}\gg1$ with fixed $N_{\mathtt{t}}$. {This implies that deploying large antenna arrays at the BS can enable highly accurate sensing  even in the presence of target-angle errors.}
\end{corollary}
\begin{proof}
    See Appendix \ref{app:asymptotic_Nt_Nr}.\hfill$\blacksquare$ 
\end{proof}
\begin{corollary}
\label{rem:asymptotic_crlb}
For fixed $N_{\mathtt{t}}$ and $N_{\mathtt{r}}$, and for small
target-angle errors, the CRLBs in \eqref{eq_closed_crlb_theta} and \eqref{eq_closed_crlb_phi} can be approximated as
\begin{align} \mathtt{CRLB}_{\psi}^{\mathtt{bf}}(\boldsymbol{\varepsilon}) &\approx \mathtt{CRLB}_{\psi}^{\mathtt{bf}}(\mathbf{0}) + \left(\mathtt{CRLB}_{\psi}^{\mathtt{bf}}(\mathbf{0})\right)^2 \rho\chi_{\mathtt{s}}N_\mathtt{t}|\beta_{\mathtt{s}}|^2 \nonumber\\ &\,\,\,\times \| \dot{\mathbf{b}}_{\psi} - r_{\psi} \dot{\mathbf{b}}_{\psi_{\mathtt{c}}} \|^2\!\! \underbrace{ \| \mathbf{a}(\theta,\phi) - \mathbf{a}(\hat{\theta},\hat{\phi}) \|^2}_{\text{transmit steering-vector mismatch}}, \label{eq:crlb_small_error_mismatch} \end{align}
where $\psi\in\{\theta,\phi\}$, $\psi_{\mathtt{c}}=\phi$ for
$\psi=\theta$, and $\psi_{\mathtt{c}}=\theta$ for $\psi=\phi$.
Furthermore, $\mathtt{CRLB}_{\psi}^{\mathtt{bf}}(\mathbf{0})$ denotes the
CRLB without target-angle errors, $\|\mathbf{a}(\theta,\phi)-\mathbf{a}(\hat{\theta},\hat{\phi})\|^2
\approx
\varepsilon_\theta^2\|\dot{\mathbf{a}}_{\theta}\|^2
+\varepsilon_\phi^2\|\dot{\mathbf{a}}_{\phi}\|^2
+2\varepsilon_\theta\varepsilon_\phi
\dot{\mathbf{a}}_{\theta}^{\H}\dot{\mathbf{a}}_{\phi}$, and
$r_{\psi}\triangleq r_{\mathtt{num}}/r_{\mathtt{den}_{\psi_{\mathtt{c}}}}$. Here,
$r_{\mathtt{num}}\triangleq
\boldsymbol{\xi}_{\mathtt{bf}}^{\T}(\boldsymbol{\vartheta})
\boldsymbol{\gamma}c_{\theta\phi}
+\rho\chi_{\mathtt{s}}|\beta_{\mathtt{s}}|^{2}N_{\mathtt{t}}^{2}
\dot{\mathbf{b}}_{\theta}^{\H}\dot{\mathbf{b}}_{\phi}$ and 
$r_{\mathtt{den}_{\psi_{\mathtt{c}}}}\triangleq
\boldsymbol{\xi}_{\mathtt{bf}}^{\T}(\boldsymbol{\vartheta})
\boldsymbol{\gamma}c_{\psi_\mathtt{c}}
+\rho\chi_{\mathtt{s}}|\beta_{\mathtt{s}}|^{2}N_{\mathtt{t}}^{2}
\|\dot{\mathbf{b}}_{\psi_\mathtt{c}}\|^{2}$.
From \eqref{eq:crlb_small_error_mismatch}, we note that the target-angle errors increase
the CRLB mainly through the transmit steering-vector mismatch.
Since this mismatch is quadratic in
$(\varepsilon_{\theta},\varepsilon_{\phi})$, the CRLB increases
quadratically with the target-angle errors.
\end{corollary}
\begin{proof}
    See Appendix \ref{app:asymptotic_crlb}.\hfill$\blacksquare$ 
\end{proof}

\begin{remark}
\label{rem:tradeoff}
Let us define the feasible sensing-power sets for a given CRLB
threshold $\mathtt{CRLB}_{\psi}^{0}$ as
$\mathcal{F}_{\psi}(\boldsymbol{\varepsilon})
\triangleq
\{\rho\geq0:
\mathtt{CRLB}_{\psi}^{\mathtt{bf}}(\boldsymbol{\varepsilon},\rho)
\leq \mathtt{CRLB}_{\psi}^{0}\}$ and
$\mathcal{F}_{\psi}(\mathbf{0})
\triangleq
\{\rho\geq0:
\mathtt{CRLB}_{\psi}^{\mathtt{bf}}(\mathbf{0},\rho)
\leq \mathtt{CRLB}_{\psi}^{0}\}$, where $\psi\in\{\theta,\phi\}$. From
\eqref{eq:crlb_small_error_mismatch}, for a given sensing power $\rho$, we
have
$\mathtt{CRLB}_{\psi}^{\mathtt{bf}}(\boldsymbol{\varepsilon},\rho)
\geq
\mathtt{CRLB}_{\psi}^{\mathtt{bf}}(\mathbf{0},\rho)$. Hence, if
$\rho\in\mathcal{F}_{\psi}(\boldsymbol{\varepsilon})$,
$\mathtt{CRLB}_{\psi}^{\mathtt{bf}}(\mathbf{0},\rho)
\leq
\mathtt{CRLB}_{\psi}^{\mathtt{bf}}(\boldsymbol{\varepsilon},\rho)
\leq \mathtt{CRLB}_{\psi}^{0}$, which yields
$\mathcal{F}_{\psi}(\boldsymbol{\varepsilon})
\subseteq
\mathcal{F}_{\psi}(\mathbf{0})$. Therefore,
$\rho_{\min_\psi}(\boldsymbol{\varepsilon})
\geq
\rho_{\min_\psi}(\mathbf{0})$, where
$\rho_{\min_\psi}(\boldsymbol{\varepsilon})
\triangleq \inf \mathcal{F}_{\psi}(\boldsymbol{\varepsilon})$. Thus, the minimum sensing power requirement to satisfy the CRLB threshold increases under target-angle errors. From \eqref{eq_power},
for a fixed downlink transmit power, the power available for communications satisfies
$\boldsymbol{\xi}_{\mathtt{bf}}^{\T}(\boldsymbol{\vartheta})
\boldsymbol{\gamma}
\leq
P/N_{\mathtt{t}}
-\rho_{\min_\psi}(\boldsymbol{\varepsilon})
\leq
P/N_{\mathtt{t}}
-\rho_{\min_\psi}(\mathbf{0})$. Hence, target-angle errors reduce the power
available for communications while increasing the sensing power $\rho$. Moreover, from
\eqref{eq_SE_theo}, the reduced communications power decreases the
term $\lambdabf_k({\bm{\vartheta}})\gamma_k$
in the numerator, while an increase in $\rho$ makes the term
$N_{\mathtt{t}}\beta_k\rho$ in the denominator to grow; this decreases the
achievable rate. 
We can conclude that the target-angle errors affect the rate indirectly through the communications and sensing power allocation tradeoff, even though they do not explicitly appear in the achievable-rate expression in~\eqref{eq_SE_theo}.
\end{remark}

The CRLB expressions in \eqref{eq_closed_crlb_theta} and \eqref{eq_closed_crlb_phi} hold for both deterministic and random target-angle errors. Since $\boldsymbol{\varepsilon}$ is modeled as a random vector, the corresponding CRLBs are also random. This motivates the use of the CRLB outage probability, which characterizes the sensing reliability in the presence of target-angle uncertainty.

\subsection{Angle-Error Models and CRLB Outage Analysis}
\label{section:angleerror_outage}
To characterize sensing reliability under target-angle uncertainty, we model the azimuth and elevation angle errors, denoted by $\varepsilon_{\theta}$ and $\varepsilon_{\phi}$, respectively, as independent random variables~\cite{yu20073},\cite{aubry2024sensor}. {We consider the following error models}: 
\begin{itemize}
\item \emph{Gaussian distribution:} The target-angle error is modeled as
$\varepsilon_{\psi}\sim\mathcal{N}(0,\sigma_{\psi}^{2})$,
$\psi\in\{\theta,\phi\}$, where $\sigma_{\psi}$ denotes the standard deviation of the angle error \cite{yu20073},\cite{aubry2024sensor}. 

\item \emph{Generalized uniform distribution:} The angle error is assumed to be
uniformly distributed as
$\varepsilon_{\psi}\sim\mathcal{U}\bigl(-\frac{\pi}{U_{\psi}},\frac{\pi}{U_{\psi}}\bigr)$,
where $\frac{\pi}{U_{\psi}}$ is the maximum absolute error \cite{uniyal2026outage}. A larger value of
$U_{\psi}$ implies a narrower error interval.

\item \emph{von Mises distribution:} The target-angle error follows the von Mises distribution, i.e.,
$\varepsilon_{\psi}\sim\mathcal{VM}(0,\kappa_{\psi})$, where $\kappa_{\psi}$ is the concentration
parameter \cite{uniyal2026outage}. Larger $\kappa_{\psi}$ values indicate a higher concentration of angle error around zero.
\end{itemize}
Since the CRLBs in \eqref{eq_closed_crlb_theta} and \eqref{eq_closed_crlb_phi} depend on the random angle-error vector $\boldsymbol{\varepsilon}$, they become random quantities under the above models. To quantify the resulting sensing reliability, we define the CRLB outage probability as
\begin{align}
\mathsf{P}_{\mathtt{out}_\psi}^{\mathtt{bf}}
(\boldsymbol{\vartheta},\boldsymbol{\gamma},\rho)
&\triangleq
\mathsf{P}
\left\{
\mathtt{CRLB}_{\psi}^{\mathtt{bf}}
(\boldsymbol{\vartheta},\boldsymbol{\gamma},\rho,\boldsymbol{\varepsilon})
>
\mathtt{CRLB}_{\psi}^{0}
\right\}
\nonumber\\
&=
1-
\mathsf{P}
\left\{
\mathtt{CRLB}_{\psi}^{\mathtt{bf}}
(\boldsymbol{\vartheta},\boldsymbol{\gamma},\rho,\boldsymbol{\varepsilon})
\leq
\mathtt{CRLB}_{\psi}^{0}
\right\}
\nonumber\\
&=
1-
F_{\mathtt{CRLB}_{\psi}}^{\mathtt{bf}}
\left(
\mathtt{CRLB}_{\psi}^{0}
\right),
\label{eq:outage_definition}
\end{align}
where $\psi\in\{\theta,\phi\}$, $F_{\mathtt{CRLB}_{\psi}}^\mathtt{bf}(\cdot)$ is the cumulative distribution function (CDF) of $\mathtt{CRLB}_{\psi}^\mathtt{bf}$, and $\mathtt{CRLB}_{\psi}^{0}$ is the predefined CRLB threshold. To facilitate the derivation of $\mathsf{P}_{\mathtt{out}_\psi}^\mathtt{bf}$, we next provide a tractable approximation for $F_{\mathtt{CRLB}_{\psi}}^\mathtt{bf}(\cdot)$.
\begin{theorem}
\label{thm:random_crlb}
Under Gaussian, generalized-uniform, and von Mises target-angle error models, the CDF of the random CRLB admits the following unified approximation:
\begin{align}
F^\mathtt{bf}_{{\mathtt{CRLB}}_{\psi}}(x)
\!\approx\!
\sum_{i=1}^{G_\theta}\sum_{j=1}^{G_\phi}
\frac{w_{\theta_i}\,w_{\phi_j}}{1\!+\!\exp\!\left(
\!-\varrho\!\left(
\frac{1}{{\mathtt{CRLB}}_\psi^\mathtt{bf}(\bm{\vartheta},\bm{\gamma},\rho,\mathbf{z}_{ij})}\!-\!\frac{1}{x}
\right)\right)},
\label{eq:cdf_unified_1overcrlb}
\end{align}
where $x>0$,
$\psi \in \{\theta,\phi\}$,
$G_\theta$ and $G_\phi$ are the order of the Gaussian-quadrature rule. Furthermore,
$\varrho>0$ is the sharpness parameter of the sigmoid approximation \cite[(11)]{shi_sigmoid} and
$\vz_{ij}\triangleq
[z_{\theta_i},
z_{\phi_j}] ^\T$ and $\{w_{\theta_i},w_{\phi_j}\}$
are the nodes and weights of the Gauss-quadrature rule, respectively, given by
\begin{align}
(z_{\psi_q},w_{\psi_q})\!\triangleq
\begin{cases}
\big(\frac{\pi}{U_\psi} z_q^{\mathtt{GL}},\frac{1}{2}w_q^{\mathtt{GL}}\big),
& \hspace{-0.3cm}\varepsilon_\psi\sim\mathcal{U}(-\frac{\pi}{U_\psi},\frac{\pi}{U_\psi}),\\
\big(\sqrt{2}\sigma_\psi z_q^{\mathtt{GH}},\frac{1}{\sqrt{\pi}}w_q^{\mathtt{GH}}\big),
& \hspace{-0.3cm}\varepsilon_\psi\sim\mathcal{N}(0,\sigma_\psi^2),\\
\Big(\pi z_q^{\mathtt{GL}},\frac{1}{2 I_0(\kappa_\psi)} w_q^{\mathtt{GL}} \\ \hspace{0.1cm}\times 
\!\exp\!\big(\kappa_\psi\cos(\pi z_q^{\mathtt{GL}})\big)\!\Big),
& \hspace{-0.2cm}\varepsilon_\psi\sim\mathcal{VM}(0,\kappa_\psi),
\end{cases}
\label{eq:mapped_theta_nodes_weights}
\end{align}
with $q\in\{i,j\}$, $z_q^{\mathtt{GL}}$ and $z_q^{\mathtt{GH}}$ being the zeros of the Legendre polynomial and Hermite polynomial, respectively, while
$w_q^{\mathtt{GL}}$ and $w_q^{\mathtt{GH}}$ are the respective Gauss--Legendre and Gauss--Hermite weights \cite[(25.4.29)]{abram_gauss_quadrature}, \cite[Proposition 2]{Marco_di_renzo}.
\end{theorem}
\begin{proof}
 See Appendix~\ref{app:random_crlb}. \hfill$\blacksquare$
\end{proof}
By Theorem~\ref{thm:random_crlb}, the CRLB outage probability can be readily obtained by substituting \eqref{eq:cdf_unified_1overcrlb} into \eqref{eq:outage_definition}.
\begin{corollary}
\label{rem:impact}
As $\sigma_{\psi}\to0$, $U_{\psi}\to\infty$, or $\kappa_{\psi}\to\infty$,
the CRLB outage probability in \eqref{eq:outage_definition} approaches zero
for a sufficiently large sharpness parameter $\varrho$ provided that
$\mathtt{CRLB}_{\psi}^{\mathtt{bf}}(\mathbf{0})<\mathtt{CRLB}_{\psi}^{0}$, where $\mathtt{CRLB}_{\psi}^{\mathtt{bf}}(\mathbf{0})$ denotes the CRLB in the absence of target-angle errors.
\end{corollary}
\begin{proof} 
As $\sigma_{\psi}\to0$ or $U_{\psi}\to\infty$, the quadrature nodes in \eqref{eq:mapped_theta_nodes_weights} satisfy
$\mathbf{z}_{ij}\to\mathbf{0}$. For the von Mises distribution, the corresponding quadrature weights become increasingly concentrated around zero error as$\kappa_{\psi}\to\infty$. Consequently, for all three angle-error models, the term $\mathtt{CRLB}_{\psi}^{\mathtt{bf}}$ in \eqref{eq:cdf_unified_1overcrlb} satisfies
$\mathtt{CRLB}_{\psi}^{\mathtt{bf}}(\mathbf{z}_{ij})
\to\mathtt{CRLB}_{\psi}^{\mathtt{bf}}(\mathbf{0})$ and we obtain
$F_{\mathtt{CRLB}_{\psi}}^{\mathtt{bf}}(\mathtt{CRLB}_{\psi}^{0})
\to
\sum_{i=1}^{G_{\theta}}\sum_{j=1}^{G_{\phi}}w_{\theta_i}w_{\phi_j}
\bigl(1+\exp(-\varrho d_{\psi}^{(0)})\bigr)^{-1}$, where
$d_{\psi}^{(0)}
\triangleq
1/\mathtt{CRLB}_{\psi}^{\mathtt{bf}}(\mathbf{0})
-1/\mathtt{CRLB}_{\psi}^{0}$. If
$\mathtt{CRLB}_{\psi}^{\mathtt{bf}}(\mathbf{0})
<\mathtt{CRLB}_{\psi}^{0}$, then
$d_{\psi}^{(0)}>0$, and thus
$\exp(-\varrho d_{\psi}^{(0)})\to0$ for a sufficiently large
$\varrho$. Hence,
$F_{\mathtt{CRLB}_{\psi}}^{\mathtt{bf}}(\mathtt{CRLB}_{\psi}^{0})
\to
\sum_{i=1}^{G_{\theta}}\sum_{j=1}^{G_{\phi}}w_{\theta_i}w_{\phi_j}$.
Consequently, from \eqref{eq:outage_definition}, we have
$\mathsf{P}_{\mathtt{out}_{\psi}}^{\mathtt{bf}}
\to
1-\sum_{i=1}^{G_{\theta}}\sum_{j=1}^{G_{\phi}}w_{\theta_i}w_{\phi_j}
\approx0$. Here, the approximation follows since the azimuth and
elevation angle errors are independently distributed and
$\sum_i w_{\psi_i}\approx1$. For example, in the Gaussian case, using
the Gauss--Hermite rule~\cite[(25.4.46)]{abram_gauss_quadrature} yields
$\sum_i w_{\psi_i}\approx(1/\sqrt{\pi})
\int_{-\infty}^{\infty}e^{-x^2}dx=1$. This result also holds for the other considered distributions.
\hfill$\blacksquare$
\end{proof}
\begin{remark}
\label{rem:crlb antennas}
From Corollary~\ref{rem:asymptotic_Nt}, under the equal-power allocation and
square UPAs, we have $\mathtt{CRLB}_{\psi}^{\mathtt{bf}}\to0$ for sufficiently large
$N_{\mathtt{t}}$ or  $N_{\mathtt{r}}$. Substituting this result into
\eqref{eq:cdf_unified_1overcrlb} and \eqref{eq:outage_definition} shows that
$\mathsf{P}_{\mathtt{out}_\psi}^{\mathtt{bf}}
\to 1-\sum_{i=1}^{G_{\theta}}\sum_{j=1}^{G_{\phi}}w_{\theta_i}w_{\phi_j}
\approx0$. 
\end{remark}

Although derived under the equal-power allocation assumption, this asymptotic behavior also holds under more general power allocation schemes, as verified in Section~\ref{section:numerical_results}. However, equal power allocation is generally suboptimal for balancing communications and sensing performance. Furthermore, as noted in Remark~\ref{rem:tradeoff}, target-angle errors do not explicitly affect the achievable rate and only degrade the sensing performance under a fixed power allocation. Their impact on the communications--sensing tradeoff therefore manifests through power allocation, motivating the robust design developed next.

\section{Reliability-Constrained Robust Power Allocation}\label{sec_opt}
\subsection{Problem Formulation}
We now develop a robust communications-centric power allocation scheme under target-angle uncertainty. The pilot-training, communications-transmission, and sensing-transmission powers are jointly optimized to maximize the communications sum rate while satisfying sensing outage and total-power constraints.

From the results in Section~\ref{sec_perf_analysis}, the achievable rate, CRLBs, and CRLB outage probabilities depend on the sensing-power allocation vector $\{\eta_k\}_{k=1}^{K}$ only through the total sensing power $\rho$.
Hence, the optimization of $\{\gamma_k,\eta_k\}_{k=1}^{K}$ is reduced to optimizing
$\{\gamma_k\}_{k=1}^{K}$ and $\rho$. The resulting optimization problem is formulated as
\begin{subequations}\label{prob:robust_power_allocation}
\begin{IEEEeqnarray}{rcl}
& \underset{\substack{\boldsymbol{\vartheta},\,\boldsymbol{\gamma},\,\rho}}
{\mathrm{max}}
& \quad
\sum_{k=1}^{K}
\mathcal{R}_{k}^{\mathtt{bf}}
(\boldsymbol{\vartheta},\boldsymbol{\gamma},\rho)
\label{prob:P0_obj}
\\
& \mathrm{s.t.}
& \quad
\mathsf{P}_{\mathtt{out}_\theta}^{\mathtt{bf}}
(\boldsymbol{\vartheta},\boldsymbol{\gamma},\rho)
\leq
{\mathsf{P}^{\mathtt{0}}_{\theta}},
\label{prob:P0_crlb_theta}
\\
& &
\quad
\mathsf{P}_{\mathtt{out}_\phi}^{\mathtt{bf}}
(\boldsymbol{\vartheta},\boldsymbol{\gamma},\rho)
\leq
{\mathsf{P}^{\mathtt{0}}_{\phi}},
\label{prob:P0_crlb_phi}
\\
& &
\quad
\bm{1}_K^{\T}\boldsymbol{\vartheta}
+
N_{\mathtt{t}}
\left(
\boldsymbol{\xi}_{\mathtt{bf}}^{\T}(\boldsymbol{\vartheta})
\boldsymbol{\gamma}
+
\rho
\right)
\leq
P_{\mathtt{t}}^{\max},
\label{prob:P0_power}
\end{IEEEeqnarray}
\end{subequations}  
where $\mathsf{P}^{0}_{\theta}$ and $\mathsf{P}^{0}_{\phi}$ in
\eqref{prob:P0_crlb_theta} and \eqref{prob:P0_crlb_phi} denote the outage thresholds for the azimuth and elevation CRLBs, respectively, and $P_{\mathtt t}^{\max}$ in
\eqref{prob:P0_power} is the total power budget. Specifically, constraint \eqref{prob:P0_power} captures the coupling between uplink pilot training and downlink communications/sensing transmission, since increasing the pilot-training power improves the channel-estimation quality and hence the downlink communications performance, while reducing the power available for subsequent communications/sensing transmission.

Problem~\eqref{prob:robust_power_allocation} is challenging to solve due to its nonconvexity. Specifically, the objective function \eqref{prob:P0_obj} is nonconcave in $(\boldsymbol{\vartheta},\boldsymbol{\gamma},\rho)$, the outage constraints \eqref{prob:P0_crlb_theta} and~\eqref{prob:P0_crlb_phi} are nonconvex, and the power constraint  \eqref{prob:P0_power} is nonaffine due to the fractional dependence of $\boldsymbol{\xi}_{\mathtt{bf}}(\boldsymbol{\vartheta})$ on $\boldsymbol{\vartheta}$. 
To address these difficulties, we develop an alternating-optimization (AO) framework based on the SCA \cite{sun2016majorization}
, which decomposes the original problem into tractable power allocation subproblems for pilot training and signal transmission/probing and solves them iteratively.

\subsection{Pilot-Power Allocation via SCA}\label{subsec:nu}
For fixed $(\bar{\boldsymbol{\gamma}},\bar{\rho})$, we first optimize the pilot-power allocation vector $\boldsymbol{\vartheta}$. Since both the objective function and constraints in \eqref{prob:robust_power_allocation} are nonconvex w.r.t. $\boldsymbol{\vartheta}$, we employ an SCA framework that replaces them with tractable quadratic surrogate functions at each iteration.

\subsubsection{Convexity of the Objective Function~\eqref{prob:P0_obj}}
Let $x^{[t]}$ be the feasible point of $x$ at the $t$-th iteration of an iterative algorithm presented next. By \cite[Lemma~12]{sun2016majorization}, we obtain the concave quadratic lower-bound surrogate around the feasible point $\boldsymbol{\vartheta}^{[t]}$ as
\begin{align}
\mathcal{R}_{k}^{\mathtt{bf}}(\boldsymbol{\vartheta})
&\geq
\widetilde{\mathcal{R}}_{k,\boldsymbol{\vartheta}}^{\mathtt{bf}}
\bigl(\boldsymbol{\vartheta};\boldsymbol{\vartheta}^{[t]}\bigr)
\notag\\
&\triangleq
\mathcal{R}_{k}^{\mathtt{bf}}(\boldsymbol{\vartheta}^{[t]})
+
\nabla_{\boldsymbol{\vartheta}}
\mathcal{R}_{k}^{\mathtt{bf}}(\boldsymbol{\vartheta}^{[t]})^{\T}
\bigl(\boldsymbol{\vartheta}-\boldsymbol{\vartheta}^{[t]}\bigr)
\notag\\
&\quad
-\frac{L_{k,\boldsymbol{\vartheta}}^{\mathtt{bf},[t]}}{2}
\bigl\|\boldsymbol{\vartheta}-\boldsymbol{\vartheta}^{[t]}\bigr\|^{2},
\label{eq:rate_lb_nu}
\end{align}
where
$\nabla_{\boldsymbol{\vartheta}}
\mathcal{R}_{k}^{\mathtt{bf}}(\cdot)$
is the gradient of the rate function
$\mathcal{R}_{k}^{\mathtt{bf}}$ w.r.t. $\boldsymbol{\vartheta}$, and
$L_{k,\boldsymbol{\vartheta}}^{\mathtt{bf},[t]}>0$ is a constant chosen to satisfy
$\mathcal{R}_{k}^{\mathtt{bf}}(\boldsymbol{\vartheta})
\geq
\widetilde{\mathcal{R}}_{k,\boldsymbol{\vartheta}}^{\mathtt{bf}}
(\boldsymbol{\vartheta};\boldsymbol{\vartheta}^{[t]})$.

\subsubsection{Convexity of Constraints \eqref{prob:P0_crlb_theta}--\eqref{prob:P0_power}}

We convexify \eqref{prob:P0_crlb_theta}--\eqref{prob:P0_power} using quadratic upper-bound surrogates
\cite{sun2016majorization}. Specifically, for $\psi\in\{\theta,\phi\}$, we construct an upper-bound
surrogate for \eqref{prob:P0_crlb_theta} and \eqref{prob:P0_crlb_phi}  around feasible
point $\boldsymbol{\vartheta}^{[t]}$ as
\begin{align}
\mathsf{P}_{\mathtt{out}_\psi}^{\mathtt{bf}}(\boldsymbol{\vartheta})
&\leq
\widetilde{\mathsf{P}}_{\mathtt{out}_{\psi,\boldsymbol{\vartheta}}}^{\mathtt{bf}}
\bigl(\boldsymbol{\vartheta};\boldsymbol{\vartheta}^{[t]}\bigr)
\notag\\
&\triangleq
\mathsf{P}_{\mathtt{out}_\psi}^{\mathtt{bf}}
(\boldsymbol{\vartheta}^{[t]})
+
\nabla_{\boldsymbol{\vartheta}}
\mathsf{P}_{\mathtt{out}_\psi}^{\mathtt{bf}}
(\boldsymbol{\vartheta}^{[t]})^{\T}
\bigl(\boldsymbol{\vartheta}-\boldsymbol{\vartheta}^{[t]}\bigr)
\notag\\
&\quad
+\frac{\ell_{\psi,\boldsymbol{\vartheta}}^{\mathtt{bf},[t]}}{2}
\bigl\|\boldsymbol{\vartheta}-\boldsymbol{\vartheta}^{[t]}\bigr\|^{2} \leq \mathsf{P}^{\mathtt{0}}_{\psi}, \quad \psi\in\{\theta,\phi\}.
\label{eq:pout_ub_nu}
\end{align}
For~\eqref{prob:P0_power}, we first define
$P_{\mathtt{tot}}^{\mathtt{bf}}(\boldsymbol{\vartheta})
\triangleq
\bm{1}_K^{\T}\boldsymbol{\vartheta}
+
N_{\mathtt{t}}
(
\boldsymbol{\xi}_{\mathtt{bf}}^{\T}(\boldsymbol{\vartheta})
\boldsymbol{\gamma}
+
\rho
)$ and similarly to \eqref{eq:pout_ub_nu},  \eqref{prob:P0_power} is iteratively updated as
\begin{align}
{P_{\mathtt{tot}}^{\mathtt{bf}}(\boldsymbol{\vartheta})}
&\leq
\widetilde{P}_{\mathtt{tot}_{\boldsymbol{\vartheta}}}
\bigl(\boldsymbol{\vartheta};\boldsymbol{\vartheta}^{[t]}\bigr)
\notag\\
&\triangleq
P_{\mathtt{tot}}^{\mathtt{bf}}(\boldsymbol{\vartheta}^{[t]})
+
\nabla_{\boldsymbol{\vartheta}}
P_{\mathtt{tot}}^{\mathtt{bf}}(\boldsymbol{\vartheta}^{[t]})^{\T}
\bigl(\boldsymbol{\vartheta}-\boldsymbol{\vartheta}^{[t]}\bigr)
\notag\\
&\quad
+\frac{\ell_{\mathtt{p},\boldsymbol{\vartheta}}^{\mathtt{bf},[t]}}{2}
\bigl\|\boldsymbol{\vartheta}-\boldsymbol{\vartheta}^{[t]}\bigr\|^{2} \leq P_{\mathtt{t}}^{\max},
\label{eq:ptot_ub_nu}
\end{align}
 where
$\nabla_{\boldsymbol{\vartheta}}
\mathsf{P}_{\mathtt{out}_\psi}^{\mathtt{bf}}
(\cdot)$
and
$\nabla_{\boldsymbol{\vartheta}}
P_{\mathtt{tot}}^{\mathtt{bf}}(\cdot)$
denote the gradients of the scalar functions
$\mathsf{P}_{\mathtt{out}_\psi}^{\mathtt{bf}}$
and $P_{\mathtt{tot}}^{\mathtt{bf}}$, respectively, w.r.t.
 $\boldsymbol{\vartheta}$. Moreover,
$\ell_{\psi,\boldsymbol{\vartheta}}^{\mathtt{bf},[t]}>0$ and
$\ell_{\mathtt{p},\boldsymbol{\vartheta}}^{\mathtt{bf},[t]}>0$ are chosen to satisfy
$\mathsf{P}_{\mathtt{out}_\psi}^{\mathtt{bf}}(\boldsymbol{\vartheta})
\leq
\widetilde{\mathsf{P}}_{\mathtt{out}_{\psi,\boldsymbol{\vartheta}}}^{\mathtt{bf}}
(\boldsymbol{\vartheta};\boldsymbol{\vartheta}^{[t]})$,
for $\psi\in\{\theta,\phi\}$, and
$P_{\mathtt{tot}}^{\mathtt{bf}}(\boldsymbol{\vartheta})
\leq
\widetilde{P}_{\mathtt{tot}_{\boldsymbol{\vartheta}}}^{\mathtt{bf}}
(\boldsymbol{\vartheta};\boldsymbol{\vartheta}^{[t]})$, respectively.

Summing up, the convex approximation of the pilot-power allocation
subproblem in \eqref{prob:robust_power_allocation} solved  at iteration $t$ is given as
\begin{align}
\mathcal{P}_{1}:\quad
\underset{\boldsymbol{\vartheta}}{\mathrm{max}}
\sum_{k=1}^{K}
\widetilde{\mathcal{R}}_{k,\boldsymbol{\vartheta}}^{\mathtt{bf}}
\bigl(\boldsymbol{\vartheta};\boldsymbol{\vartheta}^{[t]}\bigr)
\ \mathrm{s.t.}\ 
\eqref{eq:pout_ub_nu},\;
\eqref{eq:ptot_ub_nu}.
\label{prob:P1_nu}
\end{align}

\subsection{Communications and Sensing Power Allocation}
\label{subsec:gamma_rho}

After obtaining the optimal $\boldsymbol{\vartheta}^\star$, we
aim to optimize the communications and sensing power coefficients, \textit{i.e.},
$(\boldsymbol{\gamma},\rho)$.

\subsubsection{Convexity of the Objective Function~\eqref{prob:P0_obj}}

We first rewrite \eqref{eq_SE_theo} as
$
\mathcal{R}_{k}^{\mathtt{bf}}(\boldsymbol{\gamma},\rho)
=
{\tau_\mathtt{0}}\log_{2}
\left(
1+
\frac{\mathcal{S}_{k}^{\mathtt{bf}}(\boldsymbol{\gamma})}
{\mathcal{I}_{k}^{\mathtt{bf}}(\boldsymbol{\gamma},\rho)}
\right)$,
where
$\mathcal{S}_{k}^{\mathtt{bf}}(\boldsymbol{\gamma})
\triangleq
\lambda_{\mathtt{bf},k}(\boldsymbol{\vartheta})\gamma_{k}$
and
$\mathcal{I}_{k}^{\mathtt{bf}}(\boldsymbol{\gamma},\rho)
\triangleq
N_{\mathtt{t}}\beta_{k}\rho
+
N_{\mathtt{t}}
\boldsymbol{\zeta}_{\mathtt{bf},k}^{\T}(\boldsymbol{\vartheta})
\boldsymbol{\gamma}
+
\sigma_{\mathtt{c}}^{2}$.
We note that both
$\mathcal{S}_{k}^{\mathtt{bf}}(\boldsymbol{\gamma})$ and
$\mathcal{I}_{k}^{\mathtt{bf}}(\boldsymbol{\gamma},\rho)$ are affine in
$(\boldsymbol{\gamma},\rho)$. Thus, by following the lower-bound
transformation in~\cite[Eq. (22)]{nguyen2022short}, we obtain a concave
surrogate around the feasible point
$(\boldsymbol{\gamma}^{[t]},\rho^{[t]})$ as
\begin{align}
\!\!\!\!\mathcal{R}_{k}^{\mathtt{bf}}(\boldsymbol{\gamma},\rho)
&\geq
\widetilde{\mathcal{R}}_{k,\boldsymbol{\gamma},\rho}^{\mathtt{bf}}
\bigl(\boldsymbol{\gamma},\rho;
\boldsymbol{\gamma}^{[t]},\rho^{[t]}\bigr)
\notag\\
&\triangleq
\frac{{\tau_\mathtt{0}}}{\ln 2}
\left(\!
q_{\mathtt{u},k}^{\mathtt{bf},[t]}
-
\frac{q_{\mathtt{z},k}^{\mathtt{bf},[t]}}
{\mathcal{S}_{k}^{\mathtt{bf}}(\boldsymbol{\gamma})}
-
q_{\mathtt{w},k}^{\mathtt{bf},[t]}
\mathcal{I}_{k}^{\mathtt{bf}}(\boldsymbol{\gamma},\rho)
\!\right),
\label{eq:rate_lb_gamma_rho}
\end{align}
where
\begin{align}
q_{\mathtt{u},k}^{\mathtt{bf},[t]}
&\triangleq
\ln\left(
1+
\frac{\mathcal{S}_{k}^{\mathtt{bf},[t]}}
{\mathcal{I}_{k}^{\mathtt{bf},[t]}}
\right)
+
\frac{2\mathcal{S}_{k}^{\mathtt{bf},[t]}}
{\mathcal{S}_{k}^{\mathtt{bf},[t]}
+
\mathcal{I}_{k}^{\mathtt{bf},[t]}},
\label{eq:qu_gamma_rho}\\
q_{\mathtt{z},k}^{\mathtt{bf},[t]}
&\triangleq
\frac{\bigl(\mathcal{S}_{k}^{\mathtt{bf},[t]}\bigr)^{2}}
{\mathcal{S}_{k}^{\mathtt{bf},[t]}
+
\mathcal{I}_{k}^{\mathtt{bf},[t]}},
\label{eq:qv_gamma_rho}\\
q_{\mathtt{w},k}^{\mathtt{bf},[t]}
&\triangleq
\frac{\mathcal{S}_{k}^{\mathtt{bf},[t]}}
{
\bigl(
\mathcal{S}_{k}^{\mathtt{bf},[t]}
+
\mathcal{I}_{k}^{\mathtt{bf},[t]}
\bigr)
\mathcal{I}_{k}^{\mathtt{bf},[t]}},
\label{eq:qw_gamma_rho}
\end{align}
with $\mathcal{S}_{k}^{\mathtt{bf},[t]}
\triangleq
\mathcal{S}_{k}^{\mathtt{bf}}(\boldsymbol{\gamma}^{[t]})$ and 
$\mathcal{I}_{k}^{\mathtt{bf},[t]}
\triangleq
\mathcal{I}_{k}^{\mathtt{bf}}(\boldsymbol{\gamma}^{[t]},\rho^{[t]})$.

\subsubsection{Convexity of Constraints \eqref{prob:P0_crlb_theta} and \eqref{prob:P0_crlb_phi}}

Similar to~\eqref{eq:pout_ub_nu}, for $\psi\in\{\theta,\phi\}$, we construct
a quadratic upper-bound surrogate for 
\eqref{prob:P0_crlb_theta} and \eqref{prob:P0_crlb_phi} around the feasible point
$(\boldsymbol{\gamma}^{[t]},\rho^{[t]})$ as follows
\begin{align}
\mathsf{P}_{\mathtt{out}_\psi}^{\mathtt{bf}}
(\boldsymbol{\gamma},\rho)
&\leq
\widetilde{\mathsf{P}}_{\mathtt{out}_{\psi,\boldsymbol{\gamma},\rho}}^{\mathtt{bf}}
\bigl(
\boldsymbol{\gamma},\rho;
\boldsymbol{\gamma}^{[t]},\rho^{[t]}
\bigr)
\notag\\
&\triangleq
\mathsf{P}_{\mathtt{out}_\psi}^{\mathtt{bf}}
\bigl(
\boldsymbol{\gamma}^{[t]},\rho^{[t]}
\bigr)
\notag\\
&+
\nabla_{\boldsymbol{\gamma}}
\mathsf{P}_{\mathtt{out}_\psi}^{\mathtt{bf}}
\bigl(
\boldsymbol{\gamma}^{[t]},\rho^{[t]}
\bigr)^{\T}
\bigl(\boldsymbol{\gamma}-\boldsymbol{\gamma}^{[t]}\bigr)
\notag\\
&+
\frac{\partial
\mathsf{P}_{\mathtt{out}_\psi}^{\mathtt{bf}}
\bigl(
\boldsymbol{\gamma}^{[t]},\rho^{[t]}
\bigr)}
{\partial \rho}
\bigl(\rho-\rho^{[t]}\bigr)
\notag\\
& +\!\!
\frac{\ell_{\psi,\boldsymbol{\gamma},\rho}^{\mathtt{bf},[t]}}{2}
\!\!\left(\!
\bigl\|\boldsymbol{\gamma}-\boldsymbol{\gamma}^{[t]}\bigr\|^{2}
+
\bigl|\rho-\rho^{[t]}\bigr|^{2}
\right) \!\leq \mathsf{P}^{\mathtt{0}}_{\psi},
\label{eq:pout_ub_gamma_rho}
\end{align}
where 
$\nabla_{\boldsymbol{\gamma}}
\mathsf{P}_{\mathtt{out}_\psi}^{\mathtt{bf}}
(\cdot)$
and
$\partial
\mathsf{P}_{\mathtt{out}_\psi}^{\mathtt{bf}}
(\cdot)/\partial \rho$
denote the gradient and derivative of the scalar function
$\mathsf{P}_{\mathtt{out}_\psi}^{\mathtt{bf}}$
w.r.t.  $\boldsymbol{\gamma}$ and 
$\rho$, respectively, and
$\ell_{\psi,\boldsymbol{\gamma},\rho}^{\mathtt{bf},[t]}>0$
is chosen to satisfy
$\mathsf{P}_{\mathtt{out}_\psi}^{\mathtt{bf}}(\boldsymbol{\gamma},\rho)
\leq
\widetilde{\mathsf{P}}_{\mathtt{out}_{\psi,\boldsymbol{\gamma},\rho}}^{\mathtt{bf}}
(\boldsymbol{\gamma},\rho;
\boldsymbol{\gamma}^{[t]},\rho^{[t]})$.

From the above analysis, the convex approximation of the communications and sensing allocation subproblem in \eqref{prob:robust_power_allocation} solved at iteration $t$ is given as
\begin{align}
\label{prob:p2}
\mathcal{P}_{2}:\quad
\underset{\boldsymbol{\gamma},\rho}{\mathrm{max}}
\sum_{k=1}^{K}
\widetilde{\mathcal{R}}_{k,\boldsymbol{\gamma},\rho}^{\mathtt{bf}}
\bigl(
\boldsymbol{\gamma},\rho;
\boldsymbol{\gamma}^{[t]},\rho^{[t]}
\bigr)
\ \mathrm{s.t.} \
\eqref{prob:P0_power},\; \eqref{eq:pout_ub_gamma_rho}.
\end{align}
The proposed solution requires the closed-form gradients of
$\mathcal{R}_{k}^{\mathtt{bf}}$,
$\mathsf{P}_{\mathtt{out}_\psi}^{\mathtt{bf}}$, and
$P_{\mathtt{tot}}^{\mathtt{bf}}$ w.r.t. the corresponding
optimization variables in
\eqref{eq:rate_lb_nu}--\eqref{eq:ptot_ub_nu}, and \eqref{eq:pout_ub_gamma_rho}. Owing to the derived closed-form expressions in
\eqref{eq_SE_theo}, \eqref{eq:cdf_unified_1overcrlb}, and
\eqref{prob:P0_power}, these gradients can be directly computed by applying
the chain rule, as summarized in the following theorem.
\begin{theorem}
\label{thm:gradients}
For $\mathtt{bf}\in\{\mathtt{MRT},\mathtt{ZF}\}$ and
$k,n\in\{1,\ldots,K\}$, the entries of
$\nabla_{\boldsymbol{\vartheta}}\mathcal{R}_{k}^{\mathtt{bf}}$ and
$\nabla_{\boldsymbol{\vartheta}}P_{\mathtt{tot}}^{\mathtt{bf}}$ are given by
\begin{align}
\frac{\partial \mathcal{R}_{k}^{\mathtt{bf}}}{\partial \vartheta_n}
&=
\frac{{\tau_\mathtt{0}}}{\ln 2}
\frac{
\mathcal{I}_{k}^{\mathtt{bf}}U_{k,n}^{\mathtt{bf}}
-
\mathcal{S}_{k}^{\mathtt{bf}}V_{k,n}^{\mathtt{bf}}
}{
\mathcal{I}_{k}^{\mathtt{bf}}
\left(
\mathcal{S}_{k}^{\mathtt{bf}}
+
\mathcal{I}_{k}^{\mathtt{bf}}
\right)
},
\label{eq:grad_Rk_vartheta}
\\
\frac{\partial P_{\mathtt{tot}}^{\mathtt{bf}}}{\partial \vartheta_n}
&=
1
+
N_{\mathtt{t}}\Omega_n^{\mathtt{bf}},
\label{eq:grad_Ptot_vartheta}
\end{align}
where
{$\Omega_n^{\mathtt{bf}}
\triangleq
\partial s^{\mathtt{bf}}/\partial\vartheta_n$, with
$s^{\mathtt{bf}}
\triangleq
\boldsymbol{\xi}_{\mathtt{bf}}^{\T}(\boldsymbol{\vartheta})
\boldsymbol{\gamma}$ and}
\begin{align*}
U_{k,n}^{\mathtt{MRT}}
&=
2N_{\mathtt{t}}^{2}\xi_k\gamma_k\mathcal{D}_{k,n}, \hspace{0.2cm}
U_{k,n}^{\mathtt{ZF}}
=
0,
\\
V_{k,n}^{\mathtt{MRT}}
&=
N_{\mathtt{t}}\beta_k\Omega_n^{\mathtt{MRT}},
\hspace{0.78cm}
V_{k,n}^{\mathtt{ZF}}
=
N_{\mathtt{t}}\sum_{j=1}^{K}\gamma_j\Xi_{j,k,n},
\\
\Omega_n^{\mathtt{MRT}}
&=
\sum_{j=1}^{K}\gamma_j\mathcal{D}_{j,n},
\hspace{0.78cm}
\Omega_n^{\mathtt{ZF}}
=
-\frac{1}{N_{\mathtt{t}}\left(N_{\mathtt{t}}-K\right)}
\sum_{j=1}^{K}
\frac{\gamma_j\mathcal{D}_{j,n}}{\xi_j^2}.
\end{align*}
Herein, $\mathcal{D}_{k,n}$ and $\Xi_{j,k,n}$ are calculated as
\begin{align*}
&\!\!\mathcal{D}_{k,n}
\!=\!
\frac{
\!\!\tau_{\mathtt{p}}\beta_k^2\!\!
\left(
\!\delta_{kn}\!\!
\left(
\!\tau_{\mathtt{p}}\!\sum_{\ell=1}^{K}
\!\!\vartheta_{\ell}\beta_{\ell}\abs{\vp_\ell^\H \vp_k}^2
\!\!+\!
\sigma^2\!
\right)
\!\!-\!
\tau_{\mathtt{p}}\vartheta_k\beta_n \!\abs{\vp_n^\H \vp_k}^2\!
\right)
}{
\left(
\tau_{\mathtt{p}}\sum_{\ell=1}^{K}
\vartheta_{\ell}\beta_{\ell}\abs{\vp_\ell^\H \vp_k}^2
+
\sigma^2
\right)^2
},
\\
&\!\!\!\Xi_{j,k,n}
\!=\!
-\frac{1}{N_{\mathtt{t}}\left(N_{\mathtt{t}}-K\right)}
\left(
\frac{\mathcal{D}_{k,n}}{\xi_j}
+
\frac{\epsilon_k\mathcal{D}_{j,n}}{\xi_j^2}
\right),
\end{align*}
where $\delta_{kn}=1$ if $k=n$ and $\delta_{kn}=0$ otherwise. Furthermore, the entries of
$\nabla_{\boldsymbol{\vartheta}}
\mathsf{P}_{\mathtt{out}_\psi}^{\mathtt{bf}}$ and
$\nabla_{\boldsymbol{\gamma}}
\mathsf{P}_{\mathtt{out}_\psi}^{\mathtt{bf}}$, and the derivative
$\partial\mathsf{P}_{\mathtt{out}_\psi}^{\mathtt{bf}}/\partial\rho$,
are computed as
\begin{align}
\frac{\partial
\mathsf{P}_{\mathtt{out}_\psi}^{\mathtt{bf}}}
{\partial \vartheta_n}
&\!=\!
-\!\!\!\sum_{i=1}^{G_{\theta}}\sum_{j=1}^{G_{\phi}}
\!\!w_{\theta_i}w_{\phi_j}
\varrho
\varsigma_{\psi,ij}^{\mathtt{bf}}
\left(
1-\varsigma_{\psi,ij}^{\mathtt{bf}}
\right)
\Lambda_{\psi,ij}^{\mathtt{s}}
\Omega_n^{\mathtt{bf}},
\label{eq:grad_Pout_vartheta}
\\
\frac{\partial
\mathsf{P}_{\mathtt{out}_\psi}^{\mathtt{bf}}}
{\partial \gamma_n}
&\!=\!
-\!\!\!\sum_{i=1}^{G_{\theta}}\sum_{j=1}^{G_{\phi}}
\!\!w_{\theta_i}w_{\phi_j}
\varrho
\varsigma_{\psi,ij}^{\mathtt{bf}}
\!\left(
1\!-\!\varsigma_{\psi,ij}^{\mathtt{bf}}
\right)
\!\Lambda_{\psi,ij}^{\mathtt{s}}
[\boldsymbol{\xi}_{\mathtt{bf}}]_n,
\label{eq:grad_Pout_gamma}
\\
\frac{\partial
\mathsf{P}_{\mathtt{out}_\psi}^{\mathtt{bf}}}
{\partial \rho}
&\!=\!
-\sum_{i=1}^{G_{\theta}}\sum_{j=1}^{G_{\phi}}
w_{\theta_i}w_{\phi_j}
\varrho
\varsigma_{\psi,ij}^{\mathtt{bf}}
\left(
1-\varsigma_{\psi,ij}^{\mathtt{bf}}
\right)
\Lambda_{\psi,ij}^{\mathtt{r}},
\label{eq:grad_Pout_rho}
\end{align}
where we have defined
\begin{align*}
\varsigma_{\psi,ij}^{\mathtt{bf}}
&\triangleq
\frac{1}{
1+\exp\left(
-\varrho
\left(
\frac{1}
{\mathtt{CRLB}_{\psi}^{\mathtt{bf}}
(s^{\mathtt{bf}},\rho,\boldsymbol{z}_{ij})}
-
\frac{1}{\mathtt{CRLB}_{\psi}^{0}}
\right)
\right)
},
\end{align*}
\begin{align*}
\Lambda_{\theta,ij}^{\mathtt{x}}
&\triangleq
\tilde{c}_{\theta\theta,ij}^{\mathtt{x}}
-
\frac{
2\tilde{T}_{\theta\phi,ij}\,\tilde{c}_{\theta\phi,ij}^{\mathtt{x}}\,\tilde{T}_{\phi\phi,ij}
-
\tilde{T}_{\theta\phi,ij}^{2}\,\tilde{c}_{\phi\phi,ij}^{\mathtt{x}}
}{
\tilde{T}_{\phi\phi,ij}^{2}
},
\\
\Lambda_{\phi,ij}^{\mathtt{x}}
&\triangleq
\tilde{c}_{\phi\phi,ij}^{\mathtt{x}}
-
\frac{
2\tilde{T}_{\theta\phi,ij}\,\tilde{c}_{\theta\phi,ij}^{\mathtt{x}}\,\tilde{T}_{\theta\theta,ij}
-
\tilde{T}_{\theta\phi,ij}^{2}\,\tilde{c}_{\theta\theta,ij}^{\mathtt{x}}
}{
\tilde{T}_{\theta\theta,ij}^{2}
},
\end{align*}
with $\mathtt{x}\in\{\mathtt{s},\mathtt{r}\}$ and $\tilde{c}_{\psi\psi,ij}^{\mathtt{s}}
\triangleq
c_{\psi,ij}
+
\frac{\rho^{2}|\hat{c}_{\psi\tilde{\boldsymbol{\beta_\mathtt{s}}},ij}|^{2}
      c_{\tilde{\boldsymbol{\beta_\mathtt{s}}}\tilde{\boldsymbol{\beta_\mathtt{s}}},ij}}
     {T_{\tilde{\boldsymbol{\beta_\mathtt{s}}}\tilde{\boldsymbol{\beta_\mathtt{s}}},ij}^{2}}$, $
\tilde{c}_{\theta\phi,ij}^{\mathtt{s}}
\triangleq
c_{\theta\phi,ij}
+
\frac{\rho^{2}
      \hat{c}_{\theta\tilde{\boldsymbol{\beta_\mathtt{s}}},ij}
      \hat{c}_{\phi\tilde{\boldsymbol{\beta_\mathtt{s}}},ij}^{*}
      c_{\tilde{\boldsymbol{\beta_\mathtt{s}}}\tilde{\boldsymbol{\beta_\mathtt{s}}},ij}}
     {T_{\tilde{\boldsymbol{\beta_\mathtt{s}}}\tilde{\boldsymbol{\beta_\mathtt{s}}},ij}^{2}}$, 
\begin{align*}
\tilde{c}_{\psi\psi,ij}^{\mathtt{r}}
&\!\!\triangleq\!
\hat{c}_{\psi,ij}
-
\frac{
2\rho|\hat{c}_{\psi\tilde{\boldsymbol{\beta_\mathtt{s}}},ij}|^{2}
T_{\tilde{\boldsymbol{\beta_\mathtt{s}}}\tilde{\boldsymbol{\beta_\mathtt{s}}},ij}
-
\rho^{2}|\hat{c}_{\psi\tilde{\boldsymbol{\beta_\mathtt{s}}},ij}|^{2}
\hat{c}_{\tilde{\boldsymbol{\beta_\mathtt{s}}}\tilde{\boldsymbol{\beta_\mathtt{s}}},ij}
}{T_{\tilde{\boldsymbol{\beta_\mathtt{s}}}\tilde{\boldsymbol{\beta_\mathtt{s}}},ij}^{2}},
\\
\!\!\tilde{c}_{\theta\!\phi\!,ij}^{\mathtt{r}}
&\!\!\triangleq\!
\hat{c}_{\theta\!\phi,i\!j}
\!\!-\!
\frac{
\!2\rho
\hat{c}_{\theta\!\tilde{\boldsymbol{\beta_\mathtt{s}}},ij}
\hat{c}_{\phi\!\tilde{\boldsymbol{\beta_\mathtt{s}}},ij}^{*}
T_{\!\tilde{\boldsymbol{\beta_\mathtt{s}}}\tilde{\boldsymbol{\beta_\mathtt{s}}},ij}
\!-\!\!
\rho^{2}
\hat{c}_{\!\theta\!\tilde{\boldsymbol{\beta_\mathtt{s}}},ij}
\hat{c}_{\!\phi\!\tilde{\boldsymbol{\beta_\mathtt{s}}},ij}^{*}
\hat{c}_{\!\tilde{\boldsymbol{\beta_\mathtt{s}}}\tilde{\boldsymbol{\beta_\mathtt{s}}},ij}
}{T_{\tilde{\boldsymbol{\beta_\mathtt{s}}}\tilde{\boldsymbol{\beta_\mathtt{s}}},ij}^{2}}.
\end{align*}
We note that $\{\tilde{T}_{\theta\theta,ij},\tilde{T}_{\phi\phi,ij},\tilde{T}_{\theta\phi,ij},T_{\tilde{\boldsymbol{\beta}}_{\mathtt{s}}\tilde{\boldsymbol{\beta}}_{\mathtt{s}},ij}\}$ and
$\{c_{\psi,ij},\hat{c}_{\psi,ij},
c_{\theta\phi,ij},$ $\hat{c}_{\theta\phi,ij},
c_{\tilde{\boldsymbol{\beta_\mathtt{s}}}\tilde{\boldsymbol{\beta_\mathtt{s}}},ij} ,
\hat{c}_{\tilde{\boldsymbol{\beta_\mathtt{s}}}\tilde{\boldsymbol{\beta_\mathtt{s}}},ij},
\hat{c}_{\psi\tilde{\boldsymbol{\beta_\mathtt{s}}},ij}\}_{\psi\in\{\theta,\phi\}}$
are evaluated using \eqref{eq:app_Jpp}--\eqref{eq:app_jpa}, \eqref{eq_FIM} and
\eqref{chat1_thm}--\eqref{chat2_thm}, respectively.
\end{theorem}
\begin{algorithm}[t]
\caption{ AO-SCA Algorithm for Solving \eqref{prob:robust_power_allocation}}
\SetAlgoLined
\label{algo1}
\textbf{Initialization:} Set $t:=0$ and generate an initial feasible point
$\bigl(\boldsymbol{\vartheta}^{[0]},\boldsymbol{\gamma}^{[0]},\rho^{[0]}\bigr)$. \\
\Repeat{Convergence}{
    Solve \eqref{prob:P1_nu} to obtain
$\boldsymbol{\vartheta}^{\star}$ for given
    $\bigl(\boldsymbol{\gamma}^{[t]},\rho^{[t]}\bigr)$\;

    Update $\boldsymbol{\vartheta}^{[t+1]}:=\boldsymbol{\vartheta}^{\star}$\;

    Solve \eqref{prob:p2} to obtain
$\bigl(\boldsymbol{\gamma}^{\star},\rho^{\star}\bigr)$ for given
$\boldsymbol{\vartheta}^{[t+1]}$\;

    Update
    $\bigl(\boldsymbol{\gamma}^{[t+1]},\rho^{[t+1]}\bigr)
    :=
    \bigl(\boldsymbol{\gamma}^{\star},\rho^{\star}\bigr)$\;

    Set $t:=t+1$\;
}
\textbf{Output:}
$\bigl(\boldsymbol{\vartheta}^{\star},\boldsymbol{\gamma}^{\star},\rho^{\star}\bigr)$\;
\end{algorithm}

\subsection{Convergence and Computational Complexity}
\label{subsec:convergence_complexity}
Algorithm~\ref{algo1} summarizes the proposed AO-SCA framework for solving problem~\eqref{prob:robust_power_allocation}. The convex subproblem $\mathcal{P}_{1}$ in \eqref{prob:P1_nu} involves $K$ optimization variables and three convex constraints, while $\mathcal{P}_{2}$ in \eqref{prob:p2} involves $K+1$ optimization variables and three convex constraints. Therefore, the overall computational complexity of Algorithm~\ref{algo1} is
$\mathcal{O}\!\left(
I_{\mathtt{AO}}
\left(
K^{3}
+
(K+1)^{3}
\right)
\right)
$
\cite[Sec.~1.3.1]{BoydBook}, where $I_{\mathtt{AO}}$ denotes the number of
AO iterations.
We next discuss the convergence of Algorithm~\ref{algo1}. 
From
\eqref{eq:rate_lb_nu}, we have
$\widetilde{\mathcal{R}}_{k,\boldsymbol{\vartheta}}^{\mathtt{bf}}
(\boldsymbol{\vartheta};\boldsymbol{\vartheta}^{[t]})
\leq
\mathcal{R}_{k}^{\mathtt{bf}}
(\boldsymbol{\vartheta},\boldsymbol{\gamma}^{[t]},\rho^{[t]})$
and
$\widetilde{\mathcal{R}}_{k,\boldsymbol{\vartheta}}^{\mathtt{bf}}
(\boldsymbol{\vartheta}^{[t]};\boldsymbol{\vartheta}^{[t]})
=
\mathcal{R}_{k}^{\mathtt{bf}}
(\boldsymbol{\vartheta}^{[t]},\boldsymbol{\gamma}^{[t]},\rho^{[t]})$.
Since $\boldsymbol{\vartheta}^{[t+1]}$ is obtained by solving
\eqref{prob:P1_nu}, we have
$\mathcal{R}^{\mathtt{bf}}
(\boldsymbol{\vartheta}^{[t]},\boldsymbol{\gamma}^{[t]},\rho^{[t]})
\leq
\mathcal{R}^{\mathtt{bf}}
(\boldsymbol{\vartheta}^{[t+1]},\boldsymbol{\gamma}^{[t]},\rho^{[t]})$.
Similarly, from \eqref{eq:rate_lb_gamma_rho}, we have
$\widetilde{\mathcal{R}}_{k,\boldsymbol{\gamma},\rho}^{\mathtt{bf}}
(\boldsymbol{\gamma},\rho;\boldsymbol{\gamma}^{[t]},\rho^{[t]})
\leq
\mathcal{R}_{k}^{\mathtt{bf}}
(\boldsymbol{\vartheta}^{[t+1]},\boldsymbol{\gamma},\rho)$,
with equality at $(\boldsymbol{\gamma}^{[t]},\rho^{[t]})$. Since
$(\boldsymbol{\gamma}^{[t+1]},\rho^{[t+1]})$ is obtained by solving
\eqref{prob:p2}, it follows that
$\mathcal{R}^{\mathtt{bf}}
(\boldsymbol{\vartheta}^{[t+1]},\boldsymbol{\gamma}^{[t]},\rho^{[t]})
\leq
\mathcal{R}^{\mathtt{bf}}
(\boldsymbol{\vartheta}^{[t+1]},\boldsymbol{\gamma}^{[t+1]},\rho^{[t+1]})$.
Hence, Algorithm~\ref{algo1} produces a non-decreasing sequence of the objective values. Since the achievable sum rate is upper-bounded under a finite power budget $P_{\mathtt t}$, the generated objective sequence converges.

\section{Numerical Results and Discussions}
\label{section:numerical_results}
In this section, we validate the developed analytical framework and evaluate the performance of the proposed robust power allocation scheme. Unless otherwise specified, users are uniformly distributed within a circular cell of radius $1000$ m, with the BS located at the cell center $(0,0)$~\cite{nhan_isac_power_allocation}. The sensing target is located at
$(\theta,\phi)=\left(\frac{\pi}{8},\frac{\pi}{4}\right)$.
The large-scale fading coefficient of user $k$ is modeled as
$\beta_k=z_k/(r_k/r_{\mathtt{0}})^{\nu_{\mathtt{pl}}}$~\cite{Ngo2013}, where
$r_{\mathtt{0}}=100$~m, $z_k$ follows a log-normal distribution with
standard deviation $7$~dB \cite{liao}, and $r_k$ and $\nu_{\mathtt{pl}}$ denote the
distance between user $k$ and the BS, and the path-loss exponent,
respectively. We set $K=8$,
$N_{\mathtt{t}}=121$, $N_{\mathtt{r}}=25$, $\nu_{\mathtt{pl}}=3.2$,
$M=30$, $\tau_{\mathtt{p}}=K$ \cite{Ngo2013}, $\tau_{\mathtt{c}}=100$,
$\sigma^2=\sigma_{\mathtt{c}}^{2}=\sigma_{\mathtt{s}}^{2}=1$~\cite{Ngo2013}, and $\beta_\mathtt{s}=1$~\cite{li2008target}. The order of the Gaussian quadrature rule is set to
$G_{\psi}\in\{60,80\}$, while the sharpness parameter of the sigmoid
approximation and the convergence tolerance of Algorithm~\ref{algo1} are
set to $\varrho=1$~\cite{shi_sigmoid} and $10^{-3}$, respectively.
The considered benchmark power allocation schemes are
summarized as follows:
\begin{itemize}
\item \emph{Equal power}: We first equally allocate the total power among
$\bm{1}_K^{\T}\boldsymbol{\vartheta}$,
$N_{\mathtt{t}}\boldsymbol{\xi}_{\mathtt{bf}}^{\T}
(\boldsymbol{\vartheta})\boldsymbol{\gamma}$, and $N_{\mathtt{t}}\rho$
in~\eqref{prob:P0_power}. Then, we set
$\vartheta_1=\cdots=\vartheta_k=\cdots=\vartheta_K$ and
$\gamma_1=\cdots=\gamma_k=\cdots=\gamma_K$, to obtain
$(\vartheta_k,\gamma_k,\rho)=
\left(
{\frac{P_{\mathtt{t}}^{\max}}{3K}},
\frac{P_{\mathtt{t}}^{\max}}
{3N_{\mathtt{t}}\sum_{j=1}^{K}
[\boldsymbol{\xi}_{\mathtt{bf}}(\boldsymbol{\vartheta})]_j},
\frac{P_{\mathtt{t}}^{\max}}{3N_{\mathtt{t}}}
\right)$, $\forall k$.

    \item \emph{Equal C \& equal P}: 
    The power allocation
    is performed w.r.t. $\boldsymbol{\vartheta}$,
    $\boldsymbol{\gamma}$, and $\rho$, while the pilot and communications
    powers are equally allocated among the users. This is achieved by
    adding the constraints $\vartheta_1=\cdots=\vartheta_K$ and
    $\gamma_1=\cdots=\gamma_K$ into \eqref{prob:robust_power_allocation}.

    \item \emph{Non-robust}: The power allocation is performed based on
deterministic CRLB constraints evaluated at the estimated angles
$(\hat{\theta},\hat{\phi})$. This is achieved by using
$\mathtt{CRLB}_{\psi}^{\mathtt{bf}}\leq\mathtt{CRLB}_{\psi}^{0}$,
$\psi\in\{\theta,\phi\}$, instead of the outage-probability constraints
in~\eqref{prob:P0_crlb_theta} and~\eqref{prob:P0_crlb_phi}.
\end{itemize}

\begin{figure}[t]
\vspace{-0.07cm}
\small
\hspace{-2.5mm}
\subfigure[Sensing CRLBs]{
\label{fig:crlb}
\includegraphics[width=0.505\linewidth,height=3.5cm]{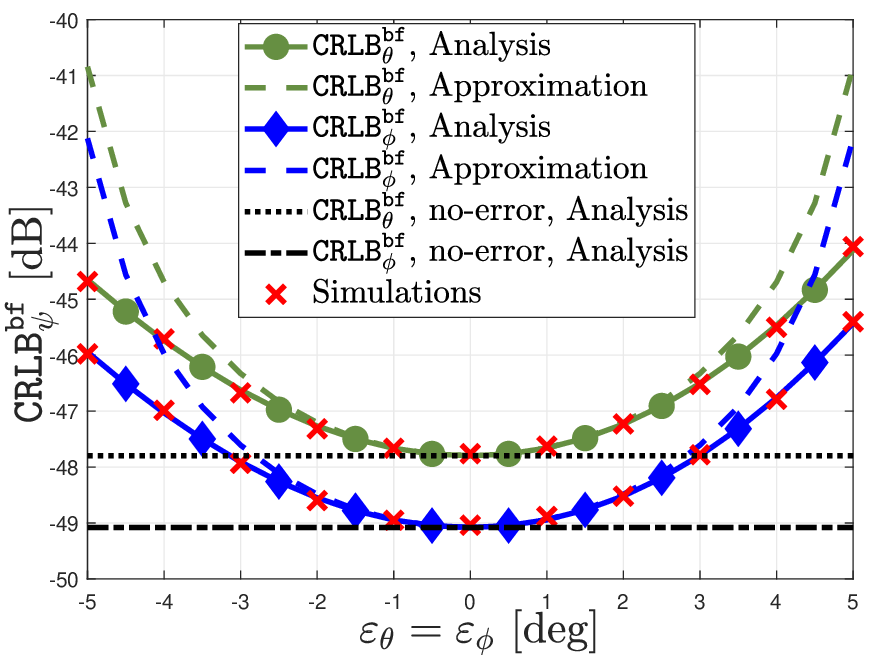}}
\hspace{-3.5mm}
\subfigure[CDFs of the CRLBs]{
\label{fig:cdf_crlb}
\includegraphics[width=0.51\linewidth,height=3.5cm]{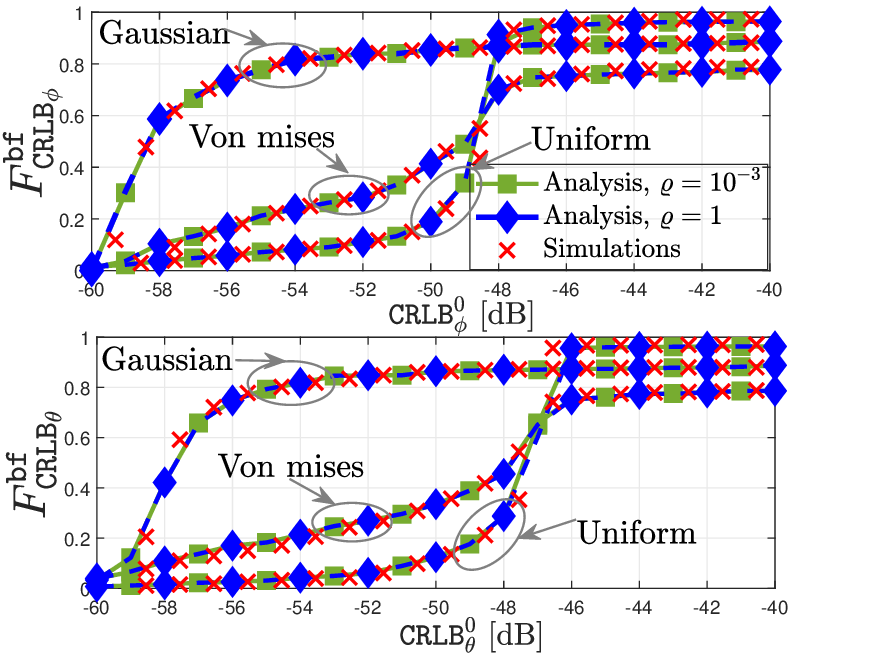}}
\caption{{Validation of (a) $\mathtt{CRLB}_\psi^\mathtt{bf}$ and (b) $F^\mathtt{bf}_{\mathtt{CRLB}_\psi}$, $\psi \in \{\theta,\phi\}$, under equal power allocation for} $\sigma_\psi=4^\circ$, $U_\psi=2$, $\kappa_\psi=10$, and different values of $\varrho$.}
	\label{cdf validation}
\end{figure}

\vspace{-0.1cm}
We first validate the analytical expressions derived for the CRLBs, their asymptotic approximations, and the corresponding CDFs in Theorem~\ref{thm:deterministic_crlb}, Corollary~\ref{rem:asymptotic_crlb}, and Theorem~\ref{thm:random_crlb}, respectively. Fig.~\ref{cdf validation} shows the analytical and Monte Carlo results for
the CRLBs under deterministic error values and for the CDFs under the  Gaussian, generalized uniform, and
von Mises  distributed errors.
 It is observed that the analytical results and the approximations closely align with the simulation ones over the entire considered range of $\varepsilon_{\psi}$ and
$\mathtt{CRLB}_{\psi}^{0}$, confirming the accuracy of the derived expressions. Moreover, Fig.~\ref{cdf validation}(b) shows that the
analysis remains tight for both considered values of the sharpness parameter
$\varrho$, justifying the adopted choice of
$\varrho$. 

\begin{figure}[t]
\centering
\includegraphics[scale=0.4]{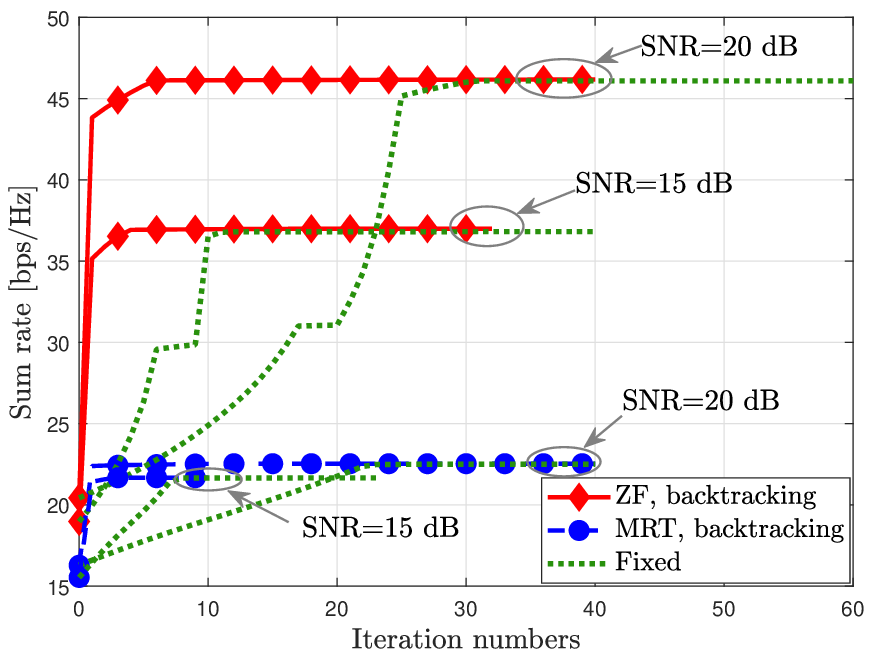}
\caption{Convergence of Algorithm~\ref{algo1} with $\mathsf{P}^{\mathtt{0}}_\theta=\mathsf{P}^{\mathtt{0}}_\phi=0.5$, $\mathtt{CRLB}_\theta^0=-48$ dB, $\mathtt{CRLB}_\phi^0\!=\!\!-40$ dB, $\sigma_\theta\!=\!8^\circ$, $\sigma_\phi\!=\!6^\circ$, $L_{\boldsymbol{\vartheta}}^{[t]}\!=\!20$,
$\ell_{\boldsymbol{\vartheta}}^{[t]}\!=\!20$, and
$\ell_{\boldsymbol{\gamma},\rho}^{[t]}\!=\!100$. }
	\label{convergence}
\end{figure}

Fig.~\ref{convergence} shows the convergence behavior of
Algorithm~\ref{algo1} for different SNR values with both MRT and ZF
precoders. The algorithm is initialized by the equal-power allocation with
$\bigl(\vartheta_k^{[0]},\gamma_k^{[0]},\rho^{[0]}\bigr)
=
\left(
{\frac{P_{\mathtt{t}}^{\max}}{3K}},
\frac{P_{\mathtt{t}}^{\max}}
{3N_{\mathtt{t}}\sum_{j=1}^{K}
[\boldsymbol{\xi}_{\mathtt{bf}}(\boldsymbol{\vartheta}^{[0]})]_j},
\frac{P_{\mathtt{t}}^{\max}}{3N_{\mathtt{t}}}
\right)$, which satisfies the power constraint
in~\eqref{prob:P0_power}.
Without loss of generality, we set
$L_{k,\boldsymbol{\vartheta}}^{\mathtt{bf},[t]}
\triangleq L_{\boldsymbol{\vartheta}}^{[t]}$, $\forall k$,
in~\eqref{eq:rate_lb_nu},
$\ell_{\theta,\boldsymbol{\vartheta}}^{\mathtt{bf},[t]}
=\ell_{\phi,\boldsymbol{\vartheta}}^{\mathtt{bf},[t]}
=\ell_{\mathtt{p},\boldsymbol{\vartheta}}^{\mathtt{bf},[t]}
\triangleq \ell_{\boldsymbol{\vartheta}}^{[t]}$
in~\eqref{eq:pout_ub_nu} and \eqref{eq:ptot_ub_nu},
and
$\ell_{\theta,\boldsymbol{\gamma},\rho}^{\mathtt{bf},[t]}
=\ell_{\phi,\boldsymbol{\gamma},\rho}^{\mathtt{bf},[t]}
\triangleq \ell_{\boldsymbol{\gamma},\rho}^{[t]}$
in~\eqref{eq:pout_ub_gamma_rho}. We employ backtracking search to obtain the parameters
$L_{\boldsymbol{\vartheta}}^{[t]}$,
$\ell_{\boldsymbol{\vartheta}}^{[t]}$, and
$\ell_{\boldsymbol{\gamma},\rho}^{[t]}$ ~\cite{guo2020weighted}.
For comparison, we also illustrate the convergence behavior obtained with fixed parameter values. As shown in Fig.~\ref{convergence}, Algorithm~\ref{algo1} converges within approximately $8$--$10$ AO iterations when backtracking search is employed, whereas $10$--$30$ iterations are required when using fixed parameter values.


\begin{figure*}[t!]
\small
    \centering
    \hspace{-2mm}
    \subfigure[Communications sum rate versus SNR]
{\includegraphics[scale=0.4]{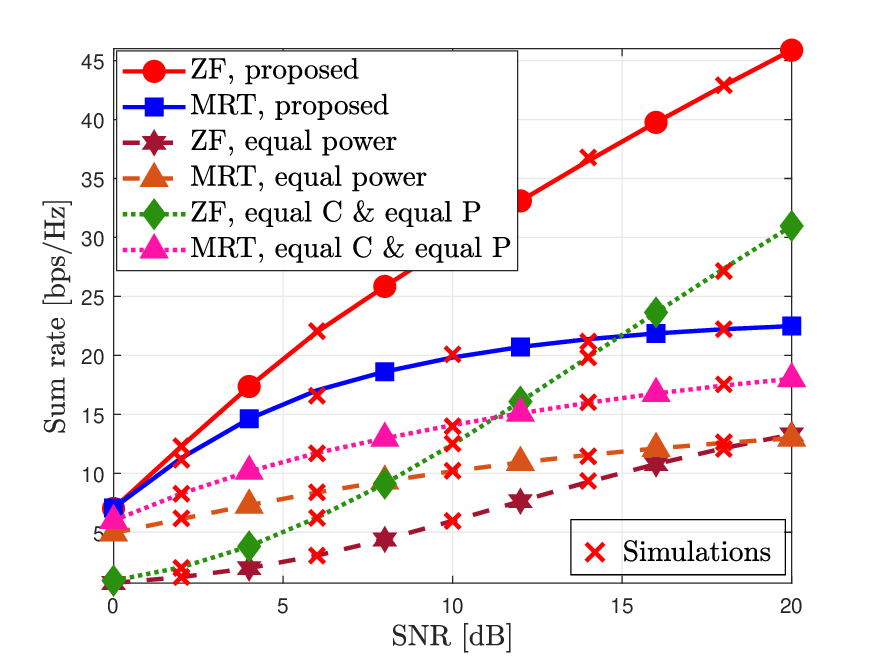}}
     \hspace{-7mm} 
    \subfigure[Sensing outage probability versus SNR]
    { \includegraphics[scale=0.4]{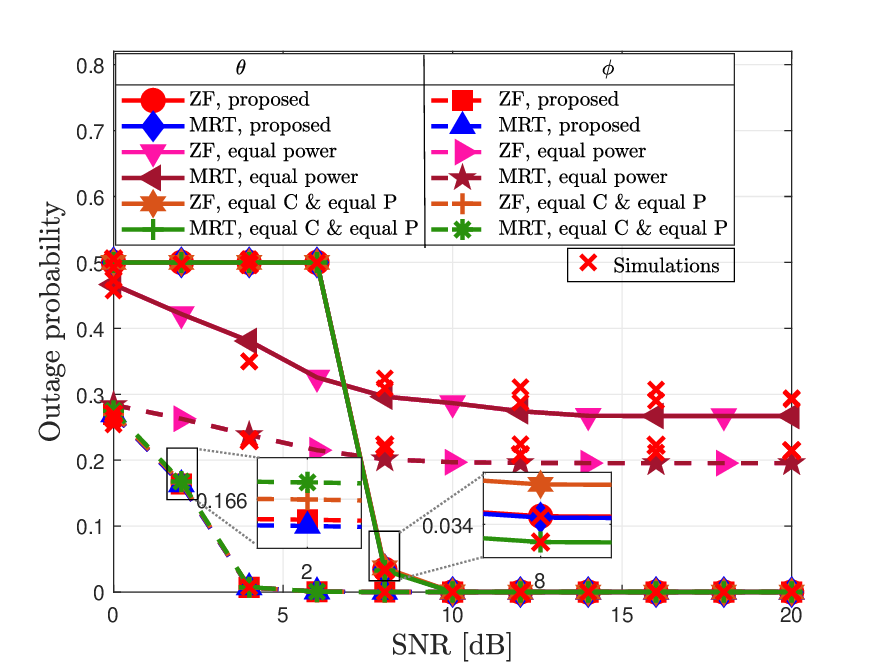}}
       \hspace{-7mm} 
    \subfigure[Power allocations versus SNR]
    { \includegraphics[scale=0.4]{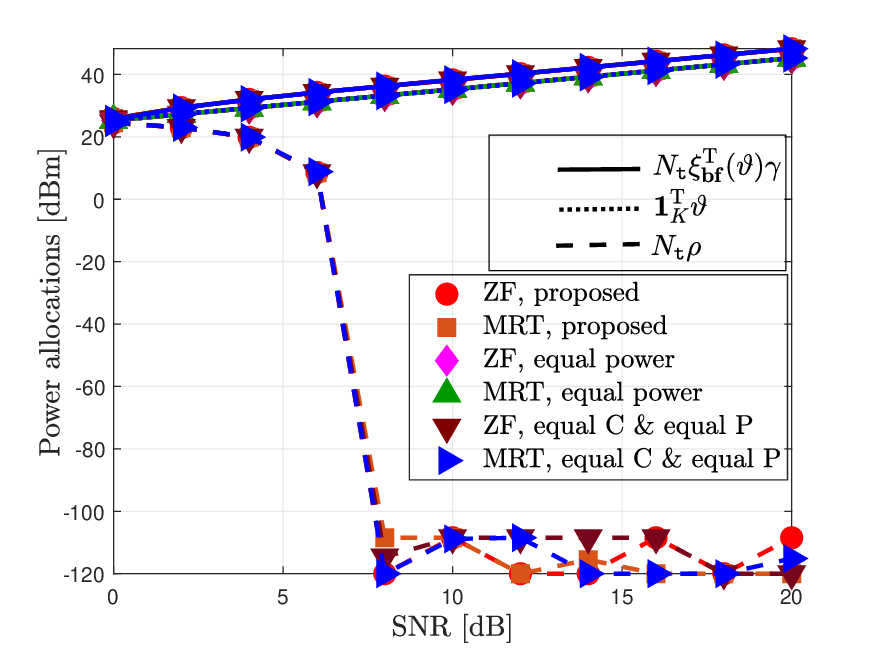}}
\caption{\small{Communications and sensing performances with ${\mathsf{P}^{\mathtt{0}}_\theta=\mathsf{P}^{\mathtt{0}}_\phi}=0.5$, $\mathtt{CRLB}_\theta^0=-48$ dB, $\mathtt{CRLB}_\phi^0=-40$ dB}, $\sigma_\theta=8^\circ$, and $\sigma_\phi=6^\circ$.
}
    \label{fig:snr_vs_metrics} 
\end{figure*}
\begin{figure*}[t!]
\small
    \centering
    \hspace{-4mm}
    \subfigure[Sum rate versus CRLB threshold]
{\includegraphics[scale=0.4]{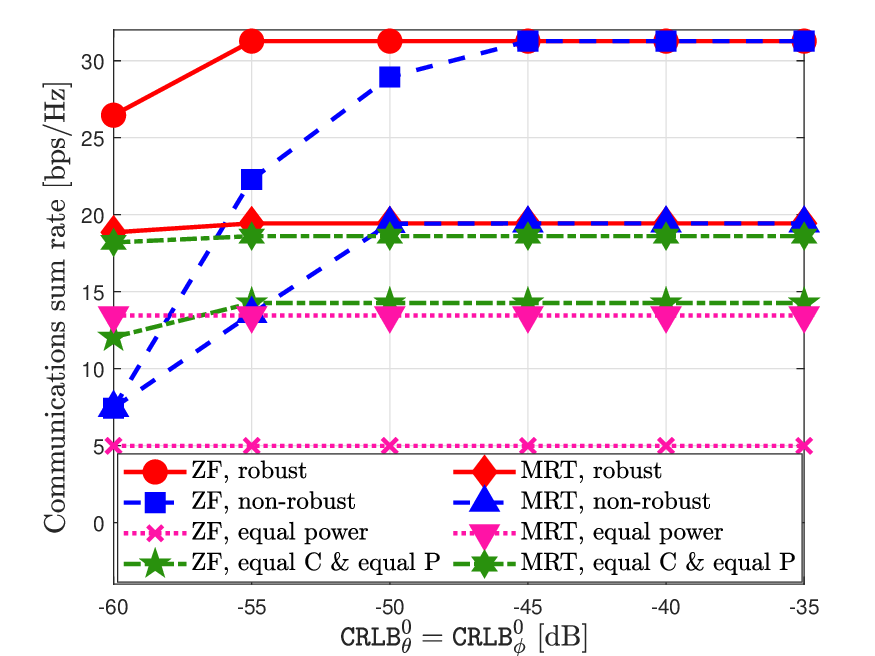}}
     \hspace{-5mm} 
    \subfigure[Sum rate versus outage threshold]
    { \includegraphics[scale=0.4]{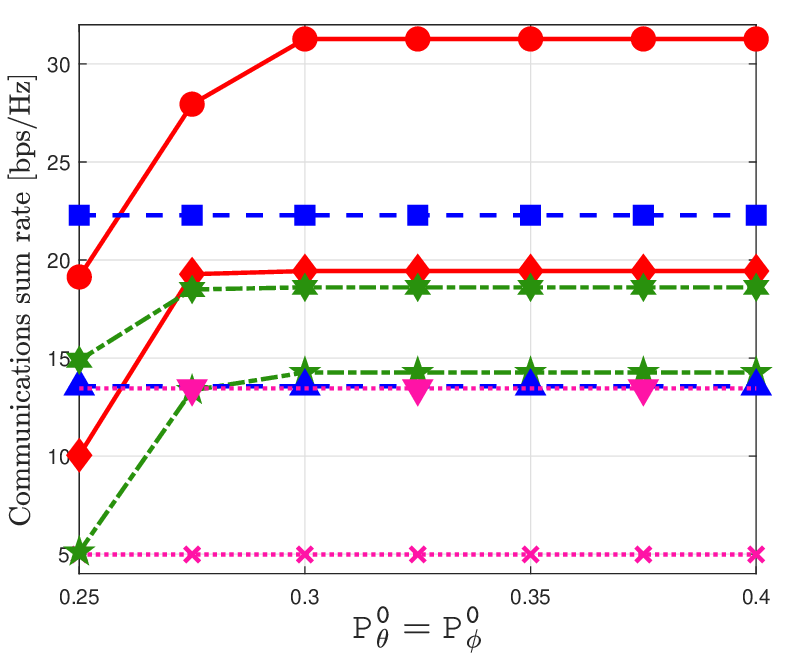}}
       \hspace{-5mm} 
    \subfigure[Sensing outage probability versus outage threshold]
    { \includegraphics[scale=0.4]{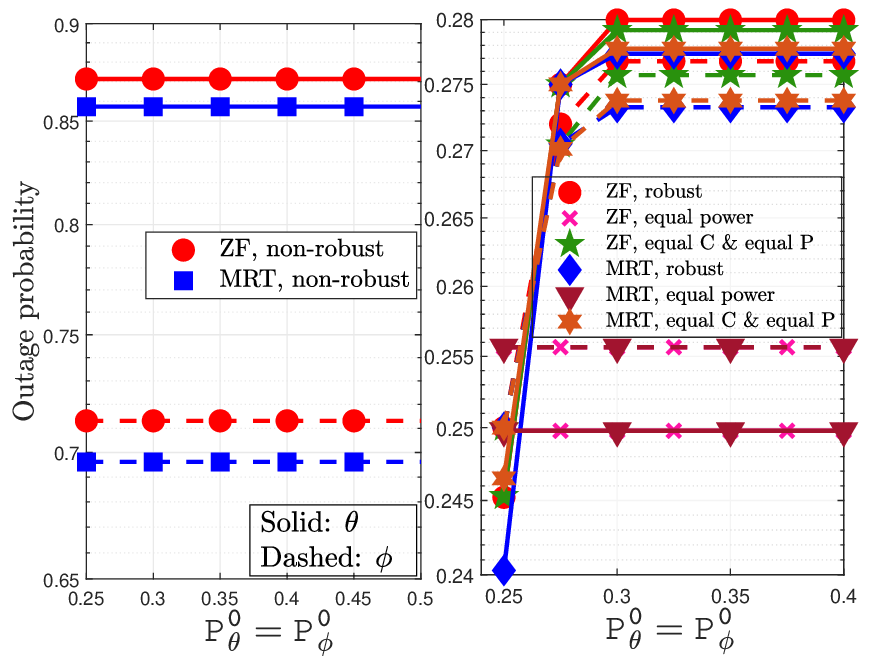}}
\caption{\small{Communications and sensing performance tradeoff, with 
SNR $=15$ dB and $\sigma_\psi=10^\circ$. 
}}
    \label{fig:tradeoff} 
\end{figure*}

Fig.~\ref{fig:snr_vs_metrics} illustrates the communications and sensing performances, together with the corresponding power allocations, achieved by the proposed and benchmark schemes. The main observations are summarized as follows.
\begin{itemize}
\item Fig.~\ref{fig:snr_vs_metrics}(a) shows that the proposed ZF-based scheme consistently achieves the highest sum rate among all ZF-based schemes across the considered SNR range, followed by the \emph{equal C \& equal P} and \emph{equal power} benchmarks. At $\mathrm{SNR}=14$ dB, the proposed ZF scheme improves the sum rate by approximately $295.5 \%$ and $85.8 \%$ compared with the ZF \emph{equal power} and ZF \emph{equal C \& equal P} schemes, respectively. Similar gains are observed for MRT-based schemes, demonstrating the effectiveness of the proposed power allocation design.

\item Fig.~\ref{fig:snr_vs_metrics}(b) shows the CRLB outage probabilities  obtained from the
derived closed-form expression and those from Monte
Carlo simulations. It is observed that the analytical results
closely align with the simulations for all considered
cases, which validates the analysis in \eqref{eq:outage_definition} and \eqref{eq:cdf_unified_1overcrlb}.
\item As observed from Fig.~\ref{fig:snr_vs_metrics}(b), the \emph{equal power} scheme provides a lower outage probability for the angle $\theta$ than the proposed scheme in the low-SNR regime. This is consistent with Fig.~\ref{fig:snr_vs_metrics}(c), where the
\emph{equal power} scheme allocates a larger power fraction to sensing,
which improves the outage probability when the SNR is low. However, as the SNR increases, the proposed scheme provides
lower outage probabilities for both $\theta$ and $\phi$ than the
\emph{equal power} scheme.
Moreover, for both MRT and ZF, the proposed scheme and the
\emph{equal C \& equal P} benchmark have similar outage probabilities
because they yield the same values of $\bm{1}_K^{\T}\boldsymbol{\vartheta}$,
$N_{\mathtt{t}}\boldsymbol{\xi}_{\mathtt{bf}}^{\T}
(\boldsymbol{\vartheta})\boldsymbol{\gamma}$, and $N_{\mathtt{t}}\rho$, as observed from Fig.~\ref{fig:snr_vs_metrics}(c) and
~\eqref{eq_closed_crlb_theta}, \eqref{eq_closed_crlb_phi}, and~\eqref{eq:cdf_unified_1overcrlb}.
\end{itemize}

Fig.~\ref{fig:tradeoff}(a) shows the communications sum rate versus sensing CRLB threshold for the robust and non-robust schemes. For all schemes, the sum rate increases as the CRLB threshold is relaxed from $-60$ dB to $-45$ dB and then gradually saturates. This is because larger CRLB thresholds relax the sensing requirements, allowing more power to be allocated to communications and pilot training. Moreover, the proposed robust ZF scheme consistently achieves a higher sum rate than the ZF non-robust scheme, with the gain being particularly pronounced in the stringent sensing regime.
As the CRLB threshold increases, the sum rate gap between the two schemes gradually decreases. This is because a less stringent sensing requirement has a minor impact on power allocation, so both schemes allocate power mainly to maximize the communications sum rate, causing their sum rates to become identical.
 
We also show the communications sum rate versus the outage threshold
$\mathsf{P}^{\mathtt{0}}_{\psi}$, for $\psi\in\{\theta,\phi\}$ in Fig.~\ref{fig:tradeoff}(b).
Similar to Fig.~\ref{fig:tradeoff}(a), the sum rate of the proposed robust
scheme initially increases with $\mathsf{P}^{\mathtt{0}}_{\psi}$
and then saturates. In contrast, the sum rate of the non-robust scheme remains constant with
respect to $\mathsf{P}^{\mathtt{0}}_{\psi}$, since its design does not account for
target-angle uncertainty and the associated outage probability constraints. It is also observed that, for comparatively stricter outage
threshold, for example at $\mathsf{P}^{\mathtt{0}}_\psi=0.25$, the robust scheme achieves
a slightly lower sum rate compared with the non-robust scheme. This is because the robust design needs to allocate more power to sensing to ensure sensing reliability under a strict outage constraint, reducing the power available for communications and pilot training. Although the non-robust design achieves a slightly higher sum rate at $\mathsf{P}^{\mathtt{0}}_\psi=0.25$, it does so at the cost of a much higher sensing outage probability. This leads to an overall poorer communications--sensing tradeoff for the non-robust scheme.  
This observation is confirmed in Fig.~\ref{fig:tradeoff}(c), where the robust schemes achieve lower outage probabilities than the non-robust schemes over the entire range of $\mathsf{P}_{\psi}^{0}$. Similar trends are observed for the MRT-based schemes.

\begin{figure}[t]
\small
\centering
\subfigure[Gaussian distribution]{
\label{fig:gaussian_sumrate}
\includegraphics[width=0.505\linewidth,height=3.5cm]{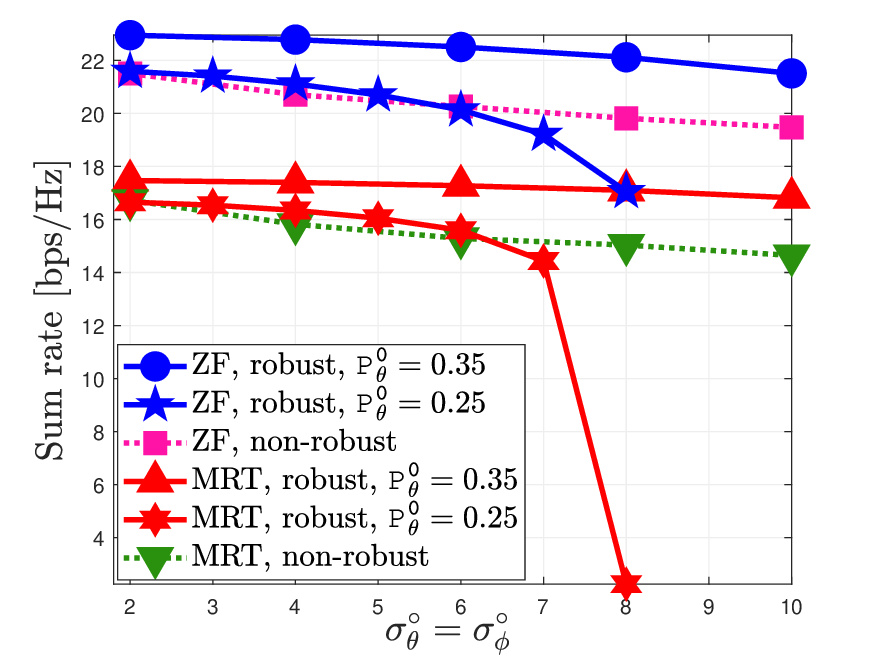}}
\hspace{-7mm}
\subfigure[Gaussian distribution]{
\label{fig:gaussian_outage}
\includegraphics[width=0.51\linewidth,height=3.5cm]{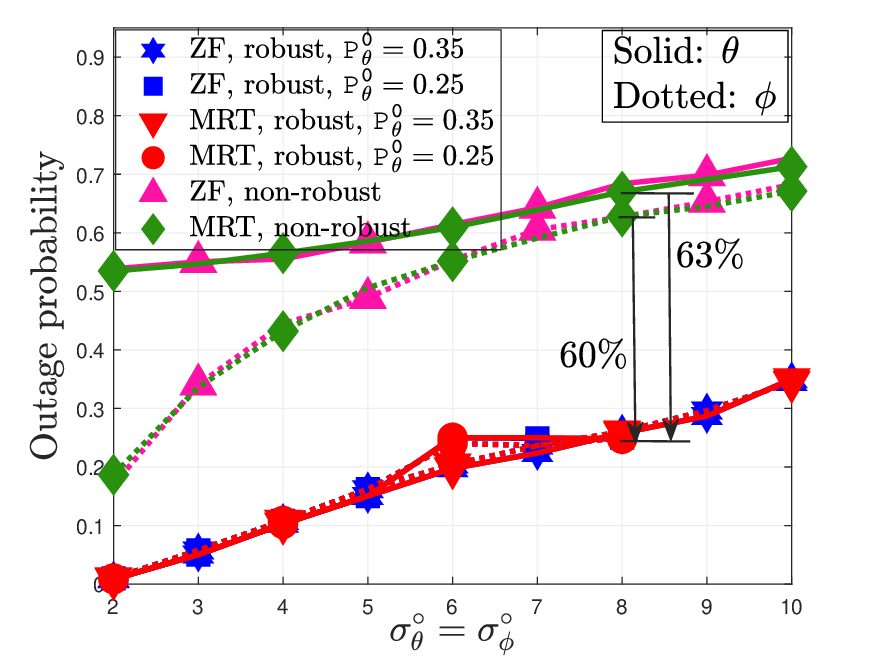}}
\\ \vspace{-0.35cm}
\subfigure[Uniform distribution]{
\label{fig:uniform_sumrate}
\includegraphics[width=0.51\linewidth,height=3.5cm]{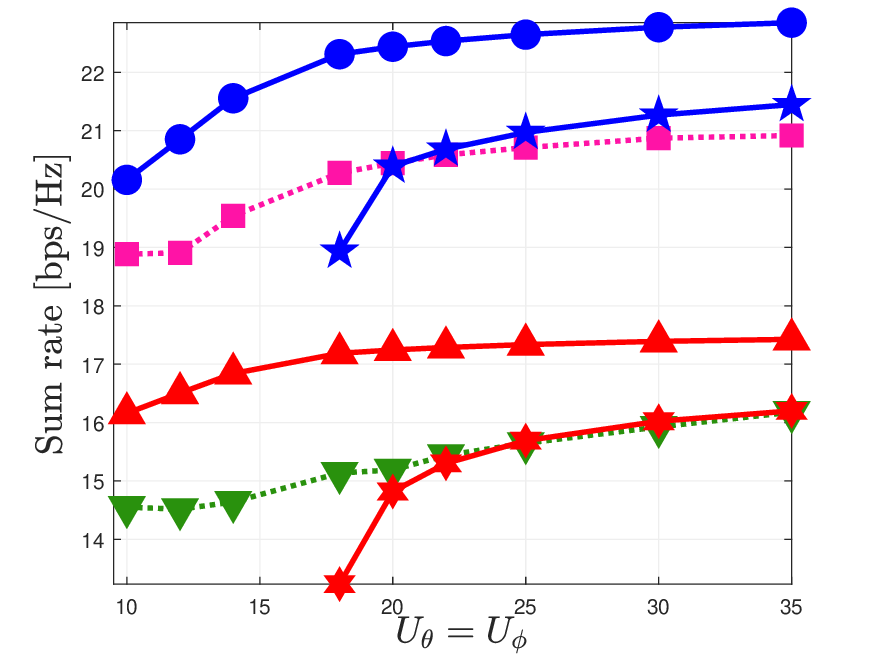}}
\hspace{-7mm}
\subfigure[Uniform distribution]{
\label{fig:uniform_outage}
\includegraphics[width=0.51\linewidth,height=3.5cm]{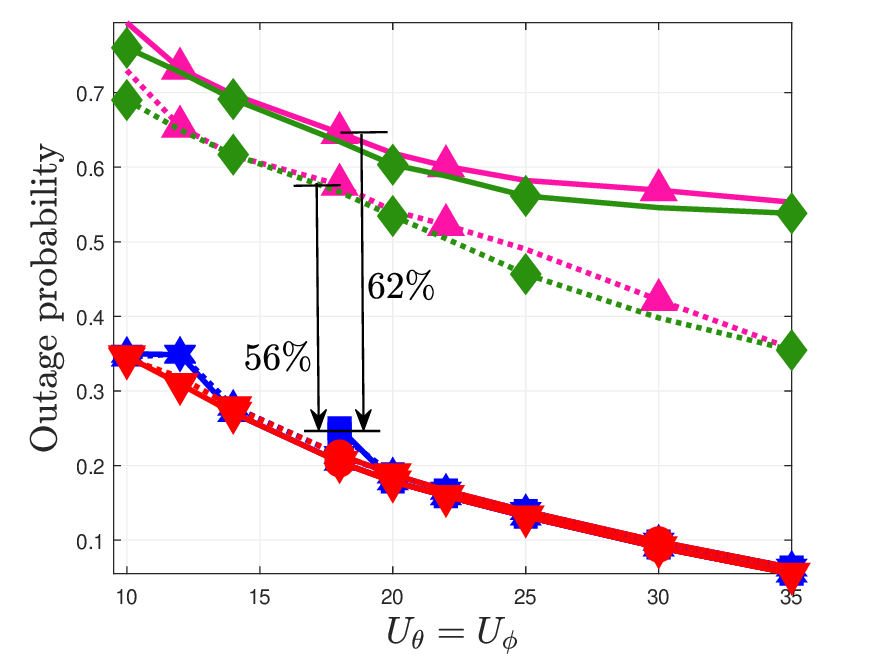}}
\\ \vspace{-0.35cm}
\subfigure[von Mises distribution]{
\label{fig:von_sumrate}
\includegraphics[width=0.51\linewidth,height=3.5cm]{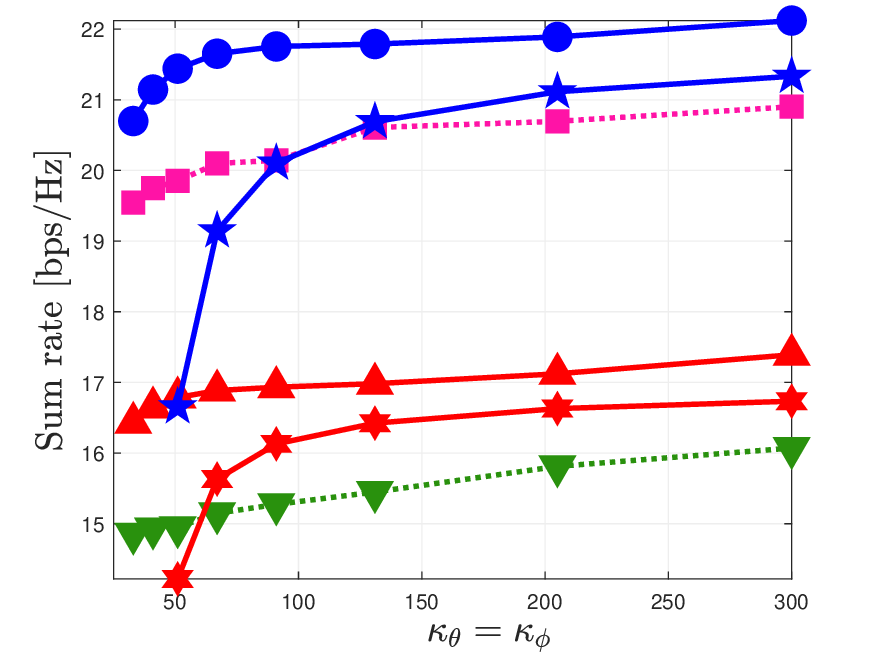}}
\hspace{-6.5mm}
\subfigure[von Mises distribution]{
\label{fig:von_outage}
\includegraphics[width=0.51\linewidth,height=3.5cm]{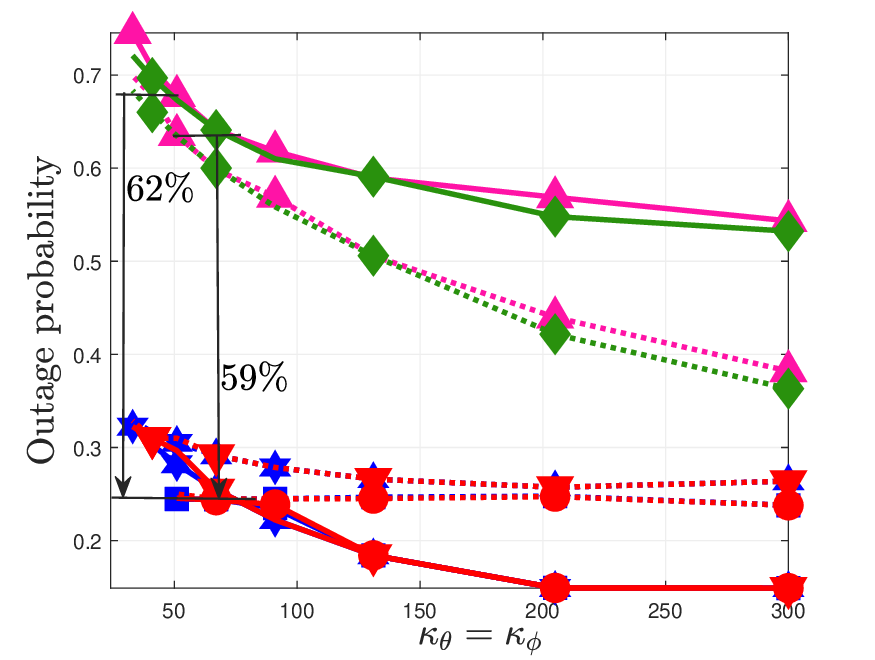}}
\caption{Impact of angle errors with different distributions on the sum rate and sensing outage probability, with $\mathtt{CRLB}_\theta^0=\mathtt{CRLB}_\phi^0=-55$ dB and $\mathsf{P}^\mathtt{0}_\theta=\mathsf{P}^\mathtt{0}_\phi$. 
}
\label{fig:error_impact}
\end{figure}

Fig.~\ref{fig:error_impact} shows the sum rate and CRLB outage
probabilities under different target-angle error models. Recall from Section~\ref{section:angleerror_outage} that, for $\psi\in\{\theta,\phi\}$, a larger
$\sigma_\psi$ corresponds to larger angle errors, whereas larger
$U_\psi$ and $\kappa_\psi$ correspond to smaller angle errors. We see from Fig.~\ref{fig:error_impact}(b) that the outage
probability increases with $\sigma_\psi$, while the corresponding sum rate
in Fig.~\ref{fig:error_impact}(a) decreases. This is because, when the
outage probability increases, the robust design needs to allocate more
power to sensing to satisfy the outage constraints. Thus, less power
is available for communications and pilot training, which reduces the sum
rate. In Figs.~\ref{fig:error_impact}(d) and~\ref{fig:error_impact}(f),
the outage probability decreases with $U_\psi$ and $\kappa_\psi$,
respectively, while the corresponding sum rates in
Figs.~\ref{fig:error_impact}(c) and~\ref{fig:error_impact}(e) increase,
following the same communications--sensing power tradeoff. These observations align with Remark \ref{rem:tradeoff} and Corollary \ref{rem:impact}.

Furthermore, for a comparatively strict sensing constraint, \textit{i.e.},
$\mathsf{P}^{\mathtt{0}}_\theta=\mathsf{P}^{\mathtt{0}}_\phi=0.25$, and severe angle errors, such as
$(\sigma_\psi,U_\psi,\kappa_\psi)=(8^\circ,18,51)$, the non-robust design
achieves a slightly higher sum rate than the robust design in Figs.~\ref{fig:error_impact}(a),~\ref{fig:error_impact}(c), and~\ref{fig:error_impact}(e). However, this comes at the cost of
sensing reliability, as also observed from Fig.~\ref{fig:tradeoff}(b). At
$(\sigma_\psi,U_\psi,\kappa_\psi)=(8^\circ,18,51)$, the robust design reduces the CRLB outage
probabilities for the azimuth and elevation angles by $(63\%,60\%)$,
$(62\%,56\%)$, and $(62\%,59\%)$ compared with the non-robust scheme under
Gaussian, uniform, and von Mises angle-error distributions, respectively. This demonstrates that the proposed robust design achieves a better rate-reliability tradeoff for the considered angle-error distributions.


\begin{figure}[t]
\small
\hspace{-2.5mm}
\subfigure[Sum rate versus $N_\mathtt{t}$]{
\label{fig:convergence_outer1}
\includegraphics[scale=0.3]{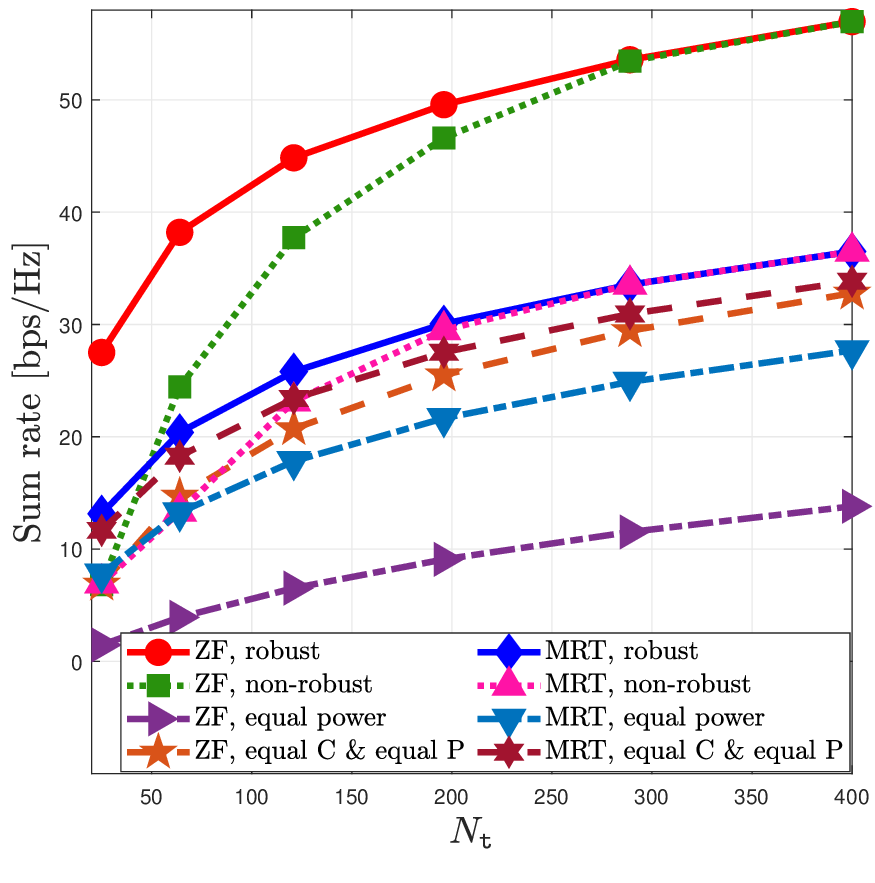}}
\hspace{-2.8mm}
\subfigure[Outage probability versus $N_\mathtt{t}$]{
\label{fig:convergence_inner1}
\includegraphics[scale=0.3]{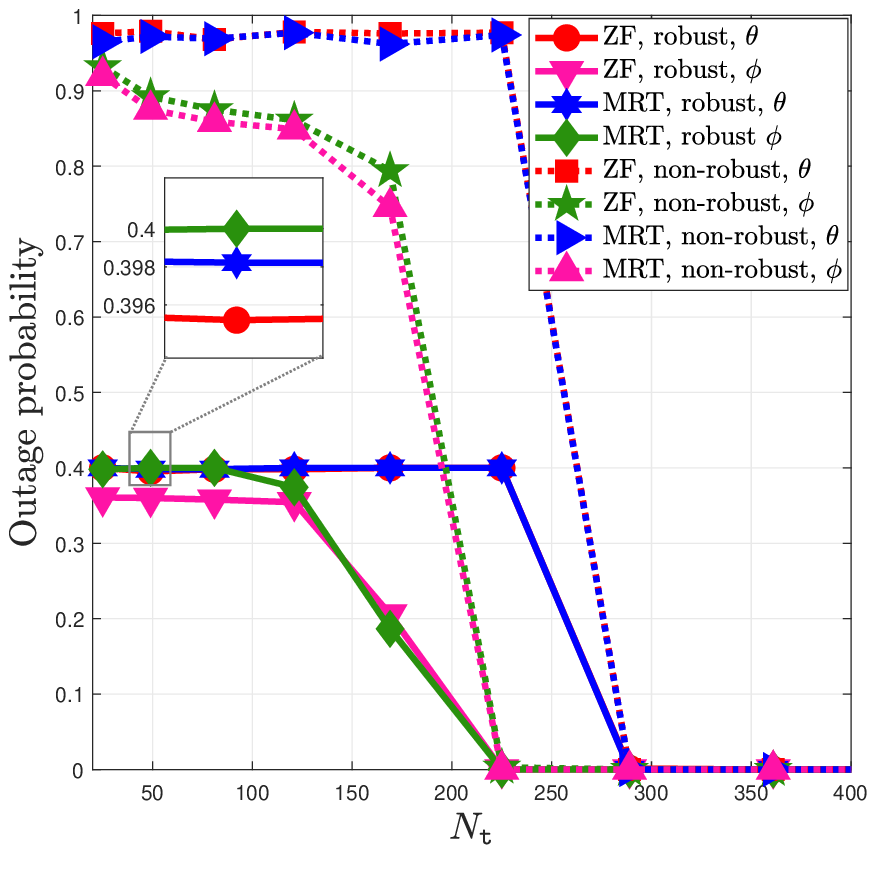}}
\caption{Communications and sensing performances, with $\mathsf{P}^{\mathtt{0}}_\psi=0.4$, for $\psi\in\{\theta,\phi\}$, $\sigma_\psi=10^\circ$, and SNR $=10$ dB.}
\label{fig:Nt}
\end{figure}

Fig.~\ref{fig:Nt} shows the impact of the number of transmit antennas on the communications and sensing performances. As shown in Fig.~\ref{fig:Nt}(a), increasing $N_{\mathtt t}$ improves the communications sum rate for all considered schemes due to the enhanced beamforming gain provided by larger antenna arrays. Moreover, the sum-rate gap between the robust and non-robust schemes gradually narrows as $N_{\mathtt t}$ increases. This decrease is consistent with Fig.~\ref{fig:Nt}(b), where the
sensing outage probabilities of both the robust and non-robust schemes
become nearly zero at high $N_{\mathtt{t}}$. Hence, the sensing outage
constraints have only a minor impact on the power allocation, and both
schemes tend to allocate power primarily to improve the communications sum
rate. Finally, the reduction in sensing outage probability with increasing $N_{\mathtt t}$ corroborates the asymptotic result in Remark~\ref{rem:crlb antennas}. 

\section{Conclusion}
\label{sec:conclusion}
This paper investigated monostatic mMIMO-ISAC systems under target-angle uncertainty. We derived closed-form expressions for the sensing CRLBs, outage probabilities, and the communications sum rate under Gaussian, generalized-uniform, and von Mises angle-error models. Building upon these results, a robust power allocation framework was developed to maximize the communications sum rate while satisfying sensing outage and total-power constraints. An efficient AO-SCA algorithm was proposed to solve the resulting nonconvex optimization problem. The analytical and simulation results demonstrated the accuracy of the developed framework and revealed the impact of target-angle errors on sensing reliability and the communications--sensing tradeoff. In particular, the proposed robust design significantly reduced sensing outage probabilities compared with a non-robust benchmark while maintaining competitive communications performance, thereby achieving a superior communications-sensing tradeoff.

\appendices

\section{Proof of Theorem~\ref{thm:deterministic_crlb}}
\label{app:crlb_error}
Applying steps similar to those in~\cite[Appendix C]{nhan_isac_power_allocation} to \eqref{eq_approx_cov}, 
we first obtain
$\mathbf{R}_{\mathtt{bf}}
 =\boldsymbol{\xi}_{\mathtt{bf}}^{\T}\boldsymbol{\gamma}\,\mathbf{I}_{N_{\mathtt{t}}}
 +\rho\mathbf{u}\mathbf{u}^{\H}$.
Then, using the fact that
$\mathbf{a}^{\H}\dot{\mathbf{a}}_{\theta}
 =\mathbf{a}^{\H}\dot{\mathbf{a}}_{\phi}
 =\mathbf{b}^{\H}\dot{\mathbf{b}}_{\theta}
 =\mathbf{b}^{\H}\dot{\mathbf{b}}_{\phi}=0$ into \eqref{eq_def_Jtt}--\eqref{eq_def_Jaa}, 
 for any sensing beamformer
 $\mathbf{u}$, {we have}
\begin{align}
\!\!\!T_{\theta\theta}
&\triangleq
\boldsymbol{\xi}_{\mathtt{bf}}^{\T}(\boldsymbol{\vartheta})
\boldsymbol{\gamma}c_{\theta}
+\rho\hat{c}_{\theta},
\hspace{0.2cm}T_{\phi\phi}
\triangleq
\boldsymbol{\xi}_{\mathtt{bf}}^{\T}(\boldsymbol{\vartheta})
\boldsymbol{\gamma}c_{\phi}
+\rho\hat{c}_{\phi},
\label{eq:app_Jpp}
\\
\!\!\!T_{\theta\phi}
&\triangleq
\boldsymbol{\xi}_{\mathtt{bf}}^{\T}(\boldsymbol{\vartheta})
\boldsymbol{\gamma}c_{\theta\phi}
\!+\!\rho\hat{c}_{\theta\phi}, \mathbf{t}_{\psi\tilde{\boldsymbol{\beta_\mathtt{s}}}}
\!\triangleq\!
\rho\,\mathtt{Re}
\!\left(\!
\hat{c}_{\psi\tilde{\boldsymbol{\beta_\mathtt{s}}}}\!(1,\jmath)\!
\right)
,
\label{eq:app_Jaa}
\\
\mathbf{T}_{\tilde{\boldsymbol{\beta_\mathtt{s}}}\tilde{\boldsymbol{\beta_\mathtt{s}}}}
&\triangleq
\Big(
\boldsymbol{\xi}_{\mathtt{bf}}^{\T}(\boldsymbol{\vartheta})
\boldsymbol{\gamma}
c_{\tilde{\boldsymbol{\beta_\mathtt{s}}}\tilde{\boldsymbol{\beta_\mathtt{s}}}}
+\rho
\hat{c}_{\tilde{\boldsymbol{\beta_\mathtt{s}}}\tilde{\boldsymbol{\beta_\mathtt{s}}}}
\Big)\mathbf{I}_{2}=T_{\tilde{\boldsymbol{\beta}}_{\mathtt{s}}
\tilde{\boldsymbol{\beta}}_{\mathtt{s}}}\mathbf{I}_{2}.
\label{eq:app_jpa}
\end{align}
Here, $\psi\in\{\theta,\phi\}$,
$\{c_{\theta},c_{\phi},c_{\theta\phi},
c_{\tilde{\boldsymbol{\beta_\mathtt{s}}}\tilde{\boldsymbol{\beta_\mathtt{s}}}},
\hat{c}_{\theta},\hat{c}_{\phi},\hat{c}_{\theta\phi},
\hat{c}_{\tilde{\boldsymbol{\beta_\mathtt{s}}}\tilde{\boldsymbol{\beta_\mathtt{s}}}},
\hat{c}_{\psi\tilde{\boldsymbol{\beta_\mathtt{s}}}}\}$ are defined in
\eqref{chat1_thm}--\eqref{chat2_thm}, 
where
$g_{0}\triangleq\mathbf{u}^{\H}\mathbf{a}$,
$g_{\theta}\triangleq\mathbf{u}^{\H}\dot{\mathbf{a}}_{\theta}$,
and $g_{\phi}\triangleq\mathbf{u}^{\H}\dot{\mathbf{a}}_{\phi}$.
From \eqref{eq:app_Jpp}--\eqref{eq:app_jpa} and \eqref{eq_FIM} we define
\begin{align}
\mathtt{CRLB}_{\theta}
&\triangleq\!\Bigl(\widetilde{T}_{\theta\theta}
        \!-\!\tfrac{\widetilde{T}_{\theta\phi}^{2}}{\widetilde{T}_{\phi\phi}}\Bigr)^{-1}\!\!,
\mathtt{CRLB}_{\phi}
\triangleq\!\Bigl(\widetilde{T}_{\phi\phi}
     \!  -\!\tfrac{\widetilde{T}_{\theta\phi}^{2}}{\widetilde{T}_{\theta\theta}}\Bigr)^{-1}\!\!.
\label{eq:app_CRLBs}
\end{align}
For $\mathbf{u}=\mathbf{a}(\hat{\theta},\hat{\phi})$, we define
$\Delta_{\mathtt{y}}
 \triangleq\sin(\theta)\sin(\phi)-\sin(\hat{\theta})\sin(\hat{\phi})$
and
$\Delta_{\mathtt{z}}\triangleq\cos(\phi)-\cos(\hat{\phi})$, for
$\mathtt{x}\in\{\mathtt{y},\mathtt{z}\}$, and
\begin{align}
\!\!\!\!D_{N_{\mathtt{t}\mathtt{x}}}(\Delta_{\mathtt{x}})
\!\triangleq\!\!\!\!\!\!\!\!
\!\!\sum_{m=-(N_{\mathtt{t}\mathtt{x}}-1)/2}^{(N_{\mathtt{t}\mathtt{x}}-1)/2}
\!\!\!\!\!\!\!\!\!\!\!e^{j\pi m\Delta_{\mathtt{x}}},
Q_{N_{\mathtt{t}\mathtt{x}}}(\Delta_{\mathtt{x}})
&\triangleq
\!\!\!\!\!\!\!\!\!\!\!\sum_{m=-(N_{\mathtt{t}\mathtt{x}}-1)/2}^{(N_{\mathtt{t}\mathtt{x}}-1)/2}
\!\!\!\!\!\!\!\!\!\!m\,e^{j\pi m\Delta_{\mathtt{x}}},
\label{eq:app_QN}
\end{align}
whose closed forms are given in~\eqref{Sn} and~\eqref{Qn}.
For brevity, let
$D_{\mathtt{x}}\triangleq D_{N_{\mathtt{t}\mathtt{x}}}(\Delta_{\mathtt{x}})$
and
$Q_{\mathtt{x}}\triangleq Q_{N_{\mathtt{t}\mathtt{x}}}(\Delta_{\mathtt{x}})$
for $\mathtt{x}\in\{\mathtt{y},\mathtt{z}\}$.
We next compute $\{g_{0},g_{\theta},g_{\phi}\}$
required in \eqref{chattheta}--\eqref{chat2_thm} as follows:
\subsubsection{\textbf{Compute} $g_{0}$}
First, we use the identity
$(\mathbf{p}\otimes\mathbf{q})^{\H}(\mathbf{r}\otimes\mathbf{s})
=(\mathbf{p}^{\H}\mathbf{r})(\mathbf{q}^{\H}\mathbf{s})$ and the results in \eqref{eq_ath} and \eqref{eq_atv} to write
\begin{align}
g_{0}
=\mathbf{u}^{\H}\mathbf{a}
=\bigl(\mathbf{a}_{\mathtt{y}}^{\H}(\hat{\theta},\hat{\phi})\,
        \mathbf{a}_{\mathtt{y}}(\theta,\phi)\bigr)
 \bigl(\mathbf{a}_{\mathtt{z}}^{\H}(\hat{\phi})\,
        \mathbf{a}_{\mathtt{z}}(\phi)\bigr).    
\end{align}
Here, based on \eqref{eq:app_QN}, we can write
$\mathbf{a}_{\mathtt{y}}^{\H}(\hat{\theta},\hat{\phi})\,\mathbf{a}_{\mathtt{y}}(\theta,\phi)=D_{\mathtt{y}}$
and
$\mathbf{a}_{\mathtt{z}}^{\H}(\hat{\phi})\,\mathbf{a}_{\mathtt{z}}(\phi)=D_{\mathtt{z}}$, which yields
$g_{0} = D_{\mathtt{y}}D_{\mathtt{z}}$ in \eqref{g0}.
\subsubsection{\textbf{Compute} $g_{\theta}$}
We next use~\cite[(88),(90)]{nhan_isac_power_allocation} to write
$\dot{\mathbf{a}}_{\theta}=\dot{\mathbf{a}}_{\mathtt{y}\theta}\otimes\mathbf{a}_{\mathtt{z}}$. Here,
$\dot{\mathbf{a}}_{\mathtt{y}\theta}
 =j\pi\cos(\theta)\sin(\phi)\,\mathbf{v}_{\mathtt{ty}}\odot\mathbf{a}_{\mathtt{y}}$ with $\mathbf{v}_{\mathtt{tx}}=[{-(N_{\mathtt{tx}}-1)}/{2},\ldots,{(N_{\mathtt{tx}}-1)}/{2}]^{\T}$, for $\mathtt{x}\in\{\mathtt{y},\mathtt{z}\}$. Based on this and \eqref{eq:app_QN}, we have
\begin{align}
   \mathbf{u}^{\H}\dot{\mathbf{a}}_{\theta}
&= \bigl(\mathbf{a}_{\mathtt{y}}^{\H}(\hat{\theta},\hat{\phi})\,
          \dot{\mathbf{a}}_{\mathtt{y}\theta}\bigr)
   \bigl(\mathbf{a}_{\mathtt{z}}^{\H}(\hat{\phi})\,
          \mathbf{a}_{\mathtt{z}}\bigr)\\ &= j\pi\cos(\theta)\sin(\phi)\,
   \mathbf{a}_{\mathtt{y}}^{\H}(\hat{\theta},\hat{\phi})
   \bigl(\mathbf{v}_{\mathtt{ty}}\odot\mathbf{a}_{\mathtt{y}}\bigr)\,D_{\mathtt{z}}. 
\end{align}
Here,
$\mathbf{a}_{\mathtt{y}}^{\H}(\hat{\theta},\hat{\phi})
 (\mathbf{v}_{\mathtt{ty}}\odot\mathbf{a}_{\mathtt{y}})
 =Q_{N_{\mathtt{ty}}}(\Delta_{\mathtt{y}})=Q_{\mathtt{y}}$,
which yields $g_{\theta} = j\pi\cos(\theta)\sin(\phi)\,Q_{\mathtt{y}}D_{\mathtt{z}}$ in \eqref{gtheta}.
\subsubsection{\textbf{Compute} $g_{\phi}$}
By using~\cite[(89)]{nhan_isac_power_allocation}, we write
$\dot{\mathbf{a}}_{\phi}
 =\dot{\mathbf{a}}_{\mathtt{y}\phi}\otimes\mathbf{a}_{\mathtt{z}}
 +\mathbf{a}_{\mathtt{y}}\otimes\dot{\mathbf{a}}_{\mathtt{z}\phi}$.
Here,
$\dot{\mathbf{a}}_{\mathtt{y}\phi}
 =j\pi\sin(\theta)\cos(\phi)\,(\mathbf{v}_{\mathtt{ty}}\odot\mathbf{a}_{\mathtt{y}})$
and
$\dot{\mathbf{a}}_{\mathtt{z}\phi}
 =-j\pi\sin(\phi)\,(\mathbf{v}_{\mathtt{tz}}{\odot}\mathbf{a}_{\mathtt{z}})$. Using this and \eqref{eq:app_QN}, we have
\begin{align}
  \nonumber \mathbf{u}^{\H}\dot{\mathbf{a}}_{\phi}
&= \bigl(\mathbf{a}_{\mathtt{y}}^{\H}(\hat{\theta},\hat{\phi})\,
          \dot{\mathbf{a}}_{\mathtt{y}\phi}\bigr)
   \bigl(\mathbf{a}_{\mathtt{z}}^{\H}(\hat{\phi})\,\mathbf{a}_{\mathtt{z}}\bigr)
  +\bigl(\mathbf{a}_{\mathtt{y}}^{\H}(\hat{\theta},\hat{\phi})\,\mathbf{a}_{\mathtt{y}}\bigr)
   \bigl(\mathbf{a}_{\mathtt{z}}^{\H}(\hat{\phi})\,\dot{\mathbf{a}}_{\mathtt{z}\phi}\bigr)\\ \nonumber &= j\pi\sin(\theta)\cos(\phi)\,
   \mathbf{a}_{\mathtt{y}}^{\H}(\hat{\theta},\hat{\phi})
   (\mathbf{v}_{\mathtt{ty}}\odot\mathbf{a}_{\mathtt{y}})\,D_{\mathtt{z}} \\ 
  &\hspace{2cm}-j\pi\sin(\phi)\,D_{\mathtt{y}}\,
   \mathbf{a}_{\mathtt{z}}^{\H}(\hat{\phi})
   (\mathbf{v}_{\mathtt{tz}}\odot\mathbf{a}_{\mathtt{z}}).  
\end{align}
By applying \eqref{eq:app_QN} to $\mathbf{u}^{\H}\dot{\mathbf{a}}_{\phi}$, we get $g_{\phi}$
in \eqref{gphi}.
Finally, substituting $\{g_0,g_\theta,g_\phi\}$ 
into \eqref{eq:app_Jpp}--\eqref{eq:app_jpa} and then using \eqref{eq:app_CRLBs} and \eqref{eq_FIM} yields
\eqref{eq_closed_crlb_theta}--\eqref{eq_closed_crlb_phi} with
\eqref{chattheta}--\eqref{chat2_thm}. Additionally, in \eqref{chat1_thm}--\eqref{chat2_thm}, we obtain the terms $\|\dot{\mathbf{a}}_{\theta}\|^{2}$,
$\|\dot{\mathbf{a}}_{\phi}\|^{2}$, and
$\dot{\mathbf{a}}_{\theta}^{\H}\dot{\mathbf{a}}_{\phi}$ using
the property
{$\mathbf a_{\mathtt y}^{\H}\dot{\mathbf a}_{\mathtt y\theta}= \mathbf a_{\mathtt y}^{\H}\dot{\mathbf a}_{\mathtt y\phi}=  \mathbf a_{\mathtt z}^{\H}\dot{\mathbf a}_{\mathtt z\phi}= \dot{\mathbf a}_{\mathtt z\theta}= 0$}
and similar properties for $\mathbf{b}(\theta,\phi)$, along with the results
in~\cite[(88)--(93)]{nhan_isac_power_allocation}. 

\section{Proof of Corollary~\ref{rem:asymptotic_Nt}}
\label{app:asymptotic_Nt_Nr}
From~\eqref{eq_power}, under the equal-power setting, we have
$\boldsymbol{\xi}_{\mathtt{bf}}^{\T}(\boldsymbol{\vartheta})\boldsymbol{\gamma}
=\rho=\mathcal{O}(N_{\mathtt{t}}^{-1})$. We first consider
$N_{\mathtt{t}}\gg1$ with fixed $N_{\mathtt{r}}$ and fixed nonzero
$\boldsymbol{\varepsilon}$. For
$N_{\mathtt{t}\mathtt{y}}=N_{\mathtt{t}\mathtt{z}}=\sqrt{N_{\mathtt{t}}}$,
\eqref{eq:norm_adottheta}--\eqref{eq:inner_adot} yield
$\|\dot{\mathbf{a}}_{\theta}\|^{2},
\|\dot{\mathbf{a}}_{\phi}\|^{2},
\dot{\mathbf{a}}_{\theta}^{\H}\dot{\mathbf{a}}_{\phi}
=\mathcal{O}(N_{\mathtt{t}}^{2})$. 
Thus, from
\eqref{chat1_thm}--\eqref{cbeta}, we obtain
$c_{\theta},c_{\phi},c_{\theta\phi}=\mathcal{O}(N_{\mathtt{t}}^{2})$
and
$c_{\tilde{\boldsymbol{\beta}}_{\mathtt{s}}
\tilde{\boldsymbol{\beta}}_{\mathtt{s}}}
=\mathcal{O}(N_{\mathtt{t}})$.
Furthermore, by using~\eqref{g0}--\eqref{gphi}, we obtain $g_0=\mathcal{O}(1)$ and
$g_{\theta},g_{\phi}=\mathcal{O}(\sqrt{N_{\mathtt{t}}})$. Therefore,
from~\eqref{chattheta}--\eqref{chat2_thm}, we have
$\hat{c}_{\theta},\hat{c}_{\phi},\hat{c}_{\theta\phi}
=\mathcal{O}(N_{\mathtt{t}})$,
$\hat{c}_{\tilde{\boldsymbol{\beta}}_{\mathtt{s}}
\tilde{\boldsymbol{\beta}}_{\mathtt{s}}}
=\mathcal{O}(1)$, and
$\hat{c}_{\theta\tilde{\boldsymbol{\beta}}_{\mathtt{s}}},
\hat{c}_{\phi\tilde{\boldsymbol{\beta}}_{\mathtt{s}}}
=\mathcal{O}(\sqrt{N_{\mathtt{t}}})$.
Substituting these asymptotic results into
\eqref{eq_closed_crlb_theta} and~\eqref{eq_closed_crlb_phi}, the
denominator of $\mathtt{CRLB}_{\psi}^{\mathtt{bf}}$ scales as
$\mathcal{O}(N_{\mathtt{t}})$. Thus,
$\mathtt{CRLB}_{\psi}^{\mathtt{bf}}\propto 1/N_{\mathtt{t}}$ for
$N_{\mathtt{t}}\gg1$ and fixed $N_{\mathtt{r}}$.
Next, we consider $N_{\mathtt{r}}\gg1$ with fixed $N_{\mathtt{t}}$.
For
$N_{\mathtt{r}\mathtt{y}}=N_{\mathtt{r}\mathtt{z}}=\sqrt{N_{\mathtt{r}}}$,
\eqref{eq:norm_bdottheta}--\eqref{eq:inner_bdot} yield
$\|\dot{\mathbf{b}}_{\theta}\|^{2},
\|\dot{\mathbf{b}}_{\phi}\|^{2},
\dot{\mathbf{b}}_{\theta}^{\H}\dot{\mathbf{b}}_{\phi}
=\mathcal{O}(N_{\mathtt{r}}^{2})$. Consequently, from
\eqref{chat1_thm}--\eqref{cbeta} and
\eqref{chattheta}--\eqref{chat2_thm}, we obtain
$c_{\theta},c_{\phi},c_{\theta\phi},
\hat{c}_{\theta},\hat{c}_{\phi},\hat{c}_{\theta\phi}
=\mathcal{O}(N_{\mathtt{r}}^{2})$,
$c_{\tilde{\boldsymbol{\beta}}_{\mathtt{s}}
\tilde{\boldsymbol{\beta}}_{\mathtt{s}}},
\hat{c}_{\tilde{\boldsymbol{\beta}}_{\mathtt{s}}
\tilde{\boldsymbol{\beta}}_{\mathtt{s}}},
\hat{c}_{\theta\tilde{\boldsymbol{\beta}}_{\mathtt{s}}},
\hat{c}_{\phi\tilde{\boldsymbol{\beta}}_{\mathtt{s}}}
=\mathcal{O}(N_{\mathtt{r}})$. Substituting these results into
\eqref{eq_closed_crlb_theta} and~\eqref{eq_closed_crlb_phi}, yields $\mathtt{CRLB}_{\psi}^{\mathtt{bf}}\propto 1/N_{\mathtt{r}}^{2}$ for
$N_{\mathtt{r}}\gg1$.


\section{Proof of Corollary \ref{rem:asymptotic_crlb}}
\label{app:asymptotic_crlb}
We consider fixed $N_{\mathtt{t}}$ and $N_{\mathtt{r}}$, and assume
$\boldsymbol{\varepsilon}\to\mathbf{0}$.
For
$\hat{\theta}=\theta+\varepsilon_{\theta}$ and
$\hat{\phi}=\phi+\varepsilon_{\phi}$, the first-order Taylor expansions yield
$\sin\hat{\theta}\sin\hat{\phi}
=\sin\theta\sin\phi+\varepsilon_{\theta}\cos\theta\sin\phi
+\varepsilon_{\phi}\sin\theta\cos\phi
+\mathcal{O}(\|\boldsymbol{\varepsilon}\|^2)$ and
$\cos\hat{\phi}=\cos\phi-\varepsilon_{\phi}\sin\phi
+\mathcal{O}(\|\boldsymbol{\varepsilon}\|^2)$. By using these approximations in 
\eqref{deltax}, we obtain
\begin{align}
\Delta_{\mathtt{y}}
&=-(\varepsilon_{\theta}\cos\theta\sin\phi
+\varepsilon_{\phi}\sin\theta\cos\phi)
+\mathcal{O}(\|\boldsymbol{\varepsilon}\|^2), \\ 
\Delta_{\mathtt{z}}&=\varepsilon_{\phi}\sin\phi
+\mathcal{O}(\|\boldsymbol{\varepsilon}\|^2),
\end{align}
which implies that
$\Delta_{\mathtt{x}}=\mathcal{O}(\|\boldsymbol{\varepsilon}\|)$,
$\mathtt{x}\in\{\mathtt{y},\mathtt{z}\}$. Furthermore, from \eqref{Sn} and
\eqref{Qn}, for fixed $N$, we have
\begin{align}
\label{dnapprox}
 D_N(\Delta)&=N-\frac{\pi^2}{24}N(N^2-1)\Delta^2
+\mathcal{O}(\Delta^4), \\
Q_N(\Delta)&=-j\frac{\pi}{12}N(N^2-1)\Delta
+\mathcal{O}(\Delta^3).
\label{qnapprox}
\end{align}
Substituting \eqref{dnapprox}--\eqref{qnapprox} into
\eqref{g0}--\eqref{gphi} and rearranging gives
\begin{align}
  |g_0|^2\approx
N_{\mathtt{t}}^2
-N_{\mathtt{t}}\delta_{\mathbf{a}}(\boldsymbol{\varepsilon}), \ g_{\psi}=q_{\psi}(\boldsymbol{\varepsilon})
+\mathcal{O}(\|\boldsymbol{\varepsilon}\|^2),
\label{go_gpsi}
\end{align}
where
$\delta_{\mathbf{a}}(\boldsymbol{\varepsilon})
\triangleq
\|\mathbf{a}(\theta,\phi)-\mathbf{a}(\hat{\theta},\hat{\phi})\|^2$ and
$q_{\psi}(\boldsymbol{\varepsilon})
\triangleq
\dot{\mathbf{a}}_{\psi}^{\H}
(\mathbf{a}(\hat{\theta},\hat{\phi})-\mathbf{a}(\theta,\phi))$.
The $|g_\psi|^2$ contribution is largely canceled by the term involving
$\hat{c}_{\psi\tilde{\boldsymbol{\beta}}_{\mathtt{s}}}$ in
\eqref{eq_closed_crlb_theta}--\eqref{eq_closed_crlb_phi}, its residual is
of relative order $\mathcal{O}(N_{\mathtt{r}}^{-1})$ and is thus omitted.
Using \eqref{go_gpsi} in
\eqref{chattheta}--\eqref{chat2_thm}, we get
\begin{align}
   \hat{c}_{\psi}(\boldsymbol{\varepsilon})
\approx
\hat{c}_{\psi}(\mathbf{0})
-\chi_{\mathtt{s}}|\beta_{\mathtt{s}}|^2
N_{\mathtt{t}}\delta_{\mathbf{a}}(\boldsymbol{\varepsilon})
\|\dot{\mathbf{b}}_{\psi}\|^2 \\
\hat{c}_{\theta\phi}(\boldsymbol{\varepsilon})
\approx
\hat{c}_{\theta\phi}(\mathbf{0})
-\chi_{\mathtt{s}}|\beta_{\mathtt{s}}|^2
N_{\mathtt{t}}\delta_{\mathbf{a}}(\boldsymbol{\varepsilon})
\dot{\mathbf{b}}_{\theta}^{\H}\dot{\mathbf{b}}_{\phi}.
\end{align}
From \eqref{eq_closed_crlb_theta} and \eqref{eq_closed_crlb_phi}, define
$I_{\psi}(\boldsymbol{\varepsilon})
\triangleq(\mathtt{CRLB}_{\psi}^{\mathtt{bf}}(\boldsymbol{\varepsilon}))^{-1}
=T_{\psi}(\boldsymbol{\varepsilon})
-\frac{S^2(\boldsymbol{\varepsilon})}
{T_{\psi_{\mathtt{c}}}(\boldsymbol{\varepsilon})}$, where
$\psi_{\mathtt{c}}=\phi$ for $\psi=\theta$ and
$\psi_{\mathtt{c}}=\theta$ for $\psi=\phi$. Here,
$T_{\psi}(\boldsymbol{\varepsilon})$ and
$S(\boldsymbol{\varepsilon})$ follow  from
\eqref{eq_closed_crlb_theta} and \eqref{eq_closed_crlb_phi} after some algebraic manipulation, with
$r_{\mathtt{den},\psi}\triangleq T_{\psi}(\mathbf{0})$ and
$r_{\mathtt{num}}\triangleq S(\mathbf{0})$. In particular, we have
\begin{align}
  T_{\psi}(\boldsymbol{\varepsilon})
&\approx r_{\mathtt{den},\psi}
-\rho\chi_{\mathtt{s}}|\beta_{\mathtt{s}}|^2
N_{\mathtt{t}}\delta_{\mathbf{a}}(\boldsymbol{\varepsilon})
\|\dot{\mathbf{b}}_{\psi}\|^2,
\label{tepsilon}\\
S(\boldsymbol{\varepsilon})
&\approx r_{\mathtt{num}}
-\rho\chi_{\mathtt{s}}|\beta_{\mathtt{s}}|^2
N_{\mathtt{t}}\delta_{\mathbf{a}}(\boldsymbol{\varepsilon})
\dot{\mathbf{b}}_{\theta}^{\H}\dot{\mathbf{b}}_{\phi}.
\label{sepsilon}
\end{align}
Substituting \eqref{tepsilon}--\eqref{sepsilon} into $I_{\psi}(\boldsymbol{\varepsilon})$ yields
$I_{\psi}(\boldsymbol{\varepsilon})
\approx I_{\psi}(\mathbf{0})
-\rho\chi_{\mathtt{s}}|\beta_{\mathtt{s}}|^2
N_{\mathtt{t}}\delta_{\mathbf{a}}(\boldsymbol{\varepsilon})
\|\dot{\mathbf{b}}_{\psi}
-r_{\psi}\dot{\mathbf{b}}_{\psi_{\mathtt{c}}}\|^2$, where
$r_{\psi}=\frac{r_{\mathtt{num}}}{r_{\mathtt{den},{\psi}_\mathtt{c}}}$.
Finally, since $\mathtt{CRLB}_{\psi}^{\mathtt{bf}}(\boldsymbol{\varepsilon})=(I_{\psi}(\boldsymbol{\varepsilon}))^{-1}$, we use the first-order Taylor expansion $(x-y)^{-1}\approx\frac{1}{x}+\frac{y}{x^2}$ for small $y$, and obtain \eqref{eq:crlb_small_error_mismatch}. 

\section{{Proof of Theorem~\ref{thm:random_crlb}}}
\label{app:random_crlb}
For $\psi\in\{\theta, \phi\}$ and $x>0$, the CDF of $\mathtt{CRLB}_\psi^\mathtt{bf}$ is given as
\begin{align}
{F_{\mathtt{CRLB}_\psi}^\mathtt{bf}}(x)
&\triangleq \mathrm{P}\!\left(\mathtt{CRLB}_\psi^\mathtt{bf} \leq x\right)
= \mathrm{P}\!\left(\frac{1}{\mathtt{CRLB}_\psi^\mathtt{bf}} \geq \frac{1}{x}\right) \nonumber\\
&= \!\mathbb{E}_{\boldsymbol{\varepsilon}}\!\!\left\{
\mathbbm{1}\!\left\{
\tfrac{1}{\mathtt{CRLB}_\psi^\mathtt{bf}(\boldsymbol{\vartheta},
\boldsymbol{\gamma},\rho,\boldsymbol{\varepsilon})}
-\tfrac{1}{x}\geq 0
\right\}\!\right\},
\label{eq:app_CDF_indicator}
\end{align}
where $\mathbbm{1}\{\cdot\}$ denotes the indicator function, \textit{i.e.},
$\mathbbm{1}\{u\geq0\}=1$ if $u\geq0$ and $\mathbbm{1}\{u\geq0\}=0$
otherwise. For notational brevity, we use
$\mathtt{CRLB}_\psi^\mathtt{bf}(\boldsymbol{\varepsilon})
\triangleq
\mathtt{CRLB}_\psi^\mathtt{bf}(\boldsymbol{\vartheta},\boldsymbol{\gamma},\rho,
\boldsymbol{\varepsilon})$. Then, assuming that $\varepsilon_\theta$ and $\varepsilon_\phi$ are independent variables, \eqref{eq:app_CDF_indicator} can be written as
\begin{align}
F_{\mathtt{CRLB}_\psi}^\mathtt{bf}(x)
&\!=\!\!\!\int_{A_\theta}\!\int_{A_\phi}\!\!\!\!\!
  \mathbbm{1}\!\left\{
  \tfrac{1}{\mathtt{CRLB}_\psi^\mathtt{bf}(\boldsymbol{\varepsilon})}
  \geq\tfrac{1}{x}\right\}
  \!\!f_\theta(\varepsilon_\theta)f_\phi(\varepsilon_\phi)
\,\mathrm{d}\varepsilon_\theta\,\mathrm{d}\varepsilon_\phi \nonumber \\
&\hspace{-1cm}\approx\!\!\!
\int_{A_\theta}\!\int_{A_\phi}
\!\!\!\!\varsigma_\psi^\mathtt{bf}\!\left(
\tfrac{1}{\mathtt{CRLB}_\psi^\mathtt{bf}(\boldsymbol{\varepsilon})}
-\tfrac{1}{x}\right)
\!\!f_\theta(\varepsilon_\theta)f_\phi(\varepsilon_\phi)
\,\mathrm{d}\varepsilon_\theta\,\mathrm{d}\varepsilon_\phi,
\label{eq:app_sigmoid_approx}
\end{align}
where $A_\psi$ and $f_\psi$ denote the integration range and probability density function (PDF) of
$\varepsilon_\psi$, respectively. The approximation in \eqref{eq:app_sigmoid_approx} follows from
$\mathbbm{1}\{u\!\geq\!0\}\!\approx\!\varsigma(u)\!\triangleq\!(1+e^{-\varrho u})^{-1}$,
$\varrho>0$~\cite{shi_sigmoid}, and the integral is solved by Gaussian quadrature as 
\subsubsection{\textbf{Gaussian}}
For $\varepsilon_\psi\sim\mathcal{N}(0,\sigma_\psi^2)$, using the change of
variables $\varepsilon_\psi=\sqrt{2}\sigma_\psi u_\psi$, we have
$f_\psi(\varepsilon_\psi)\mathrm{d}\varepsilon_\psi
= \frac{1}{\sqrt{\pi}}e^{-u_\psi^2}\mathrm{d}u_\psi$.
Substituting this PDF into \eqref{eq:app_sigmoid_approx} yields
\begin{align}
\hspace{-0.25cm}F_{\mathtt{CRLB}_\psi}^\mathtt{bf}\!(x)\!
\approx \!\frac{1}{\pi}\!\!\int_{-\infty}^{\infty}\!\int_{-\infty}^{\infty}
\hspace{-0.4cm}\varsigma^\mathtt{bf}_\psi\!\left(\!
\tfrac{1}{\mathtt{CRLB}_\psi^\mathtt{bf}(\boldsymbol{\varepsilon}_{\mathtt{G}})}
\!-\!\tfrac{1}{x}\!\right)
\!e^{-u_\theta^{2}-u_\phi^{2}}
\mathrm{d}u_\theta\mathrm{d}u_\phi,\!
\label{eq:app_gauss_int}
\end{align}
where
$\boldsymbol{\varepsilon}_{\mathtt{G}}
=[\sqrt{2}\sigma_\theta u_\theta,
\sqrt{2}\sigma_\phi u_\phi]^{\T}$.
Applying the Gauss--Hermite (GH) quadrature rule
$\int_{-\infty}^{\infty}e^{-u^2}g(u)\,\mathrm{d}u
\approx\sum_i w_i^{\mathtt{GH}}g(z_i^{\mathtt{GH}})$~\cite{abram_gauss_quadrature}
to \eqref{eq:app_gauss_int} yields \eqref{eq:cdf_unified_1overcrlb} for Gaussian case. 
\subsubsection{\textbf{Generalized uniform}}
By using
$\varepsilon_\psi=({\pi}/U_\psi)u_\psi$, the PDF of uniformly distributed error is 
$f_\psi(\varepsilon_\psi)\mathrm{d}\varepsilon_\psi
=\frac{1}{2}\mathrm{d}u_\psi$, $u_\psi\in[-1,1]$. Substituting this PDF
into \eqref{eq:app_sigmoid_approx} gives
\begin{align}
F_{\mathtt{CRLB}_\psi}^\mathtt{bf}\!(x)\!
\approx\!\!
\tfrac{1}{4}\!\int_{-1}^{1}\!\int_{-1}^{1}\!\!
\varsigma^\mathtt{bf}_\psi\!\left(
\tfrac{1}{\mathtt{CRLB}_\psi^\mathtt{bf}(\boldsymbol{\varepsilon}_{\mathtt{U}})}
-\tfrac{1}{x}\right)
\mathrm{d}u_\theta\,\mathrm{d}u_\phi,
\label{eq:app_unif_int}
\end{align}
where
$\boldsymbol{\varepsilon}_{\mathtt{U}}
=[(\pi/U_\theta)u_\theta,(\pi/U_\phi)u_\phi]^{\T}$.
Applying the Gauss--Legendre (GL) quadrature rule
$\int_{-1}^{1}g(u)\,\mathrm{d}u
\approx\sum_i w_i^{\mathtt{GL}}g(z_i^{\mathtt{GL}})$~\cite{abram_gauss_quadrature}
to \eqref{eq:app_unif_int} yields \eqref{eq:cdf_unified_1overcrlb} for the
Uniform case.
\subsubsection{\textbf{von Mises}}
For
$f_\psi(\varepsilon_\psi)=\frac{e^{\kappa_\psi\cos(\varepsilon_\psi)}}{2\pi I_0(\kappa_\psi)}$,
$\varepsilon_\psi\in[-\pi,\pi]$, using $\varepsilon_\psi=\pi u_\psi$, we have
$f_\psi(\varepsilon_\psi)\mathrm{d}\varepsilon_\psi
= e^{\kappa_\psi\cos(\pi u_\psi)}(2I_0(\kappa_\psi))^{-1}
\mathrm{d}u_\psi$, $u_\psi\in[-1,1]$. Substituting this PDF into
\eqref{eq:app_sigmoid_approx} gives
\begin{align}
\!\!F_{\mathtt{CRLB}_\psi}^\mathtt{bf}(x)
&\approx
\frac{1}{4I_0(\kappa_\theta)I_0(\kappa_\phi)}
\int_{-1}^{1}\!\int_{-1}^{1}
\!\!\varsigma_\psi^\mathtt{bf}\!\left(
\tfrac{1}{\mathtt{CRLB}_\psi^\mathtt{bf}(\boldsymbol{\varepsilon}_{\mathtt{VM}})}
-\tfrac{1}{x}\right) \nonumber \\
& \hspace{1cm}
\times e^{\kappa_\theta\cos(\pi u_\theta)}
e^{\kappa_\phi\cos(\pi u_\phi)}
\mathrm{d}u_\theta\,\mathrm{d}u_\phi,
\label{eq:app_vm_int}
\end{align}
where
$\boldsymbol{\varepsilon}_{\mathtt{VM}}=[\pi u_\theta,\pi u_\phi]^{\T}$.
Applying GL quadrature to \eqref{eq:app_vm_int} yields \eqref{eq:cdf_unified_1overcrlb}
for the von Mises case.  

\bibliographystyle{IEEEtran}
\bibliography{IEEEabrv,Bibliography}

 \end{document}